\documentclass[11pt]{article}

\usepackage[text={6.8in,8.5in},centering]{geometry}
\usepackage{times}
\usepackage{amsfonts}
\usepackage{amsmath}
\usepackage{systeme}

\usepackage{amssymb}
\usepackage{amsthm}
\usepackage{graphicx}
\usepackage{subfig} 
\usepackage{url}
\usepackage{hyperref}

\usepackage{color}

\newcommand\mtiny[1]{\mbox{\!\tiny\ensuremath{#1}}}

\graphicspath{{Fig6/}} 

\usepackage{cancel}

\newtheorem{lemma}{Lemma}[section]
\newtheorem{theorem}[lemma]{Theorem}
 \newtheorem{definition}[lemma]{Definition}
\newtheorem{proposition}[lemma]{Proposition}
\newtheorem{corollary}[lemma]{Corollary }

\newtheorem{claim}{Claim}

\newcommand{\ms}[1]{\ensuremath{\mathsf{#1}}}
\newcommand{\bra}[1]{\ensuremath{\langle#1|}}
\newcommand{\ket}[1]{\ensuremath{|#1\rangle}}
\newcommand{\argmax}{\operatornamewithlimits{arg\ max}}

\newcommand{\nm}[1]{\ensuremath{|#1|}}

\renewcommand{\bar}{\overline}

\newcommand{\ap}{s_{\mtiny{\times}}}
\newcommand{\apm}{s_{\mtiny{\times}}^{-}}
\newcommand{\app}{s_{\mtiny{\times}}^{+}}

\newcommand{\ver}{{\ms{V}}}
\newcommand{\edge}{{\ms{E}}}

\newcommand{\union}{\cup}
\newcommand{\XX}{{\ms{XX}}}
\newcommand{\dist}{{\ms{dist}}}
\newcommand{\Jxx}{J_{\ms{xx}}}
\newcommand{\Jcr}{J^{\ms{merge}}_{\ms{xx}}}
\newcommand{\Jdr}{J^{\ms{double}}_{\ms{xx}}}
\newcommand{\Jsr}{J^{\ms{single}}_{\ms{xx}}}
\newcommand{\Jsp}{J^{\ms{split}}_{\ms{xx}}}
\newcommand{\Jtr}{J^{\ms{transition}}_{\ms{xx}}}

\newcommand{\nbr}{{\ms{nbr}}}

\newcommand{\LENS}{{\ms{LENS}}}

\newcommand{\energy}{{\mathcal{E}}}
\newcommand{\ham}{{\mathcal{H}}}

\newcommand{\wmis}{\ms{mis}}

\newcommand{\GS}{\ms{GS}}
\newcommand{\FS}{\ms{FS}}
\newcommand{\mdef}{\stackrel{\mathrm{def}}{=}}
\newcommand{\sign}{\ms{sgn}}
\newcommand{\EG}{\edge(G_{\ms{driver}})}

\newcommand{\PNS}{{\bf{PNStoq}}}
\newcommand{\ES}{{\bf{EStoq}}}
\newcommand{\St}{{\bf{Stoq}}}
\newcommand{\AC}[1]{{(#1)\!\!-\!\!{\footnotesize\textsc{Anti-crossing}}}}
\newcommand{\ACgap}{{\ms{AC\!\!-\!\!Gap}}}

\newcommand{\OL}[2]{{\langle{#1}\rangle}_{#2}}

\newcommand{\EL}{{\mathcal{E}}_{L_0}}
\newcommand{\FL}{{\mathcal{E}}_{L_1}}
\newcommand{\ER}{{\mathcal{E}}_{R_0}}
\newcommand{\FR}{{\mathcal{E}}_{R_1}}

\newcommand{\oy}{{\mathcal{Y}}}

\newcommand{\Htr}{\ham^\ms{TFIM}}
\newcommand{\HD}{\ham^\ms{XX\!-\!Sys}}
\newcommand{\DDD}{{\sc Dic-Dac-Doa{}}}
\newcommand{\gic}{{\sc GIC{}}}

\def\final{1} 
\ifnum\final=1  
\newcommand{\vnote}[1]{[{\small Vicky: \bf #1}]\marginpar{*}}
\newcommand{\sidecomment}[1]{\marginpar{\tiny #1}}
\else 
\newcommand{\vnote}[1]{}
\newcommand{\sidecomment}[1]{}
\fi  

\begin{document}
\title{Essentiality of the Non-stoquastic Hamiltonians and Driver Graph Design in Quantum Optimization Annealing}
\author{
Vicky Choi\\
Gladiolus Veritatis Consulting Co.
}

\maketitle       
\begin{abstract}
One of the distinct features of quantum mechanics is that the
probability amplitude can have both positive and negative signs, which
has no classical counterpart as the classical probability must be
positive. Consequently, one possible way to achieve quantum speedup
is to explicitly harness this feature.
Unlike a stoquastic Hamiltonian whose ground state has only positive
amplitudes~(with respect to the computational basis),
a non-stoquastic Hamiltonian can be  {\em eventually stoquastic} or
{\em properly non-stoquastic} when its ground state has both positive and negative
amplitudes.
In this paper, we describe that, for some hard instances which are
characterized by the presence of an anti-crossing (AC) in a
transverse-field
quantum annealing (QA)
algorithm,
how to design an appropriate XX-driver graph (without knowing the
prior problem structure) with an appropriate
XX-coupler strength
such that the resulting non-stoquastic QA algorithm
is  {\em proper-non-stoquastic} with two bridged anti-crossings
      (a double-AC) where the spectral gap between the first and
      second level is large enough such that the system can be operated {\em
      diabatically} in polynomial time.
The speedup is exponential in the original AC-distance, which can be sub-exponential or exponential in the system size, over
the stoquastic QA algorithm, and possibly the same order of speedup over the state-of-the-art classical
algorithms in optimization.
 This work is developed based on the novel characterizations of a
  modified and generalized parametrization definition of an
  anti-crossing in
the context of quantum optimization annealing introduced in
\cite{Choi2020}.


\end{abstract}

\section{Introduction}
Adiabatic quantum computation in the quantum annealing form is a quantum
computation model proposed for solving the NP-hard combinatorial
optimization problems, see \cite{AQC-Review} for the history survey and references
therein.
A quantum annealing  algorithm is described by 
a system Hamiltonian
\begin{align}
\label{eq:Ham1}
  \ham(s) = (1-s) \ham_{\ms{driver}} + s \ham_{\ms{problem}}
\end{align}
where the driver Hamiltonian
$\ham_{\ms{driver}}$, whose ground state is known and easy to prepare;
the problem Hamiltonian $\ham_{\ms{problem}}$, whose ground state
encodes the solution to the 
optimization problem;
$s \in [0,1]$ is a parameter that depends on
time.
A typical example of the system Hamiltonian is the transverse-field
Ising model (TFIM) where the driver Hamiltonian is
$\ham_X =-\sum_{i \in \ver(G)}
\sigma^x_i$, and the problem Hamiltonian $\ham_{\ms{problem}}$ is an
Ising Hamiltonian:
$
\ham_{\ms{Ising}} = \sum_{i \in \ver(G)} h_i \sigma^z_i + \sum_{ij \in \edge(G)} J_{ij}
\sigma^z_i \sigma^z_j
$
defined on a problem graph $G=(\ver(G),
\edge(G))$, which encodes the optimization problem to be solved.
The quantum processor that implements the
many-body system Hamiltonian is also referred to as a {\em quantum
  annealer} (QA).
The QA system is initially prepared in the known ground state of $\ham_{\ms{driver}}$, and then through a quantum evolution process, it reaches the
ground state of $\ham_{\ms{problem}}$ at the end of the evolution.
We will consider the driver Hamiltonian $\ham_{\ms{driver}}$  that includes both X-driver
term: $\ham_X =-\sum_{i \in \ver(G)}
\sigma^x_i$, and XX-driver term
$\ham_{\XX} = \Jxx \sum_{ij \in
  \edge(G_{\ms{driver}})} \sigma_i^x \sigma_j^x$ where
$G_{\ms{driver}}$ is the driver graph, and  $\Jxx$ can
be a positive or negative real number.
A Hamiltonian is originally defined to be {\em stoquastic} 
\cite{stoq-def} if its non-zero off-diagonal
matrix elements in the computational basis  are all negative; otherwise, it is called {\em
  non-stoquastic}. Thus, $\ham_{\XX}$ is stoquastic if $\Jxx < 0$ and
non-stoquastic if $\Jxx >0$. Below, we will discuss a refinement of
the non-stoquasticity.
We consider the following three types of driver
Hamiltonians:
$$
\ham_{\ms{driver}} = \left\{
\begin{array}[h]{ll}
 \ham_{X}  & (D1) \\
\ham_{X} + \ham_{\XX} & (D2)\\
\ham_{X} + s \ham_{\XX} & (D3)
\end{array}
\right.
$$
Type $(D1)$ is the standard transverse field driver Hamiltonian 
with the uniform superposition state as the (initial)
ground state.
At this point it is unclear what kind of
$(D2)$ Hamiltonian is possible to prepare experimentally, even if the (initial) ground
state is known. 
For this reason, we consider the
type $(D3)$ driver Hamiltonian, which is known as the catalyst~\cite{Non-stoquastic}, 
with the uniform superposition state as the (initial)
ground state.
In this paper, we consider the NP-hard MWIS problem (which any
Ising problem can be easily reduced to \cite{Choi2020}).
We will denote the system Hamiltonian in Eq. (\ref{eq:Ham1}) for
solving the MWIS problem on the weighted 
$G$ by $\Htr(G)$ when $\ham_{\ms{driver}}= \ham_X$; and by
$\HD(\Jxx, G_{\ms{driver}},G)$ when  $\ham_{\ms{driver}}= \ham_X
+ s \ham_{\XX}= \ham_X + s (\Jxx \sum_{ij \in
  \edge(G_{\ms{driver}})} \sigma_i^x \sigma_j^x)$, without explicitly
stating the weight function $w$ on $G$, and omitting the time parameter $(s)$,
$s \in [0,1]$.

Typically, QA is assumed to be operated adiabatically, and the system remains in its instantaneous ground state
throughout the entire evolution process.
However, QA
can be operated non-adiabatically when the system undergoes diabatic
transitions to the excited states and then
return to the ground state. The former is referred to as AQA, and the
latter as DQA. Recently, some quantum enhancement with DQA were
discussed in \cite{CL2020}.
It is worthwhile to point out that the DQA defined there means that
 the system remains in a subspace spanned by a
band of eigenstates of the system Hamiltonian
and it does not necessarily return to the ground state at the end of
the evolution. We are interested in successfully solving the optimization
problem and therefore require the system returns to the ground
state. To distinguish these two versions, 
we shall refer to our version of DQA as DQA-GS.
DQA-GS has been exploited for the possible
quantum speedup  by an oracular stoquastic QA algorithm
in\cite{diabatic1,diabatic2}. More discussion on this in Section~\ref{sec:discussion}.

The running time of a QA algorithm is the 
total evolution time, $t_f=s^{-1}(1)$. According to the Adiabatic Theorem (see, e.g.,\cite{AT-rigorous} for a rigorous
statement), the run time of an AQA algorithm
is inversely proportional to a low power of
the minimum spectral gap (min-gap), $\Delta_{10}(s^*) = \min_{s}
E_1(s) - E_0(s)$, where $E_0(s)$ ($E_1(s)$, resp.)  is the energy value of
the ground state (the first excited state, resp. ) of
$\ham(s)$. 
Thus far,  the possible quantum speedup of  the
transverse-field Ising-based QA (TFQA)  as a heuristic
solver for optimization problem over state-of-the-art classical (heuristic)
algorithms has been called into questions, see \cite{CL2020} for a
discussion. 
As a matter of fact, 
the min-gap can be exponentially small in the problem size and thus a
TFQA algorithm can take an
exponential time, without achieving a quantum advantage.
Indeed, one can easily construct instances that have an exponentially
small gap due to an anti-crossing between levels corresponding to local
and global minima of the optimization function,
see
e.g., \cite{Choi2020,Amin-Choi,diff-path}.

Anti-crossing (AC), also known as avoided level
crossing or level repulsion, is a well-known concept for physicists.
In the context of adiabatic quantum optimization (AQO), the small-gap due to an anti-crossing
has been explained in terms of some established physics theory, such as
first-order phase transition~\cite{Amin-Choi}, Anderson localization
\cite{AKR}\footnote{A correction to the paper in \cite{PNAS-correction}.}. In these two cases, the argument was based on applying the
perturbation theory at the end of evolution where the anti-crossing
occurs.
Such an anti-crossing was later referred to as a perturbative
crossing, see e.g. \cite{Non-stoquastic}.
A parametrization definition of an anti-crossing was first introduced  by Wilkinson in
\cite{Wilkinson2}, and was used in \cite{Pechukas-gas2018} to study
the effect of noise on the QA system.
In \cite{Choi2020}, we gave a parametrization definition
for an anti-crossing in
the context of quantum optimization annealing where we also describe the behavior of the energy
states that are involved in the anti-crossing, including the symmetry-and-anti-symmetry
(SAS) property of
the two states at the anti-crossing point.

In this paper, we
modify and generalize the parametrization definition of the
  anti-crossing in
\cite{Choi2020}. 
We derive some novel
characteristics of such an anti-crossing and develop it into an
analytical tool for the design and analysis of the QA
algorithm.
In particular, (1) we discover the significance of the sign of the coefficients of the
states involved in the anti-crossing which leads to the revelation of the
significant distinction between the {\em proper-non-stoquastic} and
{\em eventually-stoquastic} Hamiltonians  (to be elaborated below).
(2) We derive the necessary conditions for the formation of an
anti-crossing. This provides us algorithmic insight into the
relationship between an anti-crossing and the structure of the local
and global minima of the problem. 
(3) We
derive two different formulae for computing the AC-gap. One is descriptive in that we
can use it to study how the AC-gap changes (without actually computing the gap size) as we vary one parameter in
the system Hamiltonian, including either
the parameters in the problem Hamiltonian or the \XX-coupler strength
in the driver Hamiltonian.  The other gives AC-gap bound asymptotically
in terms of AC-distance.
The gap size is not only important for the run-time of the adiabatic
algorithm\footnote{One word of caution:  because of our two-level assumption
(that other levels are far apart) in the AC definition, the presence
of an AC implies a small gap; however, the absence of an AC does
not necessarily imply a large (polynomial in the system size) gap.
That is, AC-min-gap is necessarily small, but non-AC min-gap is not
necessarily large.}, but it is also crucial for the analysis of the diabatic transition in DQA
setting.

We now describe our algorithm, named \DDD{}, which stands for: Driver graph from
      Independent-Cliques; Double Anti-Crossing; Diababtic quantum Optimization Annealing.
The idea is to design an appropriate \XX-driver
graph $G_{\ms{driver}}$  with an appropriate
\XX-coupler strength $\Jxx$
such that  for an instance $G$ (which is
characterized by the presence of an AC  in the  TFQA algorithm
 described by  $\Htr( G)$), 
the corresponding non-stoquastic Hamiltonian system that is described
by $\HD(\Jxx, G_{\ms{driver}},G)$
has a double-AC or a sequence of nested double-ACs, or a double
multi-level anti-crossing,  with the desired property, such
that one can apply DQA-GS successfully to solve the optimization problem
in polynomial time.
A high-level description of \DDD{} is described in Table 1.
The input to the algorithm is a vertex-weighted graph $G$ of the MWIS problem with
the assumption that it has a special independent-cliques (IC)
structure. We refer to such an instance as a \gic{} ({\em graph of
  independent cliques}) instance.
Each clique  in the IC  is a clique of {\em partites},  with
each partite consisting of either one single vertex or an independent
set (of vertices). 
When all
partites are single vertices, the clique of partites is the normal
clique; otherwise the clique is also known as the multi-partite graph
(not necessarily complete). The size of the clique is the
number of partites in the clique.
The algorithm consists of two main phases. Phase 1 discovers
the IC (if presented) in the graph. The information of the IC is used to
construct the \XX-driver graph.
Phase 2 runs the QA polynomial number of times, each time in polynomial annealing
time. 
The output is the best solution found, which would be the MWIS of
$G$, if the algorithm works correctly as claimed.

\begin{table}[h]
  \centering
  \begin{tabular}[c]{c}
    \small
    \noindent\fbox{
  \parbox{6in}{
\begin{description}
  \item {Input:} A \gic{} instance -- a vertex-weighted graph $G$ of the MWIS problem with
    an (unknown)  IC structure
    \item {Step 0.} Compute the MWIS-Ising problem Hamiltonian
      $\{h_i,J_{ij}\}$ from the weighted graph $G$
\item[ ] --- Phase 1: Discover the IC ---
  \item {Step 1.1} Run $\Htr(G)$ on QA adiabatically in polynomial time
    to obtain a set of local minima $P$
  \item{Step 1.2} Find an IC formed from $P$ 
    
    DIC-- set the XX-driver graph $G_{\ms{driver}}$
    according to IC
  \item[ ] --- Phase 2: Run the QA polynomial number of times ---
    \item {Step 2.1} Estimate a range of $\Jxx$ by bounding the MWIS
      instance (potential values for forming DACs)
    \item {Step 2.2} Run $\HD(\Jxx, G_{\ms{driver}},G)$ on QA
       in polynomial annealing time for each different $\Jxx$;
      \item{Output:} The best solution found
      \end{description}
      }
}
  \end{tabular}
  \caption{A high-level description of \DDD.}
  \label{tab:DDDalg}
\end{table}

It is believed that the \gic{} instances can pose an obstacle to
classical MWIS algorithmic-solvers or
stoquastic AQA because of the IC structure which 
generates a large set  (e.g. exponentially many) of near-cost maximal
independent sets, corresponding to local minima of the optimization function.
Each such local minumum is formed from one element (either one vertex or one
partite) from each clique in the IC. For example,  if there are $k$
cliques in the IC,
each of size $t$, there will be  $t^k$ local minima.
Such a set of local minima $L$ would cause a
formation of an anti-crossing in a
transverse-field
quantum annealing 
algorithm described by $\Htr(G)$. 
We show if we take all the edges (between any two partites) within the cliques as the $\XX$-couplers
in the driver graph, and if the $\XX$-coupler
strength $\Jxx (>0)$ in $\HD(\Jxx, G_{\ms{driver}},G)$ is
appropriately large, it will force  $L$ to split into
two opposite subsets, namely
$L^+$ (states with positive amplitudes)  and $L^-$ (state with
negative amplitudes), if $L$ is to occupy the instantaneous ground state (while its
energy is minimized). This in turn
will result in two anti-crossings bridged by $(L^+,L^-)$, called a
{\em double-AC} in during the evolution process of $\HD(\Jxx, G_{\ms{driver}},G)$.
That is, if we take $G_{\ms{driver}}= G|_L$,
where $G|_{L}\mdef G[\cup_{l \in L}l]$\footnote{For a vertex set $W
  \subset \ver(G)$, its induced subgraph $G[W]$ is defined as
  $G[W]=(W, \{\{u, v\} \in \edge(G): u, v \in W\})$},
the induced
subgraph with vertices from $L$, there is an appropriate \XX-coupling strength such that the resulting
non-stoquastic QA algorithm $\HD(\Jxx, G_{\ms{driver}},G)$ is {\em proper-non-stoquastic} with a
double-AC. 
Consequently, if the second-level gap between the two anti-crossings
is large enough, the system can be operated diabatically in polynomial
time, i.e. solve the problem through DQA-GS in polynomial time. 
More specifically, the system diabatically transitions to the first excited state at the
first AC; then it adiabatically follows the first excited state and
does not transition to the second excited state because of the
large second-level gap; finally, the system returns to the ground
state through the second AC.
The above idea can be generalized to the system that has a sequence
of {\em nested double-ACs} or a double {\em multi-level
  anti-crossing} where the diabatic transitions go through a cascade
of anti-crossings similar to the {\em diabatic
  cascade}~\cite{tunneling-MAL}\footnote{However, in \cite{tunneling-MAL}, the
  diabatic cascade is made possible by first exciting the system to
  very high energy states, and then came down through a cascade of
  anti-crossings.}.
In which case,
the system diabatically  transitions to a band of excited
states through a cascade of first ACs (of the nested double-ACs), with the condition that the band of eigenstates
is well separated from the next excited state (so that the system
will not transition further),
and return to the ground state through a cascade  of second ACs
(of the nested double-ACs).
The procedure of how to identify the independent cliques efficiently
(as the driver graph) and how to identify the appropriate $\Jxx$
coupler strength is described in 
Section~\ref{sec:speedup}
when we do not have the prior knowledge of the problem structure.

The possible speedup of \DDD{} over the TFQA (by
overcoming the AC) can be exponential or super-polynomial in the system size depending on the
AC-gap size,  and possibly the same order of speedup over
state-of-the-art classical algorithms in optimization.
There are further resons to support this possible quantum speedup
because we make use of the negative amplitudes which is one of the
distinct features of quantum mechanics
that
has no classical counterpart as the classical probability must be
positive. It appears that there is no equivalent classical ways to
implement \DDD.
Furthermore, it is believed that the {\em proper-non-stoquastic}
Hamiltonians can not be efficiently simulated by classical methods
through quantum Monte Carlo (QMC) algorithms as explained below in
Section\ref{sec:PN}.

\subsection{Non-stoquasticity: Eventually Stoquastic vs Proper
  Non-stoquastic}
\label{sec:PN}
The distinction of the stoquasticity is vital in quantum
simulations because of the `sign' problem, see e.g.,\cite{vgp,curing-non-stoq}
and the references therein. In particular, in \cite{vgp}, Hen
introduced the VGP to further classify the stoquasticity of the
Hamiltonian. 
The concept of stoquasticity  
is also important from a computational
complexity-theory viewpoint, see \cite{de-sign} for the summary.

From the algorithmic design perspective, we
distinguish the stoquasticity of the Hamiltonian
based on the Perron-Frobenius (PF) Theorem\footnote{This property was
  used by others e.g. \cite{Jarret-Jordan,Jarret2} in spectral gap analysis
  but in a different way.}
which we
recast in our terminology.
\begin{theorem}(Perron-Frobenius Theorem)
If a real Hermitian Hamiltonian
$A$ is entry-wise non-negative (i.e. all entries are non-negative), then the ground state
of $-A$ is a non-negative vector (all entries are non-negative). 
\end{theorem}
We say $A$ has the PF property if the ground state
of $-A$ is a non-negative vector.
It has been shown that there are more general matrices that have the
PF property \cite{general-PF1,general-PF2}. In particular, $A$ has the
PF property if $A$ is {\em eventually} non-negative, i.e., there
exists $k_0>0$ such that  $A^k \ge 0$ (entry-wise non-negative)
for all $k >k_{0}$.
In particular, the non-stoquastic $\ham_{\XX}$ with small positive
$\Jxx$ can have the
PF property. 
\begin{definition}
A non-stoquastic $\ham$ is called {\em proper} if it fails to have the
PF property (that is, the
ground state has both positive and negative entries), otherwise is
called {\em eventually stoquastic}.
\end{definition}
We shall denote the proper non-stoquastic (eventually stoquastic,
stoquastic, resp.) by \PNS{} (\ES, \St, resp.).
An $\HD(\Jxx,G_{\ms{driver}},G)$ is \ES{}, if for all $s
\in [0,1]$, the corresponding $\ham(s)$ is \ES{}; otherwise, i.e. if
exists $s \in [0,1]$ such that the corresponding $\ham(s)$ is \PNS{},
the $\HD(\Jxx,G_{\ms{driver}},G)$ is \PNS{}.
See Figure~\ref{fig:NEStoq} for the stoquasticity of $\HD(\Jxx,G_{\ms{driver}},G)$.
\begin{figure}[h]
  \centering
  \includegraphics[width=0.7\textwidth]{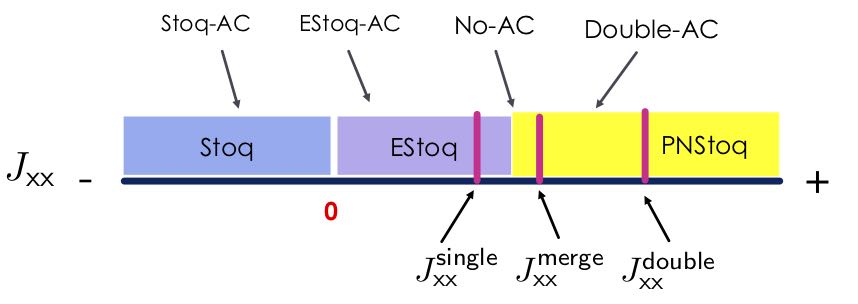}
  \caption{Stoquasticity of $\HD(\Jxx,G_{\ms{driver}},G)$ where
$\ham_{\XX}= \Jxx \sum_{ij \in
  \edge(G_{\ms{driver}})} \sigma_i^x \sigma_j^x$. $\HD(\Jxx,G_{\ms{driver}},G)$ is
\St{}(in blue) if $\Jxx<0$; it is \ES{}(in purple) if $0<\Jxx < \Jtr$; and it is \PNS{}(in yellow)
if $\Jxx \ge \Jtr$. 
For $\HD(\Jxx,G_{\ms{driver}},G)$ described in Section \ref{sec:speedup},
the system has a double-AC when $\Jxx
\in (\Jcr,\Jdr]$ where $\Jcr \ge \Jtr$. As $\Jxx$ decreases from $\Jdr$,   the bridge between the two ACs shrinks,
and the two ACs eventually merges, resulting in no AC, at
$\Jcr$. Note that no AC does not necessarily imply that the gap is not small.
}
  \label{fig:NEStoq}
\end{figure}

Recently, it has been claimed
that for ``typical'' systems, the min-gap of the non-stoquastic Hamiltonian is smaller
than its ``de-signed'' stoquastic counterpart \cite{de-sign}.
However, we show by a counter-argument (in Observation 1 of
Section~\ref{sec:speedup}) and a 
 counter-example (see 
Figures~\ref{fig:AC-strength} and \ref{fig:mg-Jxx}) that the opposite is true when the
non-stoquastic Hamiltonian is 
eventually stoquastic and an appropriate driver graph (which can be
constructed efficiently) is taken into consideration. The driver graph is either not explicitly
addressed or assumed to be the same as the problem graph in
\cite{de-sign}.

For the \PNS{},  its ground state has both positive and
negative amplitudes. It is this  proper non-stoquasticity feature, which is exclusively
quantum mechanics, that we will make use of.
Furthermore, there are reasons to believe that \PNS{} 
Hamiltonians (with \XX-driver Hamiltonians)\footnote{As Itay Hen
  pointed out that the system
  with $-\ham_{X}$ as the driver Hamiltonian is \PNS{} but it is also VGP.}
are not VGP \cite{vgp}, and thus not QMC-simulable.

\subsection{Preliminary and Notation}
We now introduce some necessary notation.
Let $\ket{E_k(s)}$ ($E_k(s)$  respectively) be the instantaneous
eigenstate (energy respectively) of the system Hamiltonian $\ham(s)$
in Eq.(\ref{eq:Ham1}) at time $s \in [0,1]$, i.e., 
$\ham(s) \ket{E_k(s)} = E_k(s) \ket{E_k(s)}$ for $k=0,1, \ldots
2^N-1,$ where $N$ is the number of qubits in the system. 
For convenience, we write $E_k \mdef E_k(1)$ and $\ket{E_k}\mdef
\ket{E_k(1)}$ for the energy and state of the problem (final)
Hamiltonian $\ham_{\ms{problem}}=\ham(1)$. For convenience, we write
$H_P=\ham_{\ms{problem}}$, and $H_D=\ham_{\ms{driver}}$, and let
$\delta H = H_P-H_D$. 

Let $2^{[N]}$ denote all possible bit
strings of length $N$ which correspond to all possible states of the
problem Hamiltonian. Each bit string corresponds to a subset of
$\{0,1,...,N-1\}$ (the position with value 1 corresponds to the
element in the subset). When there is no confusion, we use the subsets
and the corresponding bit strings interchangeably.
Thus, $2^{[N]}$ also represents all possible
subsets of  $\{0,1,...,N-1\}$, and $\{\ket{k} : k \in 2^{[N]}\}$
consists of all problem states (also known as the classical
states). Note  that the problem ground state$\ket{E_0}$ is in general not the
same as the zero state $\ket{0}$.

We express the instantaneous eigenstates ($\ket{E_0(s)},
\ket{E_1(s)}$)  in terms of the classical states:
$$
\left\{
    \begin{array}{ll}
    \ket{E_0(s)} = \sum_{k=0}^{2^N-1}
  \textcolor{red}{c_k(s)} \ket{k}, &\sum_{k=0}^{2^N-1} |c_k(s)|^2=1\\
\ket{E_1(s)} = \sum_{k=0}^{2^N-1}
  \textcolor{blue}{d_k(s)} \ket{k}, &\sum_{k=0}^{2^N-1} |d_k(s)|^2=1
    \end{array}
\right.
$$
That is, we have ${c_k(s)} = \bra{E_0(s)}k\rangle$ and ${d_k(s)} =
\bra{E_1(s)}k\rangle$ for all $k<2^N$.

In general (in the proper non-stoquastic case),
$c_k(s)$ can be positive or negative, and/or can
change the sign during the evolution course.
Since the squared
overlap $|c_k(s)|^2$ would lose the sign of $c_k(s)$, 
we introduce the signed overlap, $\sign (c_k(s))|c_k(s)|^2$.
We will
visualize our results using the signed overlaps whenever the signs are
significant. In particular, the evolution of the signed overlaps of
$c_k(s)$ and $d_k(s)$ will help us understand the formation of the
anti-crossings during the quantum evolution.

For $A \subset 2^{[N]}$ (e.g. a set of local minima states), we denote
the projection of the  eigenstate $\ket{E_i(s)}$ onto $A$ by $\ket{A_i(s)}= \widehat{P}_{A}\ket{E_i(s)}$ where
$\widehat{P}_{A} = \sum_{a \in A} \ket{a}\bra{a}$ is the projection operator, and
$\nm{A_i(s)}=||A_i(s)||$ denotes its norm. For example,
$\ket{A_0(s)} = \sum_{k \in A} c_k(s) \ket{k}$ and
$|A_0(s)| = \sum_{k \in A} |c_k(s)|^2$; similarly for
        $\ket{A_1(s)}$ and $|A_1(s)|$ (with $c_k(s)$ replaced by
        $d_k(s)$).

Let $\Delta_{ij}(s) = E_i(s)-E_j(s)$ be the instantaneous spectral gap
between the $i$th and $j$th energy levels, for $s \in [0,1]$.


\paragraph{Driver-dependent neighborhood and distance.}
As in \cite{Choi2020}, we define the driver-dependent neighbourhood $\nbr_{H_D}(\ket{k})=\{\ket{q}:
\ket{q}=\ms{Op}_i(\ket{k}),  i=0,\ldots, p \}$, where  $H_D=\sum _{i=0}^p
\ms{Op}_i$ is the driver Hamiltonian and $k,q \in 2^{[N]}$.
By this definition, 
$\nbr_{\ham_\ms{X}}$ consists of the single-bit flip neighbourhood of
the state. 
For example, 
$\nbr_{\ham_\ms{X}}(\ket{101}) = \{
\ket{10\underline{0}}, \ket{1\underline{1}1},
\ket{\underline{0}01}
\}$.
Define $\nbr_{H_D}(L) = \cup_{l \in L} \nbr_{H_D}(\ket{l})$.

Define $\dist_{H_D}(\ket{k}, \ket{q)}$ to be the 
number of $\ms{Op}$s in the minimum path 
between $\ket{k}$ and $\ket{q}$ (computational states).
For $L, R \subset 2^{N}$,
define 
$\dist_{H_D}(L,R) \mdef \min_{l \in L, r\in R}
\dist_{H_D}(\ket{l},\ket{r})$.
For example, $\dist_{H_X}(L,R)$ is the minimum Hamming distance
between sets in $L$ and $R$.

\paragraph{Sets $L$, $n(L)$, $\tilde{L}$, $\bar{L}$.}
For $L, R \subset 2^{N}$, the space $2^{[N]}$ is partitioned to $\tilde{L} \cup \tilde{R}$,
where
$\tilde{L} = L \cup n(L)$ and $\tilde{R} = R \cup n(R)$.
The set $n(L)$  ($n(R)$ resp.)
consists of states that are driver-distance closer to $L$ than $R$
($R$ than $L$ resp.).
That is, 
\begin{align*}
  \begin{cases}
    n(L) = \{ x  \in 2^{[N]}:\dist_{H_D}(\{x\},L)\le \dist_{H_D}(\{x\},R) \}\\
     n(R) = \{ x  \in 2^{[N]}:\dist_{H_D}(\{x\},L) > \dist_{H_D}(\{x\},R)
     \}\\
  \end{cases}
\end{align*}

See Figure \ref{fig:LR-partition} for an illustration.
\begin{figure}[h]
  \centering
  \includegraphics[width=0.4\textwidth]{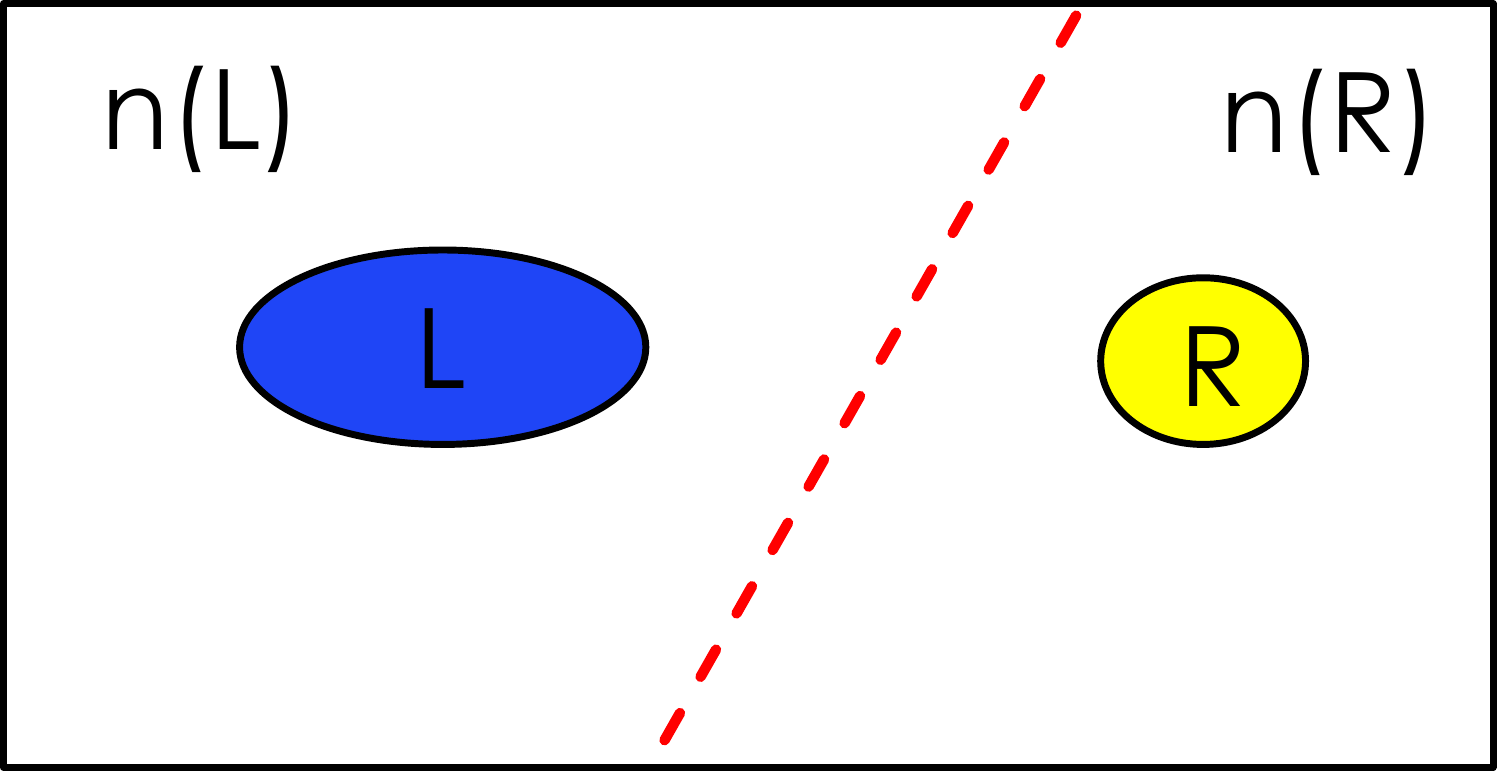} 
  \caption{The space $2^{[N]}$ is partitioned to $\tilde{L} \cup \tilde{R}$,
where
$\tilde{L} = L \cup n(L)$ and $\tilde{R} = R \cup n(R)$.}
  \label{fig:LR-partition}
\end{figure}

We also denote $\bar{L} = L \cup \nbr_{H_D}(L)$ and $\bar{R} = R
\cup \nbr_{H_D}(R)$.
The set $\nbr_{H_D}(A)$ contains the low-energy neighboring-states
(LENS) introduced in \cite{Choi2020}. In particular,  $\LENS(A) =\{k \in \nbr_{H_D}(A): \ms{energy}(k) \le
\ms{threshold}\}$, for some instance dependent $\ms{threshold}$.

The paper is organized as follows. We describe our results in Section
2. In Section 2.1, we give the definition of 
anti-crossing, and describe four characterizations of the AC, the
necessary conditions for the formation of an AC, and the two
formulae for computing the AC-gap. In Section 2.2, we describe our algorithm
\DDD{}. We conclude with our discussion in Section 3. The proofs are
included in Section 4.

\section{Results}

\subsection{Anti-crossing: A Tool for Design and Analysis of QA algorithm}

\subsubsection{New Defintion of Anti-Crossing}
Informally, in the context of the quantum annealing,
we define an anti-crossing between two consecutive levels with  the following four conditions:
(1) the
anti-crossing point corresponds
to a local minimum in the energy spectrum between the two interacting
levels; 
(2) within the anti-crossing interval, all other energy levels are
far away from the two interacting levels;
(3) there are two small sets of ``non-negligible'' states involved in the anti-crossing;
(4) a ``sharp  exchange'' of the non-negligible states within a small
width of the anti-crossing
interval.

More formally, a new parametrization definition of an anti-crossing
between the lowest two levels (which can be straightforwardly
generalized to any two consecutive levels) during the evolution of $\ham(s)$ in Eq. (1) is defined as
follows.
\begin{definition} 
We say that there is an $\AC{L,R}$ at $\ap$
 if  there exists two disjoint subsets $L$ and $R$  ($\subset 2^{[N]}$), 
 and  $\delta>0$, $\gamma >0$ such that
 \begin{description}
   \item[(i)] For $s \in [\ap-\delta, \ap+\delta]$,
$\Delta_{10}(\ap) \le \Delta_{10}(s)$;
\item [(ii)]For $s \in [\ap-\delta, \ap+\delta]$,
$\Delta_{10}(s) << \Delta_{k0}(s)$, for all $k>1$;

\item[(iii)] Within the time interval $[\ap-\delta, \ap+\delta]$, both
$\ket{E_0(s)}$ and $\ket{E_1(s)}$
are mainly composed of states from $L$ and $R$.
More precisely, 
let $\widetilde{L} = L \cup n(L)$, $\widetilde{R} = R \cup n(R)$,
such that all the possible states $2^{[N]}= \widetilde{L} \cup
\widetilde{R}$, and $\widetilde{L} \cap
\widetilde{R} = \emptyset$.
For $s \in [\ap-\delta, \ap+\delta]$,
\begin{align}
  \label{eq:12}
    \ket{E_i(s)} = \ket{\widetilde{L_i}(s)} + \ket{\widetilde{R_i}(s)}
  \simeq \ket{L_i(s)} +  \ket{R_i(s)}
\end{align}
        where $|L_i(s)| + |R_i(s)| \in [1-\gamma, 1]$ and
        $|n(L_i)(s)| +
        |n(R_i)(s)| \in [0, \gamma]$, $i=0,1$.
        (Recall: $\ket{A_0(s)} = \sum_{k \in A} c_k(s) \ket{k}$ and
        $\ket{A_1(s)} = \sum_{k \in A} d_k(s) \ket{k}$.)
        
      \item[(iv)] Before the anti-crossing at  $\ap^{-} \equiv \ap -\delta$,
        $\ket{E_0(\apm)} \simeq \ket{L_0(\apm)}$ (or $|{L_0(\apm)}|
        \ge (1-\gamma)$), $\ket{E_1(\apm)} \simeq
        \ket{R_1(\apm)}$; after the anti-crossing  at $\ap^{+} \equiv \ap + \delta$,
        $\ket{E_0(\app)} \simeq \ket{R_0(\app)}$, $\ket{E_1(\app)} \simeq
        \ket{L_1(\app)}$.

      \end{description}
\end{definition}
      For convenience, we shall refer $(L,R)$ as the two {\em arms} of the
anti-crossing (they appear in the left and right of each energy
level).  We shall refer $\delta$ as the AC-width, and
$\dist_{H_D}(L,R)$ as the AC-distance.
One can take the AC-width to be  the $\delta$ such that the minimum of the four ``corner
overlaps'' is maximized, i.e.  $\delta = \max_{\xi>0} \min\{ \langle
E_0(\ap - \xi) \ket{L_0(\ap-\xi)}, \langle
E_1(\ap - \xi) \ket{R_1(\ap-\xi)}, \langle
E_0(\ap + \xi) \ket{R_0(\ap+\xi)}, \langle
E_1(\ap + \xi) \ket{L_1(\ap+\xi)}
\}$. In general, we require $\gamma$ to be near
$1$. Sometimes it is possible to increase $\gamma$ by enlarging $L$
and/or $R$ such that $\AC{L^*,R^*}$ is an anti-crossing with a larger
$\gamma$ (for the same $\delta$), where $L \subseteq L^*, R \subseteq
R^*$.
A figure depicting an $\AC{L,R}$ is shown in
Figure \ref{fig:anti-crossing}.

\begin{figure}[h]
  \centering
  \includegraphics[width=0.5\textwidth]{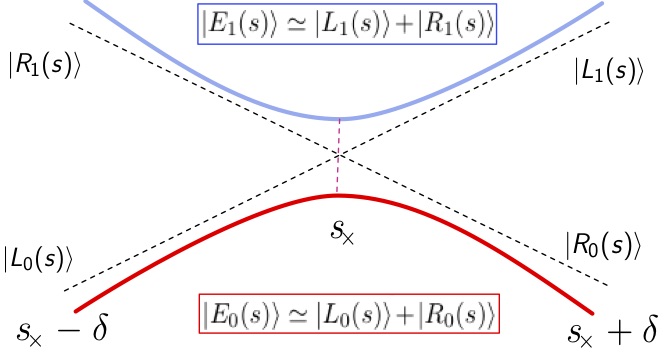} \\
(a)\\
\includegraphics[width=0.4\textwidth]{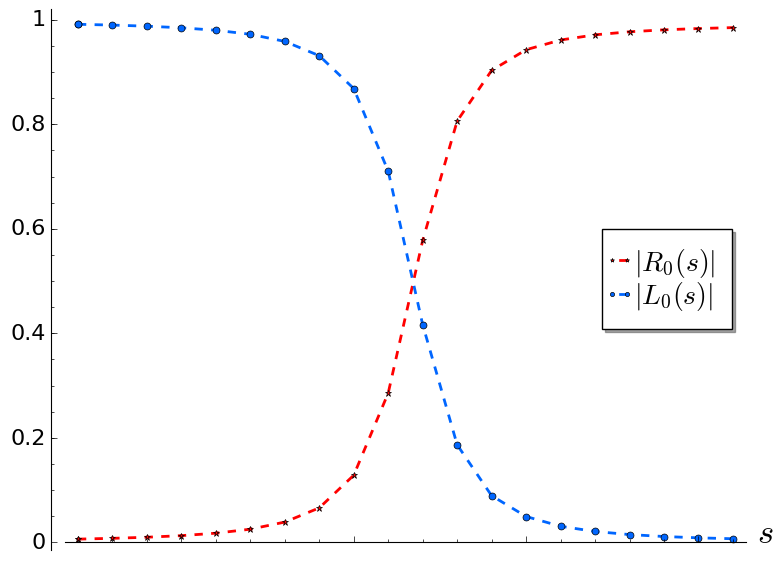} 
\includegraphics[width=0.4\textwidth]{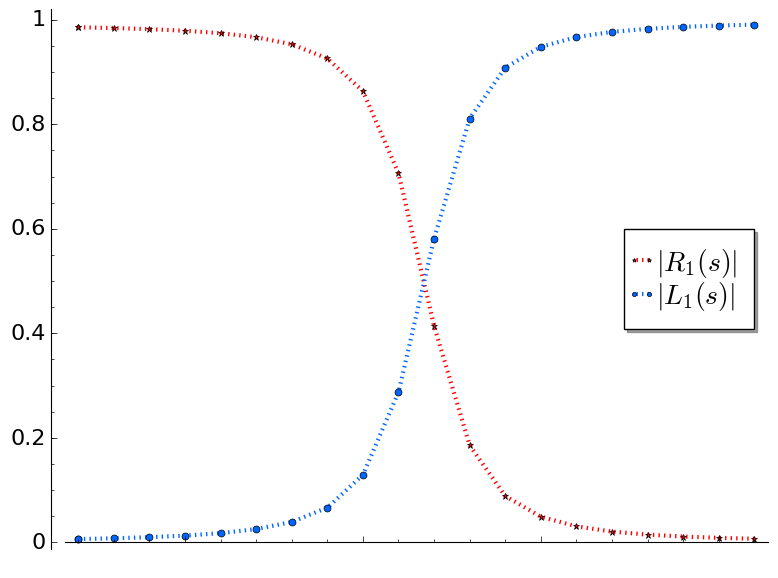} \\
(b)
  \caption{An {\em $\AC{L,R}$} at $\ap$. (a) Within the time interval $[\ap-\delta, \ap+\delta]$, both
$\ket{E_0(s)}$ and $\ket{E_1(s)}$
are mainly composed of states from $L$ and $R$. That is, $\ket{E_0(s)}
\simeq \ket{L_0(s)} +  \ket{R_0(s)}$ and $\ket{E_1(s)} \simeq \ket{L_1(s)} + \ket{R_1(s)}$
where 
$\ket{A_i(s)} = \sum_{k \in A} f^i_k(s) \ket{k}$,  with $f^0_k =c_k,
f^1_k=d_k$, $A \in \{L,R\}$, $i=0,1$.
(b) When $s$ goes from $\ap -\delta$ to $\ap+\delta$ , the state $\ket{E_0(s)}$ shifts
  from $L$ to $R$ while $\ket{E_1(s)}$ shifts from $R$ to $L$:
$|L_0(s)| \downarrow$, $|R_0(s)| \uparrow$ while $|R_1(s)|
  \downarrow$, $|L_1(s)| \uparrow$, where
$|A_i(s)| = \sum_{k \in A} |f^i_k(s)|^2$ is the total magnitude in
$A_i(s)$, for  $A \in \{L,R\}$, $i=0,1$.
  The figures show the evolution of $|L_0(s)|$ and
$|R_0(s)|$ (left)  and $|L_1(s)|$ and
$|R_1(s)|$ (right) within the interval $[\ap-\delta, \ap+\delta]$.
}
  \label{fig:anti-crossing}
\end{figure}

By definition, there is an exchange between the same arm of the two
levels.
For $\epsilon_v>0$, we say the $\AC{L,R}$ is a
$(1-\epsilon_v)$-full-exchange anti-crossing if
\begin{align}
  \label{eq:13}
  \begin{cases}
    |c_k (\apm)| \doteq_{\epsilon_v}  |d_k (\app)| \mbox{ for all } k \in
    \tilde{L} \\
     |c_k (\app)| \doteq_{\epsilon_v}  |d_k (\apm)| \mbox{ for all } k \in
    \tilde{R} \\
  \end{cases}
\end{align}
such that 
        \begin{align}
          \label{eq:full-exchange}
        \begin{cases}
         |\bra{L_0(\apm)}L_1(\app)\rangle| \approx |L_0(\apm)| \approx |L_1(\app)|\\
           |\bra{R_1(\apm)}R_0(\app)\rangle| \approx|R_1(\apm)| \approx |R_0(\app)|
        \end{cases}
      \end{align}
where $A \doteq_{\epsilon_v} B$ means $|A-B| \le \epsilon_v$,      
      This condition in Eq. (\ref{eq:full-exchange}) is required to
      derive both 
      the SAS
      property in (C4) in Section \ref{sec:char} and the necessary conditions (Eq.~(\ref{eq:nece}) in Section \ref{sec:nece}.
      
In this paper, we fix a small $\epsilon_v>0$, and assume an
$\AC{L,R}$ satisfies this condition when we do not mention
explicitly.

The earlier definition in \cite{Choi2020} is a special case with
  $L=\{\ket{E_1}\}$ (corresponds to $\ket{\FS}\mdef \ket{E_1}$), and $R=\{\ket{E_0}\}$ (corresponds
  to $\ket{\GS}\mdef \ket{E_0}$), without explicitly requiring the SAS
  property within $\epsilon$ at $\ap$ in the definition. 
We re-introduce the ``negligible'' sets $n(L)$  and $n(R)$ back to the
definition for the purpose of computing the \ACgap{}. 
The coefficients of states
        in $n(L), n(R)$ are ``negligible'' but not
        ``vanishing''. Indeed, we shall show that the
        ``negligible'' states are what contribute to
      the  gap size. (Similar to the idea that it is the high-order correction
      terms in the perturbation formula that contribute to the
      perturbative-crossing gap size.)

\subsubsection{New Characterizations of the Anti-crossing}
      \label{sec:char}
We derive some useful characteristics of the AC based on the
perturbation theory. 
The derivation is based on the idea
that  in the
neighborhood of $\ap$ the behavior of these two energy levels can be
approximately computed by non-degenerate perturbation theory for the
two lowest
energy levels while the other energy levels are far away enough to be
neglected.

\begin{proposition}
  \label{prop1}
 An $\AC{L,R}$ has the following four properties.
  \begin{description}
\item[(C1)] (''Hyperbolic-like curves'')  The two energy levels form two opposite parabolas around $\ap$: \footnote{Throughout this paper, for convenience, we will use the notation $\doteq$
such that $A \doteq B$ means $|A - B| \le \epsilon$ for some small error tolerance
$\epsilon > 0$.}
\begin{align}
  \label{eq:9}
\begin{cases}
  E_1(\ap+\lambda) \doteq \alpha \lambda^2 +\beta_1 \lambda + E_1(\ap)\\
          E_0(\ap+\lambda) \doteq  -\alpha \lambda^2 +\beta_0 \lambda + E_0(\ap)
\end{cases}
\end{align}
for $\lambda \in [-\delta,\delta]$, where $\alpha>0$,
$\alpha =  \frac{|\bra{E_1(\ap)}\delta H\ket{E_0(\ap)}|^2}{\Delta_{10}(\ap)}>0$
and $\beta_1 \simeq \beta_0$.
Thus, we have $E_1(\ap+\lambda) - E_0(\ap+\lambda) \doteq 2\alpha \lambda^2 +
\Delta_{10}(\ap)$.
The gap at the AC point is a local minimum of the gap
spectrum.

\item[(C2)] (''States-exchange'') When $s$ goes from $\ap -\delta$ to $\ap+\delta$ , the state $\ket{E_0(s)}$ shifts
  from $L$ to $R$ while $\ket{E_1(s)}$ shifts from $R$ to $L$:
$|c_l(s)| \downarrow$ for all $l \in L$, $|c_r(s)| \uparrow$ for all
$r \in R$; while $|d_r(s)| \downarrow$ for all $r \in R$, $|d_l(s)| \uparrow$ for all
$l \in L$, where $\downarrow$ denotes decreasing and $\uparrow$
denotes increasing.
In particular, for $s :\ap -\delta \leadsto \ap+\delta$,
$|L_0(s)| \downarrow$, $|R_0(s)| \uparrow$ while $|R_1(s)|
  \downarrow$, $|L_1(s)| \uparrow$. 

\item[(C3)] (``Opposite signs'') One arm is of the same sign in the two states; the other arm is
  of the opposite sign in the two states. That is, it is either
  \begin{align}
    \label{eq:8}
    \begin{cases}
       L_0 \mbox{ and } L_1 \mbox{ in the same sign: } \sign(c_l(s))\sign( d_l(s)) =+1 \mbox{ for all } l \in L\\
      R_0\mbox{ and } R_1 \mbox{ in the opposite sign: } \sign(c_r(s))\sign( d_r(s)) =-1 \mbox{ for all } r \in R 
    \end{cases}
  \end{align}
  or the reverse.

\item[(C4)] (``Symmetry-and-anti-symmetry(SAS) at $\ap$'')
  \begin{align}
    \begin{cases}
      c_l(\ap) \doteq d_l(\ap) \mbox{ for } l \in \widetilde{L}\\
         c_r(\ap) \doteq - d_r(\ap) \mbox{ for } r \in \widetilde{R}
       \end{cases}
    or 
    \begin{cases}
      c_l(\ap) \doteq -d_l(\ap) \mbox{ for } l \in \widetilde{L}\\
         c_r(\ap) \doteq d_r(\ap) \mbox{ for } r \in \widetilde{R}
    \end{cases}
  \end{align}

Thus, we have
$\ket{\widetilde{L}(\ap)} \doteq \ket{\widetilde{L_0}(\ap)} \doteq
(+/-) \ket{\widetilde{L_1}(\ap)}$ and $\ket{\widetilde{R}(\ap)} \doteq
\ket{\widetilde{R_0}(\ap)} \doteq (-/+) \ket{\widetilde{R_1}(\ap)}$,
i.e.
\begin{align*}
  (*)
  \begin{cases}
    \ket{E_0(\ap)} \doteq \ket{\widetilde{L}(\ap)} \textcolor{red}{+} \ket{\widetilde{R}(\ap)}\\
           \ket{E_1(\ap)} \doteq \ket{\widetilde{L}(\ap)} \textcolor{red}{-} \ket{\widetilde{R}(\ap)}
         \end{cases}
  \mbox{ or } (**)
  \begin{cases}
    \ket{E_0(\ap)} \doteq \ket{\widetilde{L}(\ap)} +\ket{\widetilde{R}(\ap)}\\
           \ket{E_1(\ap)} \doteq \textcolor{red}{-} \ket{\widetilde{L}(\ap)}+ \ket{\widetilde{R}(\ap)}
  \end{cases}
\end{align*}
Furthermore, $|A_i(\ap)| \in [1/2-\gamma/2, 1/2+\gamma/2]$, for $A\in
\{L,R\}$, $i=0,1$.
  \end{description}
\end{proposition}

\paragraph{AC-signature.} We refer to the two reverse curves in the Property (C2) as the
``AC-signature''  to describe the
``sharp exchange'' between the ground state and the first excited
state at the anti-crossing interval. This is more general than the Hamming weight
operator $<\!HW\!>$  (when the
problem ground state is assumed to be
the zero vector)  introduced in \cite{tunneling-MAL} that describes
the changes of the ground state wavefunction (without describing the reverse changes of
the first excited state wavefunction).

\paragraph{Significance of the Sign.} It is worthwhile to emphasize
the signs of the coefficients of the
two arms in Property (C3) are critical. There are two possible cases and are depicted
in Figure~\ref{fig:arms-same-different}. The two anti-crossings can be
connected/bridged through a common arm which is of the opposite sign.
That is, 
the bridge consists of both positive and negative amplitudes
(requiring the 
proper-non-stoquastic condition). 
Furthermore, since one arm of the AC must be in the
  opposite sign, if the ground state only allows the positive sign (as
   in the stoquastic case), the possible combinations (+/+, -/+)
  will be less in the stoquastic case than the possible combinations
  (+/+, -/-, -/+, +/-) in the non-stoquastic cases. This partially explains why one
  would observe more ACs in the non-stoquastic case.

\begin{figure}[t]
  \centering
$$
  \begin{array}[h]{cc}
    \includegraphics[width=0.32\textwidth]{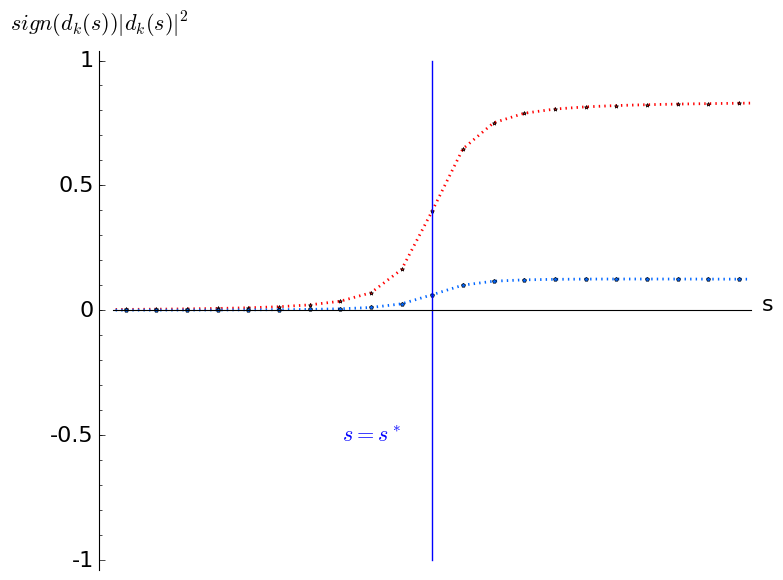} &  
\includegraphics[width=0.32\textwidth]{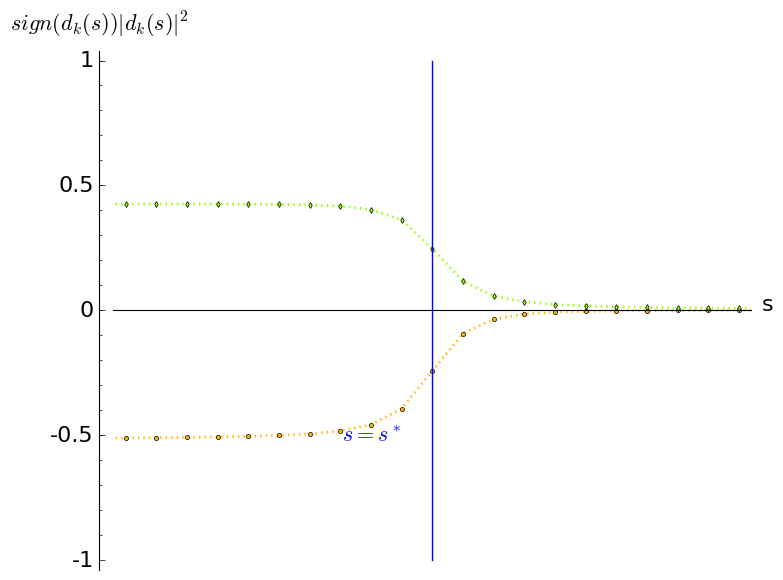} \\
    (a) L_1:|d_l(s)| \uparrow & (c) R_1: |d_r(s)| \downarrow\\
         \includegraphics[width=0.32\textwidth]{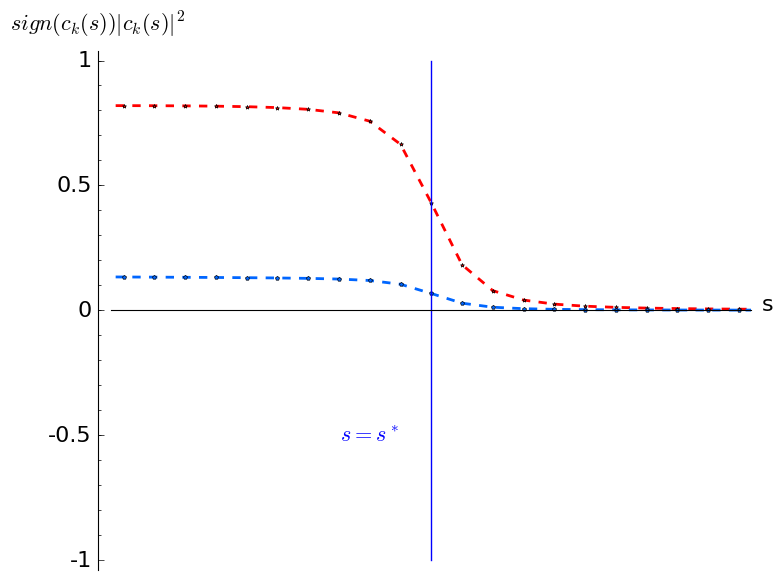} &  
\includegraphics[width=0.32\textwidth]{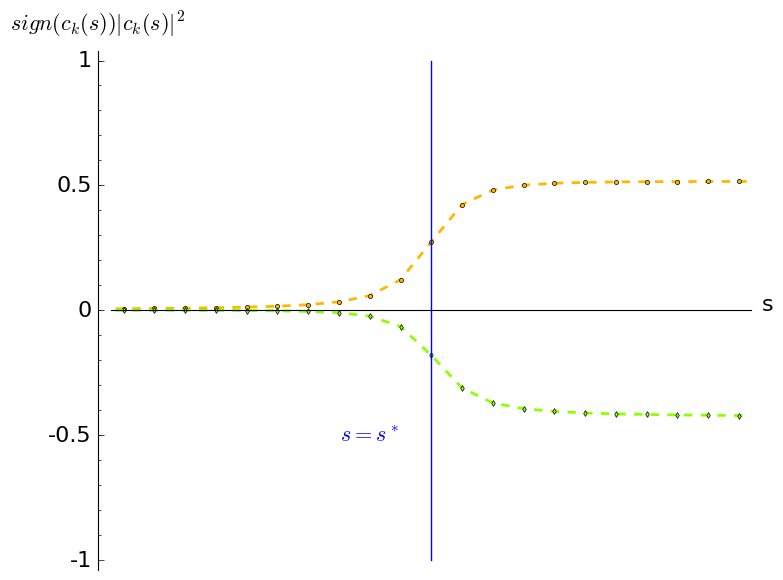} \\
    (b) L_0:|c_l(s)| \downarrow & (d) R_0:|c_r(s)| \uparrow\\
 \hline
  \end{array}
  $$
   $$
  \begin{array}[h]{cc}
    \includegraphics[width=0.32\textwidth]{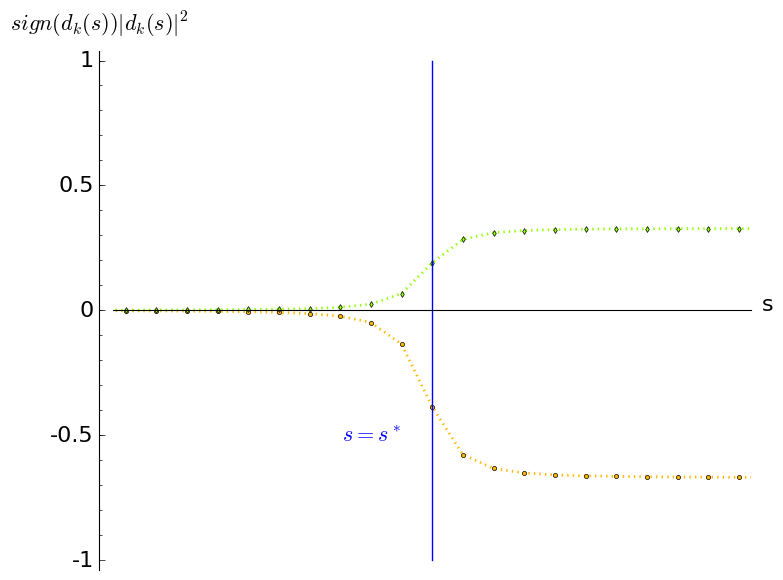} &  
\includegraphics[width=0.32\textwidth]{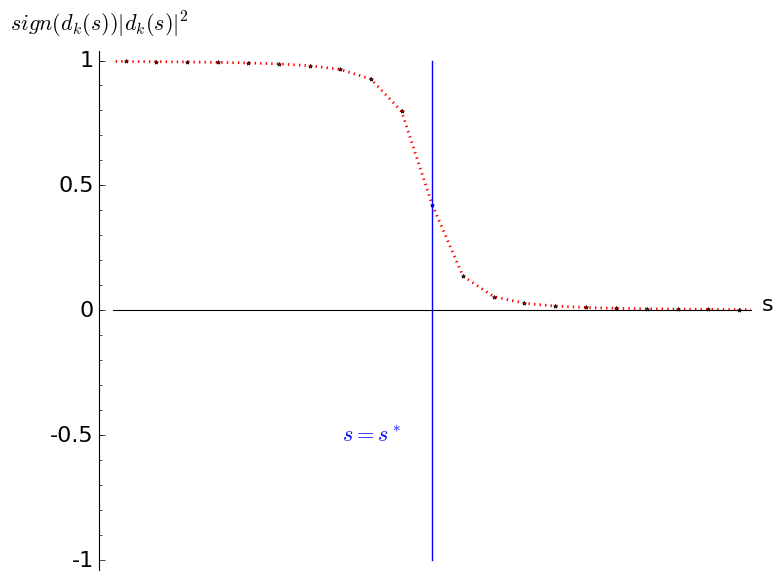} \\
(e) L'_1: |d_l(s)| \uparrow & (g) R'_1: |d_r(s)| \downarrow \\
 \includegraphics[width=0.32\textwidth]{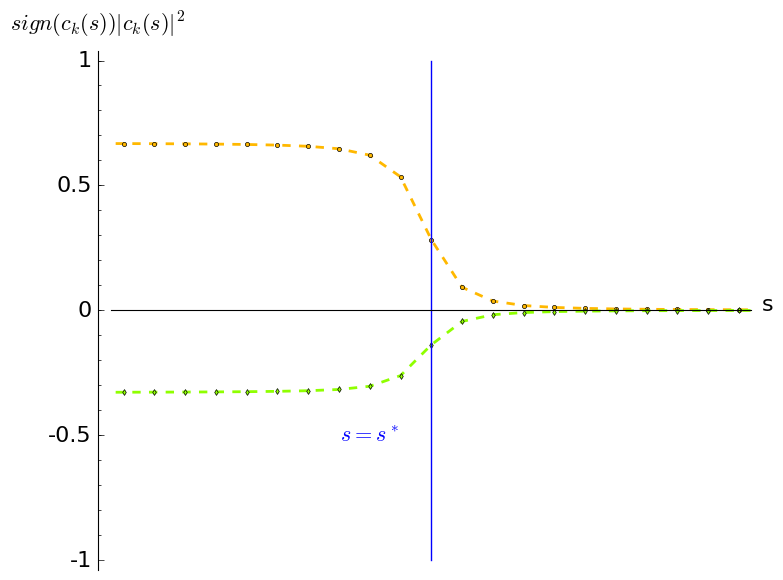} &  
\includegraphics[width=0.32\textwidth]{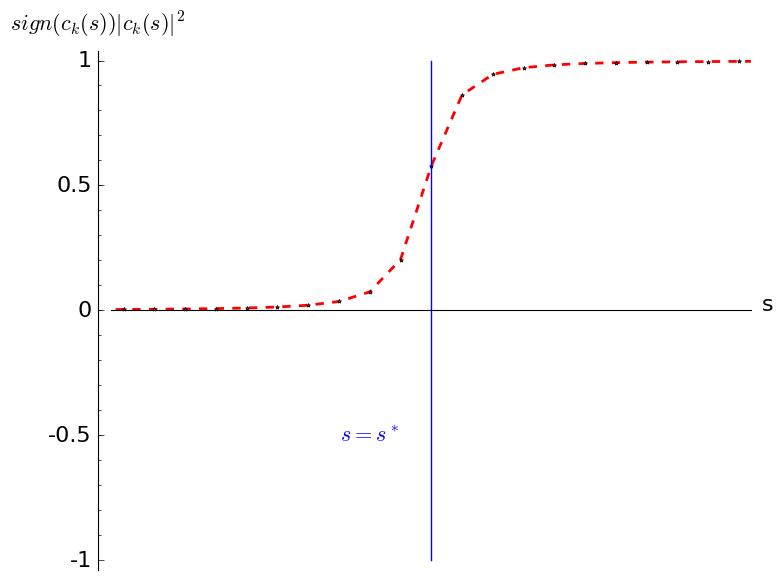} \\
(f) L'_0: |c_l(s)| \downarrow & (h) R'_0: |c_r(s)| \uparrow \\
\hline
  \end{array}
$$
  \caption{(Above)$L_1$ in (a)  and $L_0$ in (b) in the same sign: $\sign(c_l)\sign(d_l)=+1
    \forall l \in L$; $R_1$ in (c)  and $R_0$ in (d)  in the different sign: $\sign(c_r)\sign(d_r)=-1 \forall
    r \in R$.
    (Below) $L'_1$ in (e)  and $L'_0$ in (f) in the different sign; $R'_1$
    in (g)  and $R'_0$ in (h)  in the same sign.
    These are examples from a bridged double-AC: $\AC{L,R}$ and
    $\AC{L',R'}$ where $R=L'$ is the common arm of the opposite
    sign which is a necessary condition for a bridged double-AC to be formed. 
}
  \label{fig:arms-same-different}
\end{figure}

\subsubsection{Necessary Conditions for the Formation of the
  Anti-crossing}
\label{sec:nece}
In this section we derive the necessary conditions for the
formation of an $\AC{L,R}$.

\begin{theorem}
The necessary conditions for an $\AC{L,R} $ at $\ap$ with width $\delta$ are:
 \begin{align}
   \label{eq:nece}
   \begin{cases}
  \OL{\delta H}{\ket{E_0(\apm)}} -\OL{\delta H}
  {\ket{E_1(\apm)}}  \approx  2 \eta^2\frac{\delta }{\Delta_{01}(\ap)} \\
  \OL{\delta H}{\ket{E_1(\app)}} -\OL{\delta H}
  {\ket{E_0(\app)}} \approx 2 \eta^2\frac{\delta }{\Delta_{01}(\ap)} 
   \end{cases}
 \end{align}
 where $\eta =\bra{E_1(\ap)}\delta H\ket{E_0(\ap)}$
 and 
 $\OL{\delta H}{\ket{\psi}} \mdef
\bra{\psi}\delta H\ket{\psi}$. (Recall: $\delta H = H_P - H_D$.)

In particular, the necessary conditions described by $L,R$ are approximately:
 \begin{align}
   \label{eq:17}
   \begin{cases}
     \OL{\delta H}{L_0,\apm}  > \OL{\delta H}{R_1,\apm} \\
     \OL{\delta H}{L_1,\app}  > \OL{\delta H}{R_0,\app} 
   \end{cases}
 \end{align}
where $\OL{\delta H}{\ket{E_i(s^*)}}$ is approximated by $\OL{\delta
  H}{A_i,s^*} \mdef \bra{A_i(s^*)} \delta\ham\ket{A_i(s^*)}$ when
$\ket{E_i(s^*)} \simeq \ket{A_i(s^*)}$, for $i\in \{0,1\}$, $A \in \{L,R\},
s^* \in \{\apm, \app\}$.
\label{thm:nece}
\end{theorem}

Since $\delta H= H_P - H_D$, when $\OL{H_P}{L_i,s^*}  \approx
\OL{H_P}{R_{\bar{i}},s^*}$, for $i \in \{0,1\}$, $s^* \in \{\apm,
\app\}$,  the necessary conditions in Eq.~(\ref{eq:17}) become $\OL{-H_D}{L_i,s^*}  >
\OL{-H_D}{R_{\bar{i}},s^*}$. When there is no confusion, we drop
the subscripts $i$ and $\bar{i}$  of $L, R$.  When $H_D=H_X$, $\OL{-H_X}{A,s^*} \approx
\bra{\bar{A}(s^*)} \sum_{i} \sigma_i^x \ket{\bar{A}(s^*)}$ is approximated
by the overlap between $A$ and $\LENS(A)$. Thus the condition that   $\OL{-H_X}{L,s^*}  >
\OL{-H_X}{R,s^*}$ would mean that $L$ has {\em more} LENS than $R$ as
observed in \cite{Choi2020}. Furthermore, we show that with some extra
conditions, the necessary conditions are also sufficient. We also
compare the necessary conditions with the arguments of the
anti-crossing presented in the AQA algorithm for random Exact-Cover 3 instances by Altshuler~et~al in \cite{AKR}.
In our subsequent work, we will further elaborate these results.

\begin{corollary}
The  \ACgap{} of an $\AC{L,R} $ at $\ap$ with width $\delta$ is
given by
  \begin{align}
    \label{eq:coro}
    \Delta_{10}(\ap) \approx
    \begin{cases}
        2 \eta^2\frac{\delta }{\OL{\delta H}{\ket{E_0(\apm)}} -\OL{\delta H}
      {\ket{E_1(\apm)}}}\\
  2 \eta^2\frac{\delta }{\OL{\delta H}{\ket{E_1(\app)}} -\OL{\delta H}
  {\ket{E_0(\app)}}}
    \end{cases}
  \end{align}
  \label{coro}
\end{corollary}
The above corollary follows directly from Eq.(\ref{eq:nece}).
The equation of \ACgap{} in Eq.~(\ref{eq:coro}) provides a description
how the \ACgap{} depends on the two factors: the numerator $\delta$
(anti-crossing width) and the difference between $L$ and $R$ in the
denominator. This will give us as a tool to 
study the effect on the gap size
  (without actually computing the gap size)
  by analyzing how the parameters of the anti-crossing  evolve
  as we vary one parameter in
the system Hamiltonian, including either
the parameters in the problem Hamiltonian or the \XX-coupler strength
in the driver Hamiltonian.
In particular, we use Corollary \ref{coro} to justify the Observation 1 in Section~\ref{sec:speedup}.

\subsubsection{\ACgap{} Bound}

In this section we derive  the analytical \ACgap{}  bound for the
anti-crossing that satisfies the SAS properties (for a  small
$\epsilon_v$) and also large $\gamma$ (such that the corresponding
$L,R$ are taking as large as possible).


\begin{theorem}
  \label{thm-gap}
  Suppose that $R$ consists of the problem ground state $\ket{\GS}$
  and $L$ consists of
  the lowest few almost degenerate  excited states. Then  \ACgap{} 
  $\Delta_{10}(\ap) = \Theta(\zeta^{\dist_{H_D}(L,R)})$ where
  $0<\zeta<1$ and $\dist_{H_D}(L,R)$ is the driver-distance between
  $L$ and $R$ (referred to as the AC-distance).
\end{theorem}

Remark.
The $\ACgap{}$ is necessarily exponentially small (in the AC-distance),
  however, it is not necessarily true that every exponentially
small (even in the problem size) gap  corresponds to an anti-crossing defined here. For example,
it is not clear if the exponential small gap example presented in
\cite{Jarret-Jordan} corresponds to our anti-crossing because they
only show that there exists an excited state whose energy is close to
the ground energy level and thus condition (ii) in our
definition is not necessarily satisfied.
A possible future work is to generalize the AC definition between one
energy level and one narrow band of closely together energy
levels.

\subsection{Quantum Speedup by {\sc Dic-Dac-Doa}}
\label{sec:speedup}
\paragraph{MWIS Problem.}
We use the NP-hard maximum-weighted independent set  (MWIS) as our
model problem.
See, e.g. \cite{MWIS2019} for a recent classical algorithm for MWIS.
We make use of the structure of {\em maximal independent sets} to construct
the driver graph.
As described in \cite{Choi2020}, MWIS and Ising problem can be
efficiently reduced to each other. The MWIS-Ising Hamiltonian
is specified by a problem graph $G$ and a weight vector $w$ on its
vertices and a $J_{\ms{penalty}}$ (In our examples in this paper, we
  assume $J_{\ms{penalty}}=4$ when we omit to mention). The formulae
  for computing the corresponding $\{h, J\}$ are
    also described in the Appendix.

\subsubsection{An illustrative Example}

An MWIS graph with 9 weighted vertices is depicted in
Figure~\ref{fig:G1}.
The global minimum corresponds to the maximum independent set $\{2,5,8\}$
with total weight of $4.02$.
$G$ has 8 local minima, $\{ \{v_1,v_2,v_3\}: v_1 \in \{0,1\}, v_2 \in
\{3,4\}, v_3 \in  \{5,6\} \}$, corresponding to the 8 maximal
independent sets,
with total weights ranging from
$3.70$ to $3.95$.
This graph can be scaled to a graph of $3n$ vertices, where the
global minimum consists of $n$ vertices, while there are $2^n$
local minima formed by $n$ independent-cliques.
\begin{figure}[h]
  \centering
  $$
  \begin{array}[h]{cc}
    \includegraphics[width=0.45\textwidth]{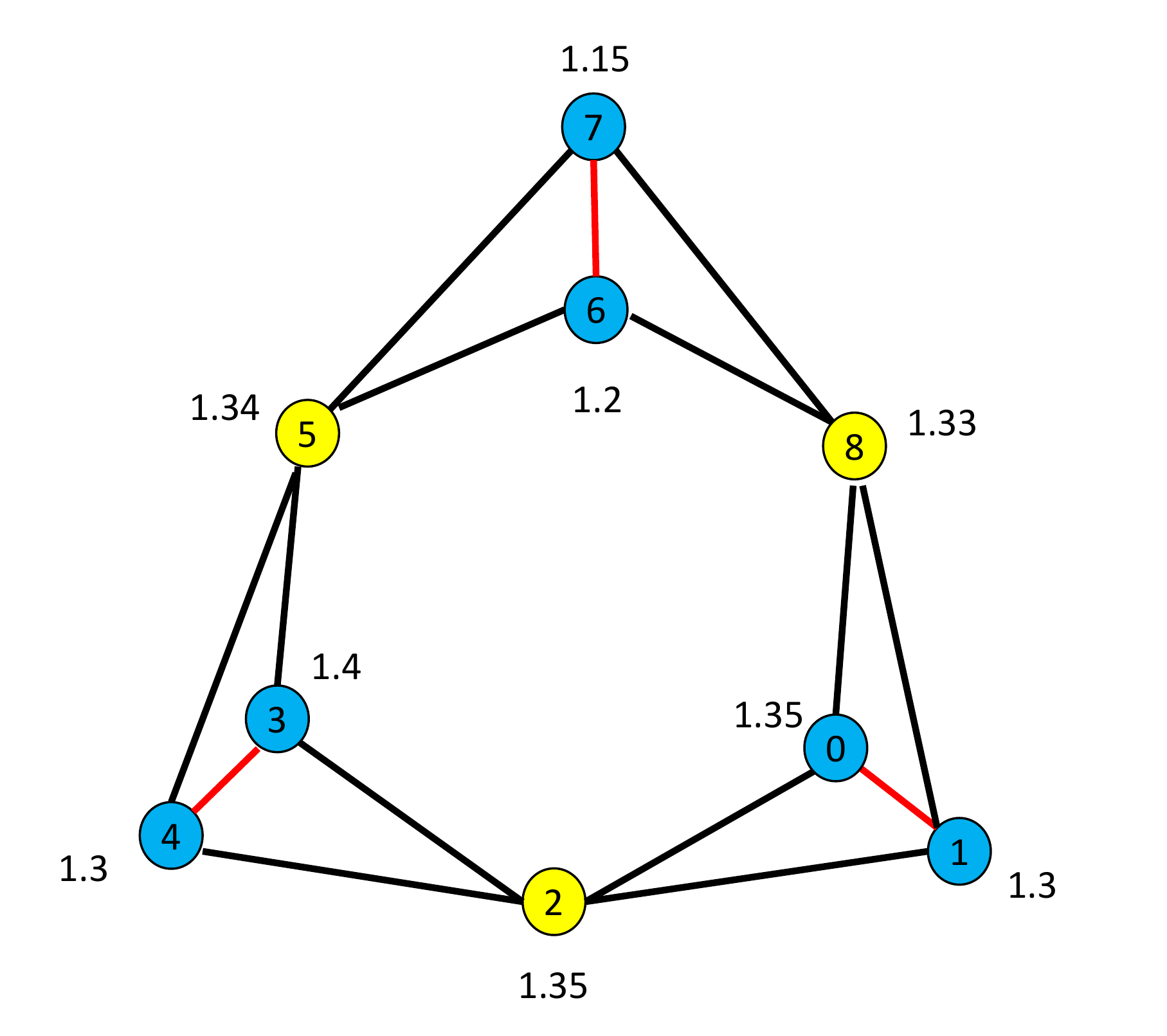}&
                                                          \includegraphics[width=0.45\textwidth]{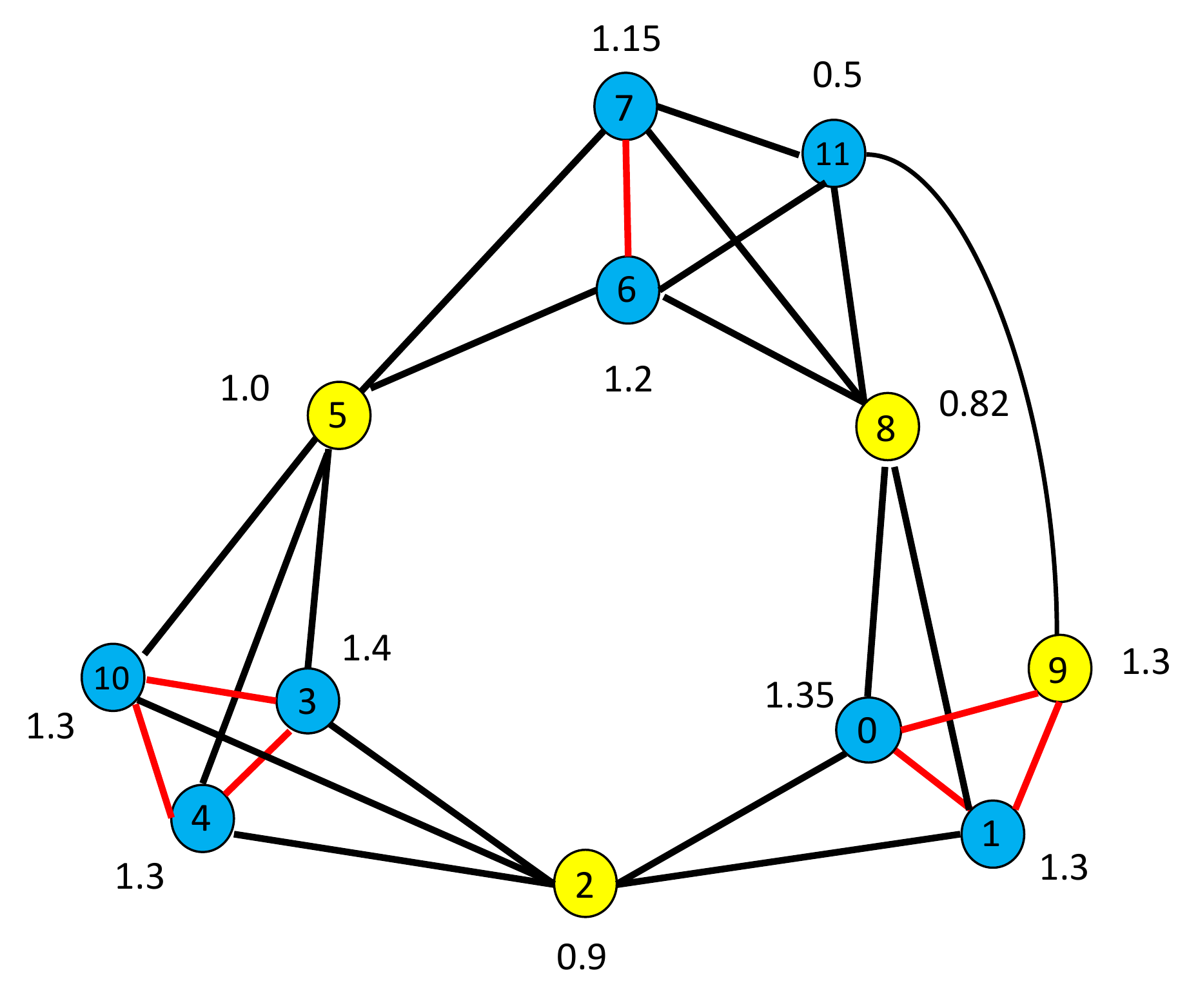} \\
    (a) & (b)
  \end{array}
  $$
  \caption{(a) An MWIS problem graph $G$ with 9 weighted vertices
    (edges are in black and red).
      The global minimum corresponds to the maximum independent set
      $\{2,5,8\}$ (in yellow)
with total weight of $4.02$.
$G$ has 8 local minima, $L=\{ \{v_1,v_2,v_3\}: v_1 \in \{0,1\}, v_2 \in
\{3,4\}, v_3 \in  \{6,7\} \}$,
with weights ranging from
$3.70$ to $3.95$.
The local-minima  subgraph $G|_{L}$ consists  of 3 disjoint cliques
(edges in red) : $\{\{0,1\},\{3,4\},\{6,7\}\}$.
This graph can be scaled to a graph of $3n$ vertices, where the
global minimum consists of $n$ yellow vertices, while there are $2^n$
local minima formed by $n$ independent-cliques.
(b) An  MIS graph $G'$ with 12 weighted vertices (with 3 extra
vertices from $G$ in (a)). The global
minimum is $\{2,5,8,9\}$ (in yellow). There are three independent-cliques that form the local minima, $L=\{ \{v_1,v_2,v_3\}: v_1 \in \{0,1,9\}, v_2 \in
\{3,4,10\}, v_3 \in  \{6,7\} \}$. The edges in $G'|_{L}$ are in red.
}
\label{fig:G1}
\end{figure}

We will compare the evolution of stoquastic $\Htr(G)$ with
the evolutions of
$\HD(\Jxx, G_{\ms{driver}},G)$ for various values of $\Jxx$ for
the weighted graph $G$ in Figure~\ref{fig:G1}.
We use three metrics to compare the evolutions of the different
algorithms: (1) gap-spectrum, shown in Figure~\ref{fig:mgs}; (2) AC-signature (total overlaps of the
two arms with the ground state and the first excited state wavefunctions), shown  in Figure~\ref{fig:Q9-LR}; (3) Signed
overlaps (of the lowest seven problem states with the ground state and the first excited state wavefunctions), shown in Figure \ref{fig:Q9-SignedOverlap}.

\begin{figure}[t]
  \centering
$$
  \begin{array}[h]{cc}
   
\includegraphics[width=0.38\textwidth]{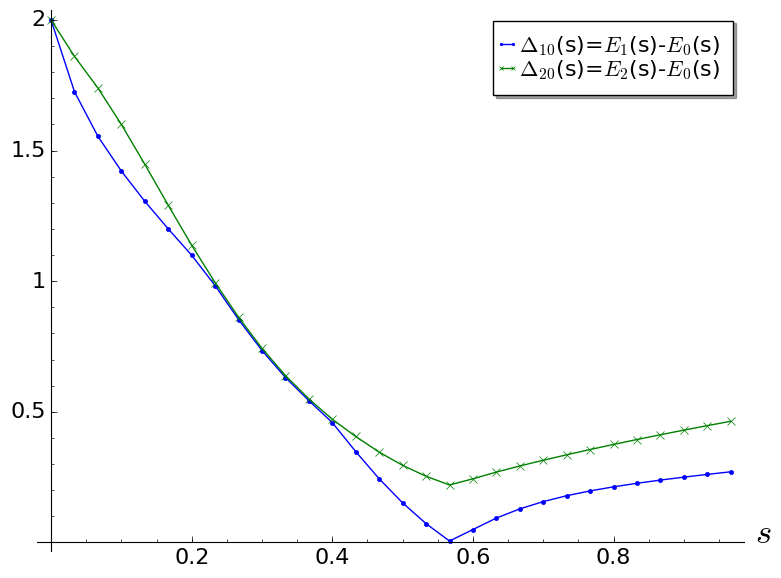} &
                                                      \includegraphics[width=0.38\textwidth]{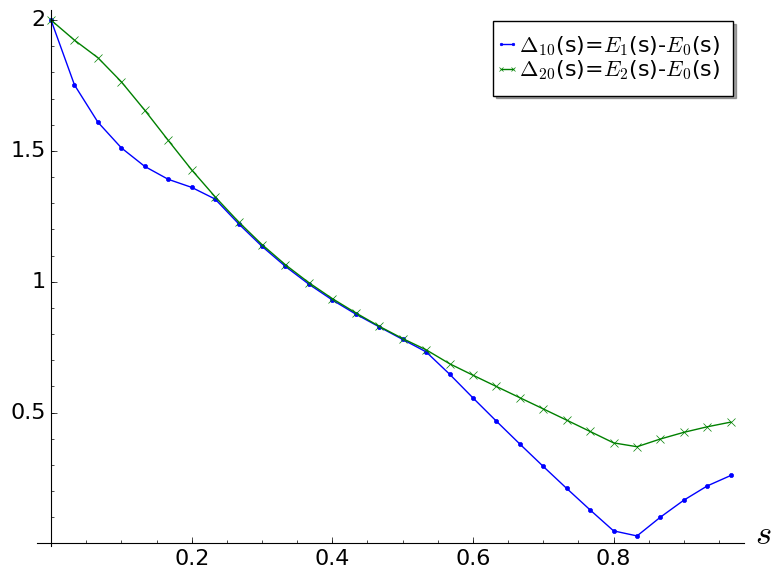}\\
    (a) \Htr(G) &  (b) \HD(\Jxx, G_{\ms{driver}},G),
                            \Jxx=-0.8, G_{\ms{driver}}=G|_{L} \\
    \hline\\
       \includegraphics[width=0.38\textwidth]{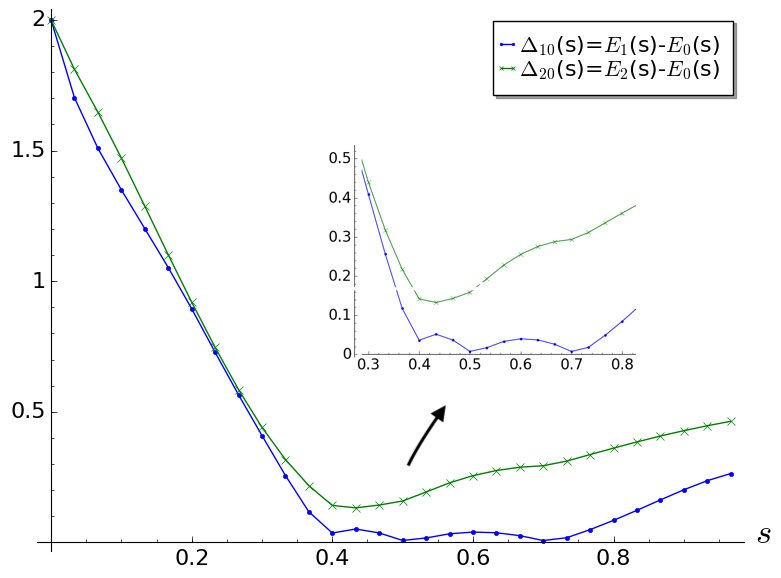} &
                                                                      \includegraphics[width=0.38\textwidth]{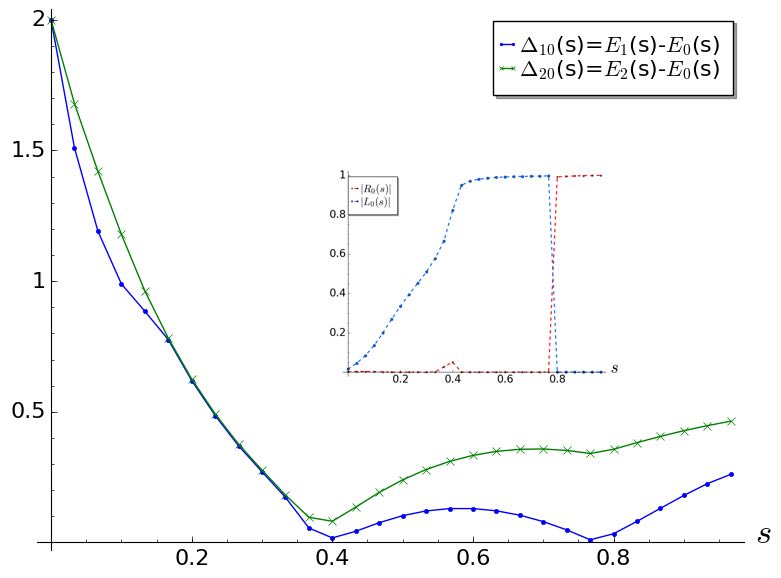}  \\
     (c) \HD(\Jxx, G_{\ms{driver}},G), \Jxx=+0.8, G_{\ms{driver}}=G|_{L} &
                                                       (d) \HD(\Jxx,
                                                                               G_{\ms{driver}},G), \Jxx=+0.8,
                                                                             G_{\ms{driver}}=G\\
    \hline\\
    \includegraphics[width=0.38\textwidth]{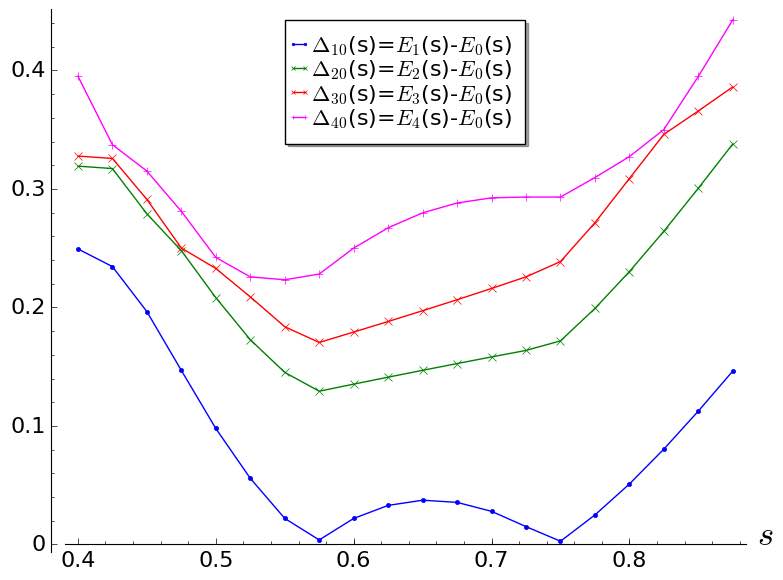} &
                                                                      \includegraphics[width=0.38\textwidth]{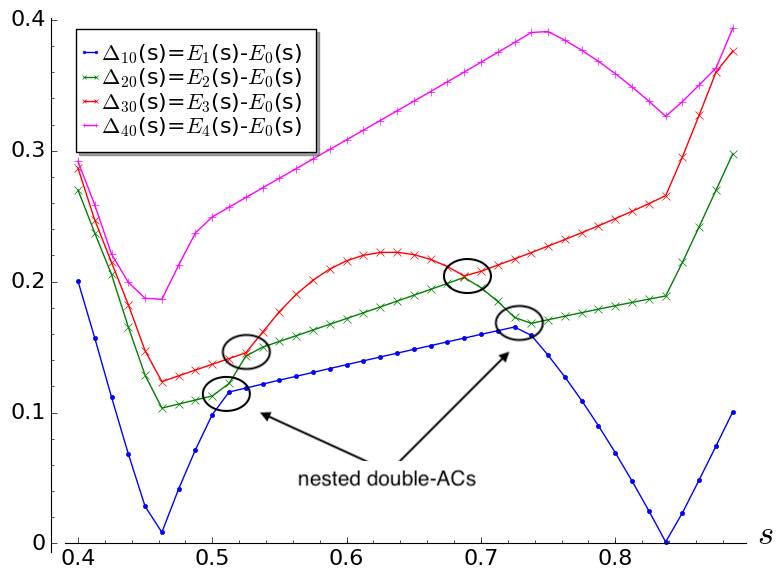}  \\
     (e) \HD(\Jxx, G_{\ms{driver}},G'), \Jxx=+0.8, G_{\ms{driver}}=G|_{L} &
                                                       (f) \HD(\Jxx,
                                                                               G_{\ms{driver}},G'), \Jxx=+1.1,
                                                                                                                  G_{\ms{driver}}=G|_{L}
    \\
    \hline
  \end{array}
$$
\caption{The  {\em gap-spectrum} comparison of $\Htr(G)$ with
  $\HD(\Jxx, G_{\ms{driver}},G)$ 
for 
the weighted graph $G$ shown in Figure~\ref{fig:G1}(a).
There is one local minimum in the gap-spectrum for the stoquastic
Hamiltonians in (a) and (b). There are three local minima in the
gap-spectrum of the proper
non-stoquastic Hamiltonian, where $G_{\ms{driver}}=G|_{L}$, in (c). The
last two local minima in (c) correspond to two
bridged anti-crossings (a double-AC) with a large second-level gap (in
green). In (d), where
$G_{\ms{driver}}=G$, there is no longer a double-AC, but one AC
(the first minimum corresponds to a non-AC local minimum).
In (e) and (f), the problem graph is $G'$  in Figure~\ref{fig:G1}(b). There is a
double-AC with a large second-level gap for $\Jxx=+0.8$ in (e); there
are three nested double-ACs when $\Jxx=+1.1$ in (f), with a large gap
from the 4th-level (in magenta).
}
  \label{fig:mgs}
\end{figure}

\begin{figure}[t]
  \centering
$$
  \begin{array}[h]{cc}
 L/R \mbox{ Overlap with } $\ket{E_{0}(s)}$ &
                                                                  L/R
                                              \mbox{ Overlap
                                                                  with }
                                                      $\ket{E_{1}(s)}$\\
   \includegraphics[width=0.4\textwidth]{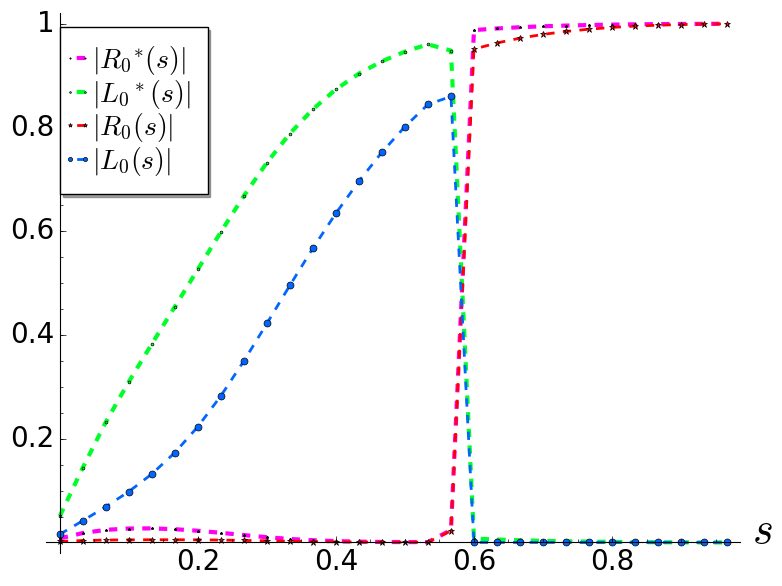} &
                                                                   \includegraphics[width=0.4\textwidth]{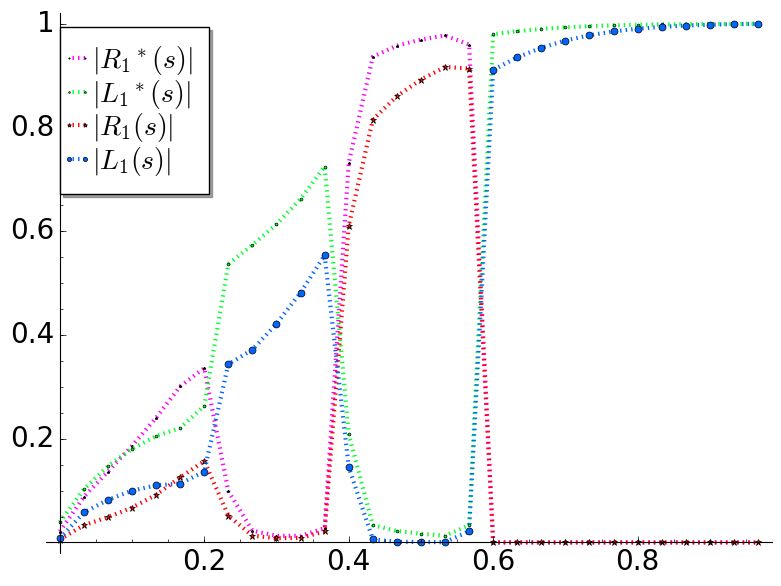}\\
    (a) \ham_{X}: \AC{L^*,R^*}  &\\
      \hline\\
  \includegraphics[width=0.4\textwidth]{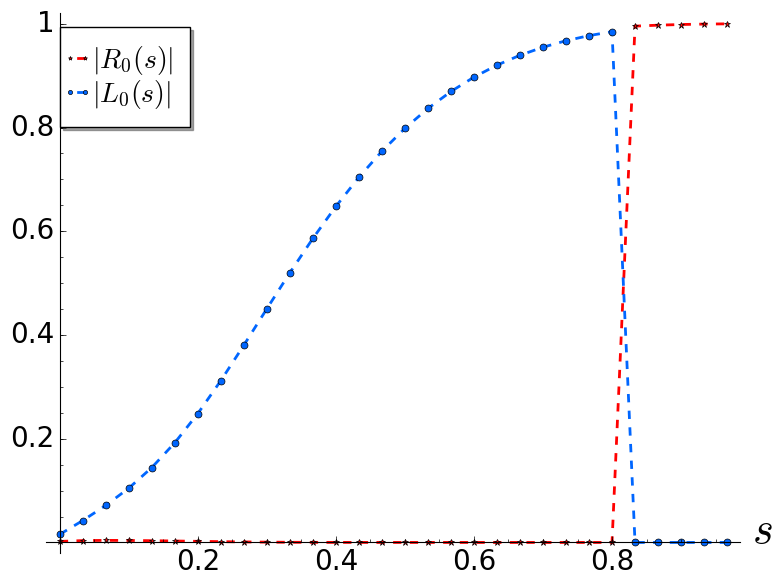} &
                                                                    \includegraphics[width=0.4\textwidth]{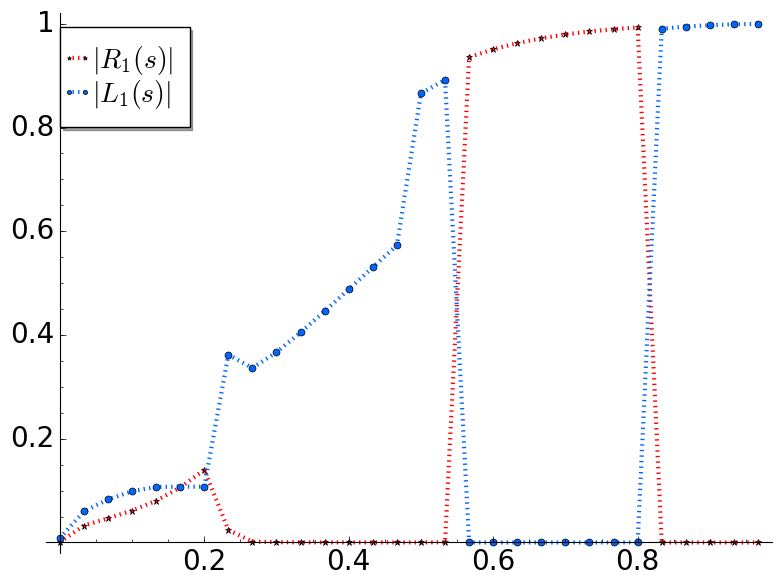}\\
    (b) \ham_{\XX}, \Jxx=-0.8 : \AC{L,R}  &\\
     \hline\\
  \includegraphics[width=0.4\textwidth]{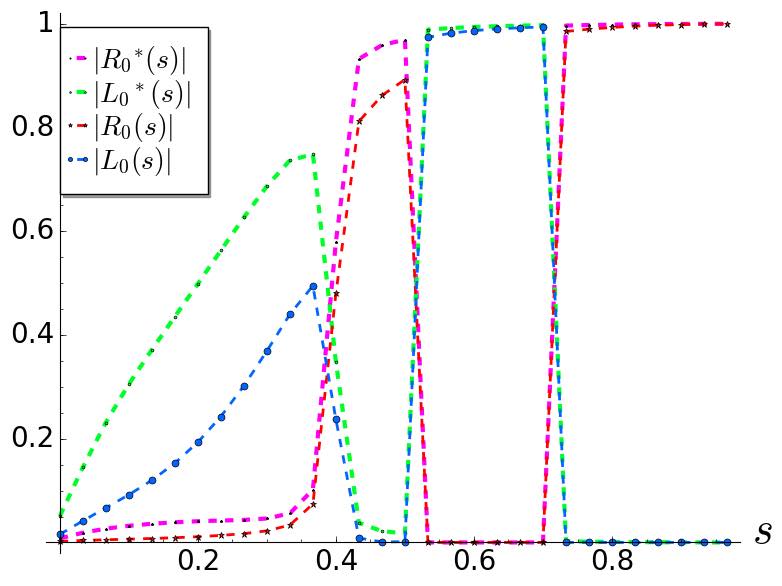} &
                                                                    \includegraphics[width=0.4\textwidth]{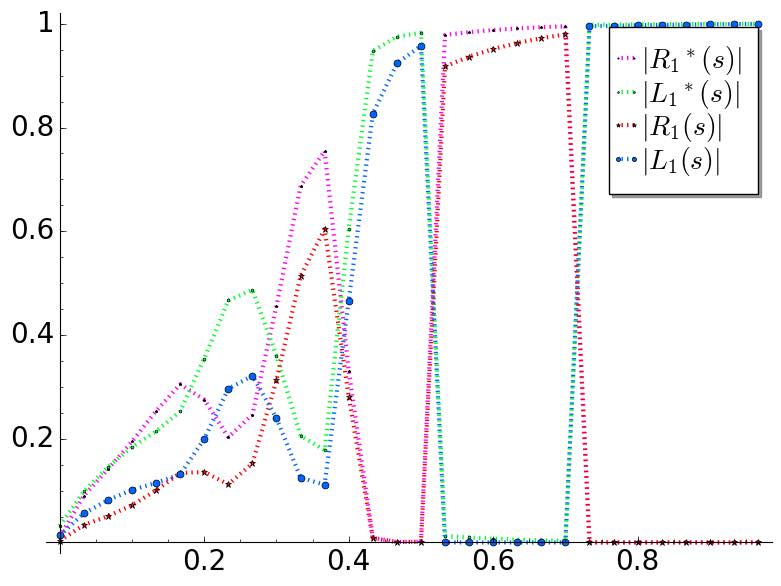}\\
    (c) \ham_{\XX}, \Jxx=+0.8: \mbox{two bridged anti-crossings}\\
                  \hline
  \end{array}
$$
\caption{The ``AC-signature'' comparison of $\Htr(G)$ with
  $\HD(\Jxx, G_{\ms{driver}},G)$ 
for
the weighted graph $G$ in Figure~\ref{fig:G1}(a).
The L/R overlaps with $\ket{E_{0}(s)}$ are shown in left, and with
  $\ket{E_{1}(s)}$ are shown in right,
  where $     L^* = L
   \union \nbr_{H_D}(L), R^* = R
   \union \nbr_{H_D}(R)$.
In both (b) and (c), $G_{\ms{driver}}=G|_{L}$.  In (c), there is a
double-AC : (1st AC) $\AC{R^*,L}$ at $\sim 0.5$ and (2nd AC)
$\AC{L,R}$ at $\sim0.7$.
 }
  \label{fig:Q9-LR}
\end{figure}

\begin{figure}[t]
  \centering
$$
  \begin{array}[h]{cc}
 \mbox{Signed overlap with }$\ket{E_{0}(s)}$ &
                                                                \mbox{Signed
                                        overlap
                                                                  with }
                                                      $\ket{E_{1}(s)}$\\
     \includegraphics[width=0.43\textwidth]{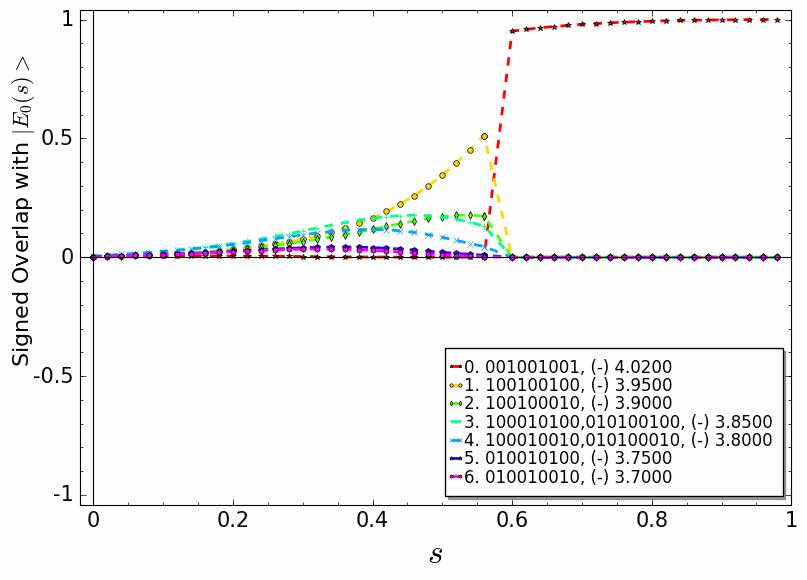} &
                                                          \includegraphics[width=0.43\textwidth]{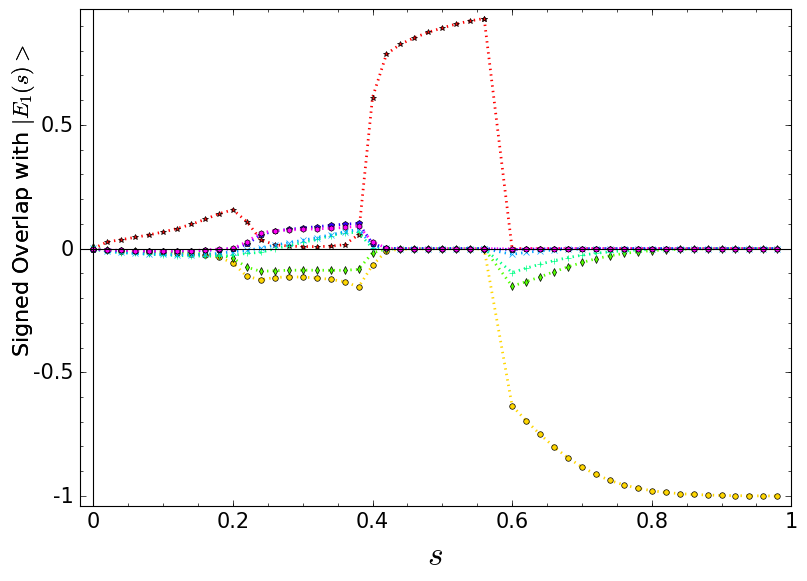}\\
    (a) \Htr(G): \AC{L,R}  &\\
    \hline\\
    \includegraphics[width=0.43\textwidth]{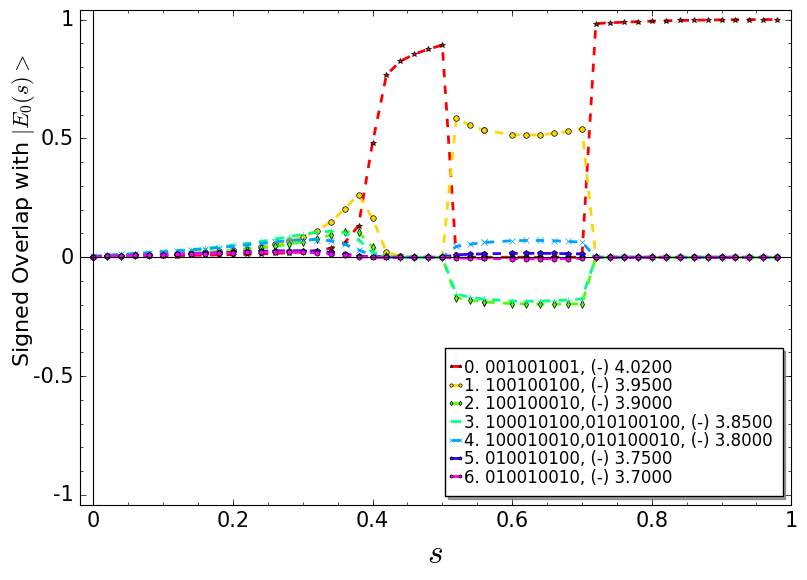} &
                                                                    \includegraphics[width=0.43\textwidth]{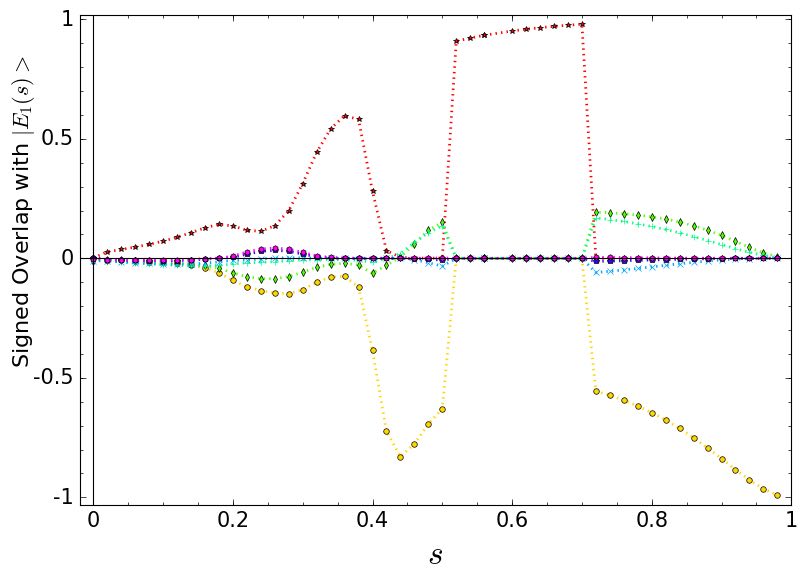}\\

\includegraphics[width=0.4\textwidth]{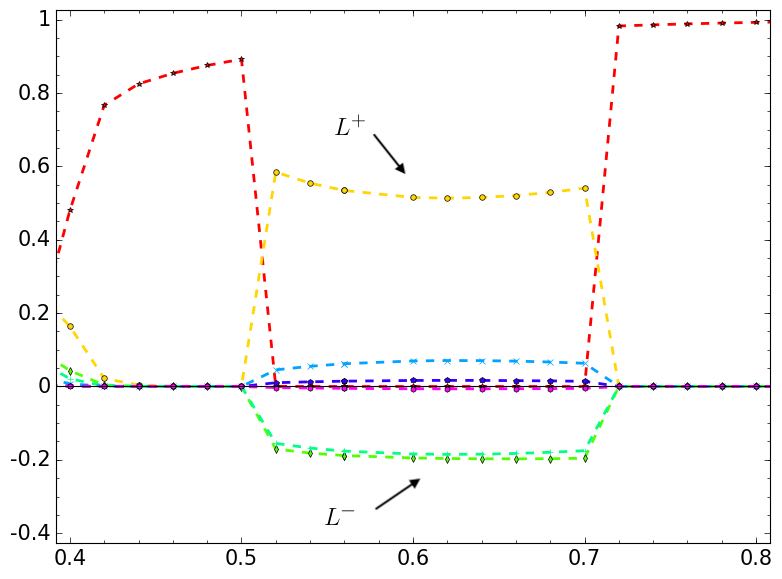}
                                             &\includegraphics[width=0.4\textwidth]{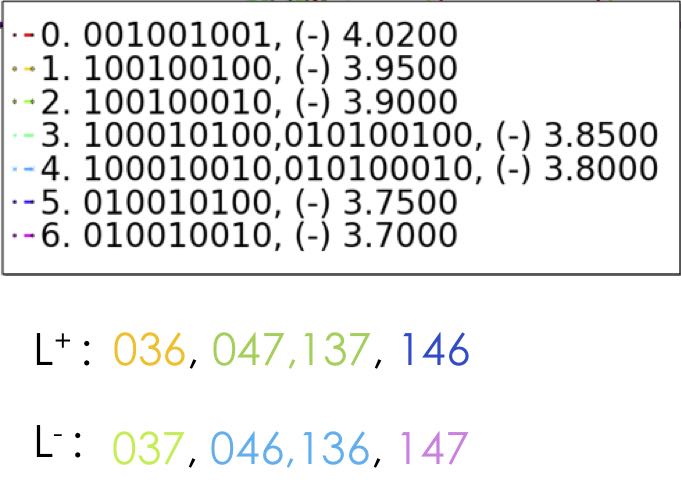}\\
    (b) \Jxx=+0.8: L \mbox{ is splited into } L^{+} \mbox{ and
    } L^{-}\\
    \hline\\
   \end{array}
$$
\caption{The  {\em signed-overlap} comparison of $\Htr(G)$ with
  $\HD(\Jxx, G_{\ms{driver}},G)$ 
for 
the weighted graph $G$ in Figure~\ref{fig:G1}(a).
  Signed Overlap of the seven lowest energy levels with
  $\ket{E_{0}(s)}$ (shown in left) and with 
  $\ket{E_{1}(s)}$ (shown in right).
  In (b), $G_{\ms{driver}}=G|_{L}$. $L$ is
  split into $L^+$ and $L^-$. There is exactly one XX-coupler between
  each state in $L^+$ and $L^-$. There is either zero or more than one
  XX-coupler within $L^+$ (or $L^-$).}
  \label{fig:Q9-SignedOverlap}
\end{figure}

\subsubsection{Local Minima Subgraph and Driver Graph $G_{\ms{driver}}$}

\paragraph{Local minima subgraph with a special independent-cliques (IC)
  structure.}
Let $L=\{l_1,l_2, \ldots, l_m\}$ be a set of local minima of the MWIS problem. 
Each set $l_i$ consists of a subset of
vertices in the problem graph, e.g. $l_1=\{0,3,7\}$ in Figure~\ref{fig:G1}, corresponding to
a weighted maximal independent set.
Any two {\em disjoint} local minima  in
$L$, say $l_i$ and $l_j$, e.g. $l_i=\{0,3,7\}$ and $l_j=\{1,4,6\}$ in Figure~\ref{fig:G1}, form a bipartite graph, where each vertex in
$l_i$  ($l_j$ resp.) must be adjacent to at least one vertex in $l_j$
($l_i$ resp.) because of the maximality of $l_i$ and $l_j$. The
bipartite graph formed by $l_i$ and $l_j$ can consist of several disconnect components,
each component being a connected bipartite subgraph (which is not necessarily
complete, i.e. not all edges between the two partites are present).
Inductively, any three disjoint local minima form a tripartite graph,
consisting of several connected components of tripartite subgraphs. 
A (connected) multi-partite graph can be considered as a generalized clique by
replacing each partite (an independent
set) with one super-vertex and the edges between the partites by one
edge between the two corresponding super-vertices.
It is in this sense we refer a multi-partite graph as a clique of {\em
  partites}. 
For example, a bipartite graph
is a clique of two partites; a tripartite graph
is a clique of three partites.
That is, we generalize a clique of vertices to be  a clique of partites where
each partite consisting of either one single vertex or an independent
set. When all partites are single vertices, the clique of partites is
the normal
clique or a complete subgraph; otherwise the clique is a multi-partite subgraph. 
(For simplicity, in this paper we only illustrate the normal
cliques.
The general case of the cliques of partites will be illustrated in our subsequent work.)
The size of the clique is the
number of partites in the clique.

We assume the subgraph formed by sets in $L$, denoted
by $G|_{L}\mdef G[\cup_{l \in L}l]$,
consists of a set of $\kappa$ vertex-disjoint components such that each
set in $L$ is formed from one element in each component.
By the maximality of the independent sets in $L$, each component is necessarily  a clique of {\em partites},  with
each partite consisting of either one single vertex or an independent
set.
Notice that these $\kappa$ disjoint cliques (of partites) in $G|_{L}$ can be
connected in the original graph $G$, e.g. two cliques are connected
through an edge in the $G$.
We further assume that the cliques (of partites) are {\em independent} in
that there are no edges between vertices from any two cliques.
That is, we assume that the local minima subgraph has a
 special {\em independent-cliques} (IC) structure in that 
the local minima $L$ is covered by a set of independent
  cliques of partites such that each local minimum in $L$ is formed
  by one partite from each clique in the IC.
If the weights within each clique is approximately the same, there
will be $\Pi_{i}^\kappa t_i$ many almost degenerate local minima, where
$t_i$ is the size of the $i$th clique, and thus  
an independent-cliques structure would cause a formation of an
anti-crossing for the TFQA algorithm, and would also pose a challenge for the
branch-and-bound based classical algorithms.

\begin{claim}
   Suppose that $\Htr(G)$ has an $\AC{L,R}$, where $L$
   consists of a set of local minima
and
  $R$ consists of the global minimum.
   of the MWIS problem on $G$.
  Furthermore, we assume that the local minima subgraph has an
  independent-cliques structure, i.e.  $G|_{L}$
  which consists of $\kappa$ cliques (of partities)  with size
  $t_i$.
  Let $G_{\ms{driver}}=G|_{L}$.
 We consider $\HD(\Jxx, G_{\ms{driver}},G)$ for different $\Jxx$
 while  $G_{\ms{driver}},G$ are fixed.  
  \begin{description}
  \item[(I)] There exists a $\Jsr>0$ such that
    for $\Jxx \in [0,\Jsr]$,
    $\HD(\Jxx, G_{\ms{driver}},G)$ is eventually stoquastic, and 
    $\ACgap(-\Jxx) < \ACgap(0) <\ACgap(+\Jxx)$, where $\ACgap(\Jxx)$
  denotes the anti-crossing gap of $\HD(\Jxx,
  G_{\ms{driver}},G)$.
Also, for $\Jxx \in [-\Jsr,\Jsr]$, $\ACgap(\Jxx)$ increases as $\Jxx$
increases. 
    \item[(II)]
There exist $\Jcr >0, \Jdr>0$ such that 
for $\Jxx \in (\Jcr,
\Jdr]$, $\HD(\Jxx, G_{\ms{driver}},G)$ is proper non-stoquastic
and it has a double-AC bridged
by $(L^{+},L^{-})$:  a $\AC{R,L}$ at $s_1(\Jxx)$ and an $\AC{L,R}$
at $s_2(\Jxx)$, where $L^{+}=\{l \in L: c_l(s)>0\} (\neq \emptyset)$ and $L^{-}=\{l \in
L: c_l(s)<0\} (\neq \emptyset)$ for $s \in [s_1(\Jxx),
s_2(\Jxx)]$.
Furthermore,  if $1\le t_i \le 2$, for all $i=1, \ldots, \kappa$, then
there is no AC within $[s_1(\Jxx), s_2(\Jxx)]$ between
the first and second energy level, and there is a $\Jxx$ such that
$\Delta_{21}(s)$ is large, for  $s \in [s_1(\Jxx), s_2(\Jxx)]$.
  \end{description}
\label{main-thm}
  \end{claim}

\begin{figure}[t]
  \centering
$$
  \begin{array}[h]{cc}
    \includegraphics[width=0.4\textwidth]{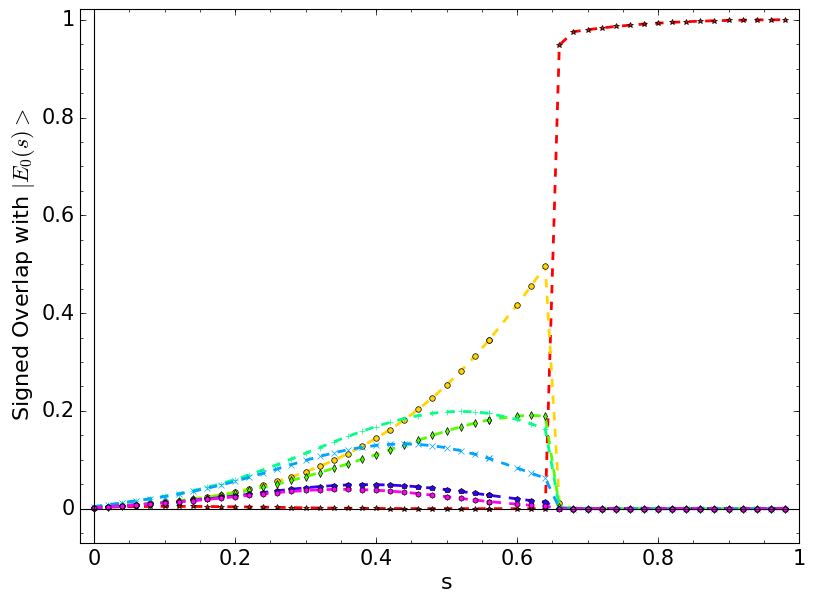} &
                                                              \includegraphics[width=0.4\textwidth]{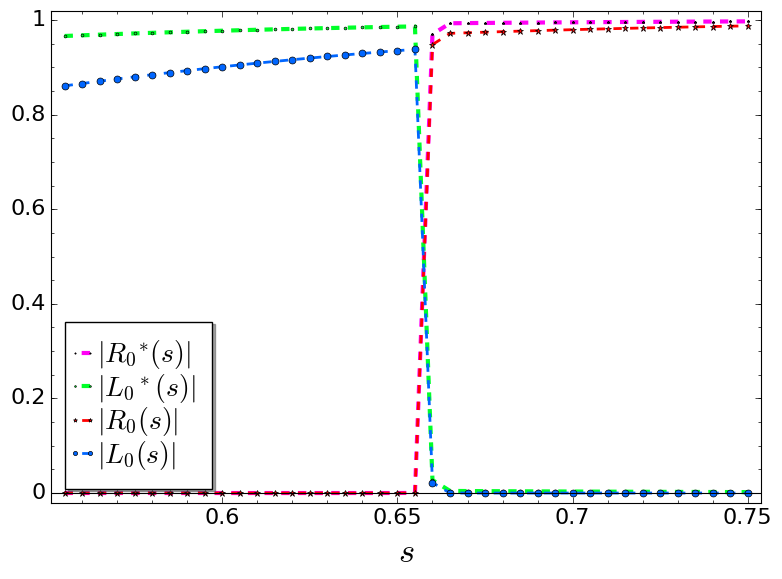} \\
    (a) \Jxx=-0.3: \Delta=2.15e-4, s^*=0.6596&\\
    \hline\\
         \includegraphics[width=0.4\textwidth]{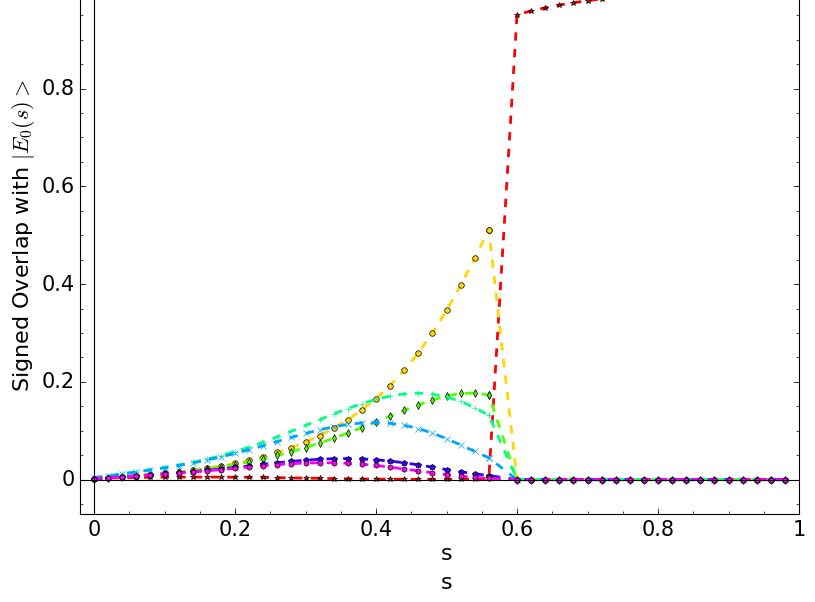}
                                                            &\includegraphics[width=0.4\textwidth]{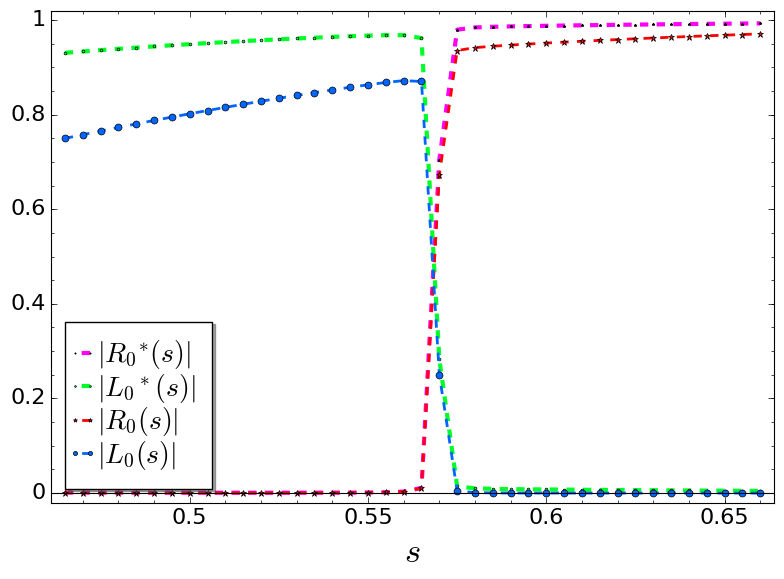}
    \\
    (b) \Jxx=0: \Delta=1.58e-3, s^*=0.5696 & \\
    \hline\\
    \includegraphics[width=0.4\textwidth]{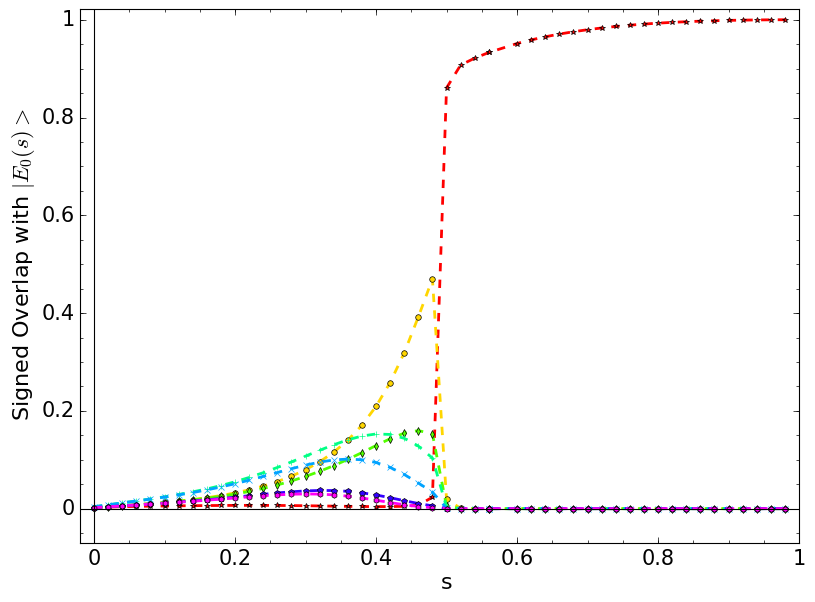}&\includegraphics[width=0.4\textwidth]{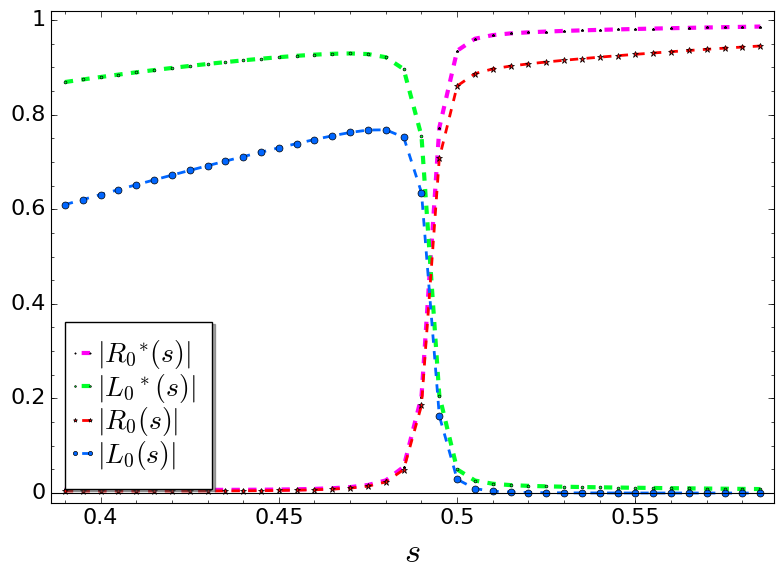} \\
    (c) \Jxx=+0.3: \Delta=7.00e-3, s^*=0.4928 & \\
   \hline
  \end{array}
$$
 
\caption{ Comparison of the overlaps of the lowest seven problem
  states (left) and L/R overlaps (right) with $\ket{E_0(s)}$ of $\HD(\Jxx, G_{\ms{driver}},G)$, for
    $\Jxx=-0.3, 0, +0.3$. As $\Jxx$ increases (within eventually
    stoquastic region), the anti-crossing point $\ap$ shifts to the left,
    the anti-crossing width $\delta$ increases, $|L_0(\apm)|$
    decreases, resulting in the increase of the
    \ACgap{} size.}
  \label{fig:AC-strength}
\end{figure}

\begin{figure}[t]
  \centering
$$
  \begin{array}[h]{cc}
   
\includegraphics[width=0.38\textwidth]{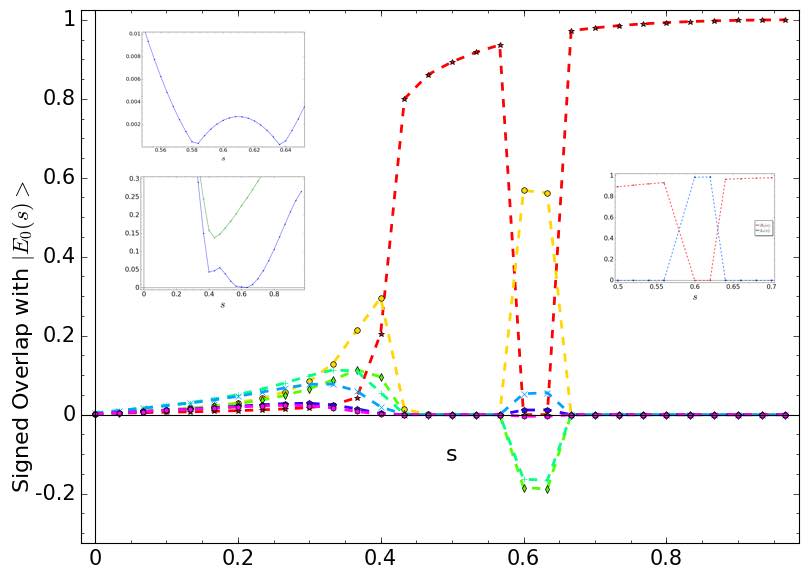} &  \includegraphics[width=0.38\textwidth]{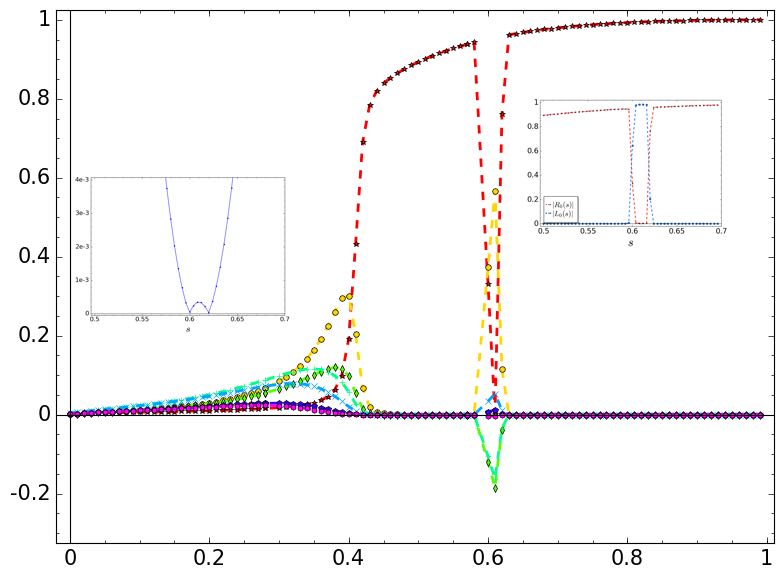} \\
    (a) \Jxx=+0.73&  (b) \Jxx=+0.725\\
    \hline\\
    \includegraphics[width=0.38\textwidth]{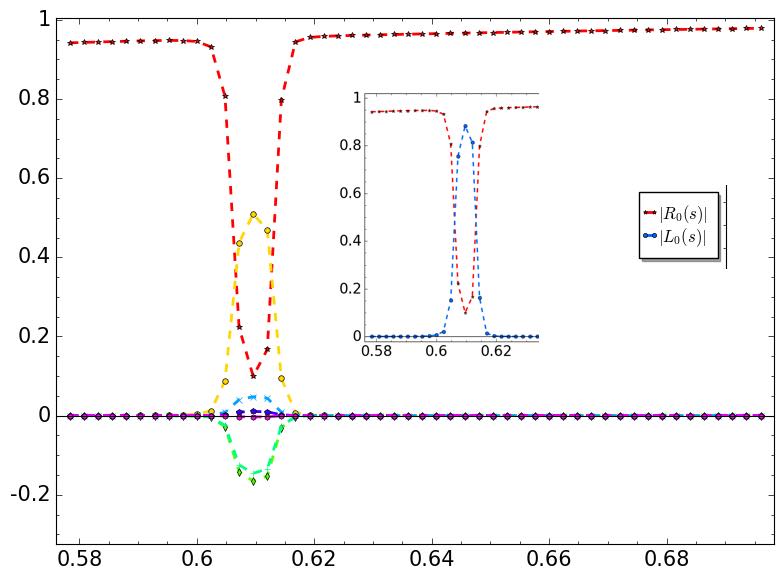} &  \includegraphics[width=0.38\textwidth]{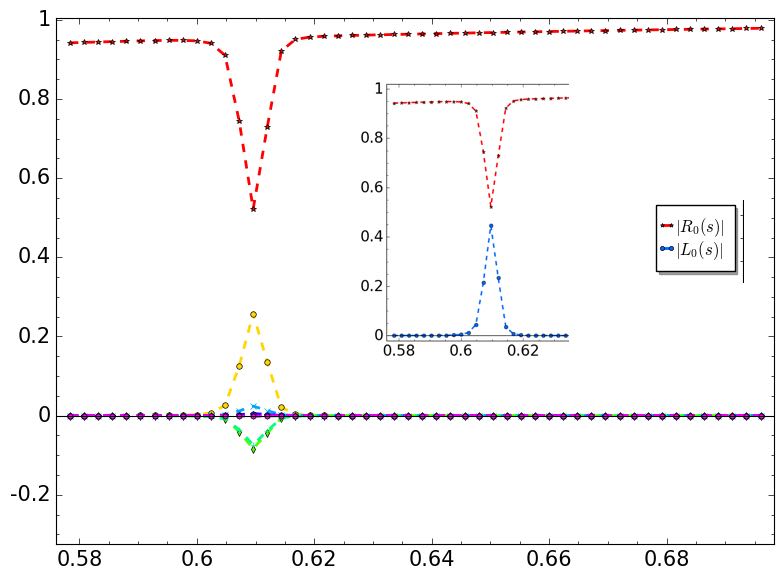} \\
    (c) \Jxx=+0.7248&  (d) \Jxx=+0.7247\\
    \includegraphics[width=0.38\textwidth]{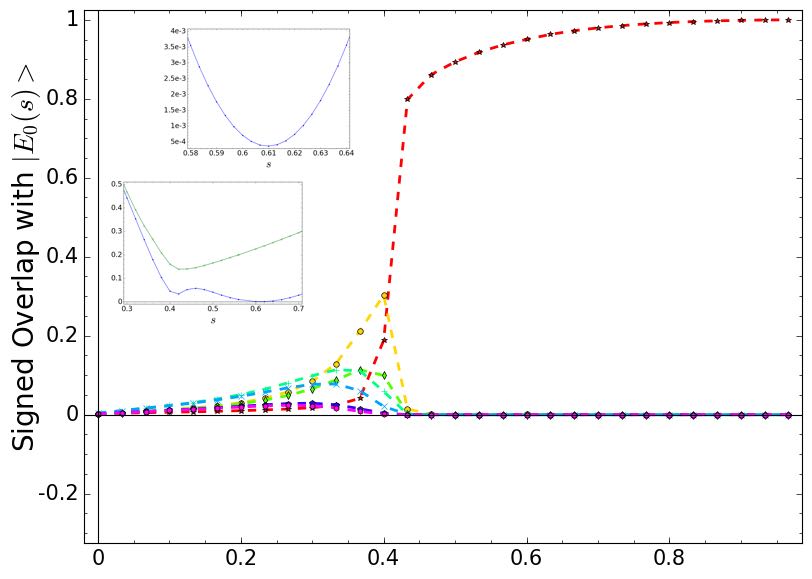} &
                                                          \includegraphics[width=0.38\textwidth]{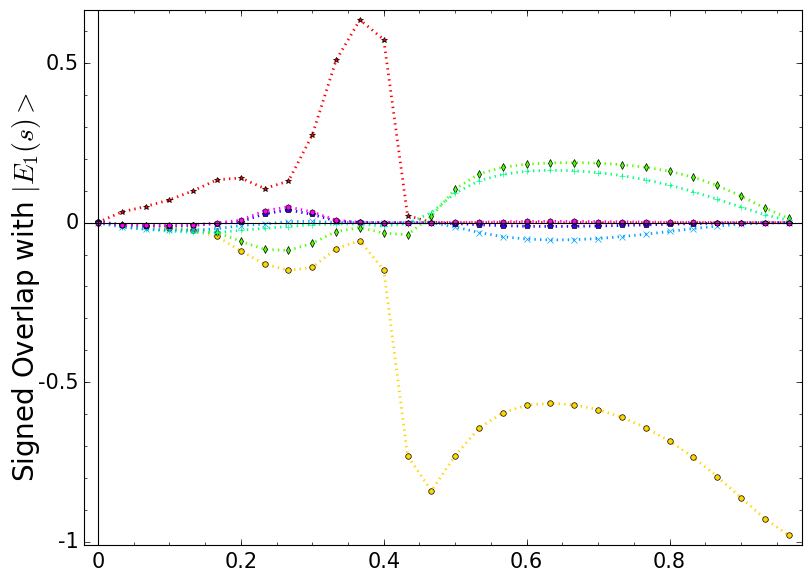}\\
    (e) \Jxx=+0.724 &  (f)
                                                         \Jxx=+0.724
                                                         (\mbox{Overlap
                                                         with
                      }\ket{E_1(s)})\\
    \hline
  \end{array}
$$
\caption{Evolution of the double-AC in $\HD(\Jxx,
  G_{\ms{driver}},G)$ for $\Jxx \in [\Jcr,\Jdr]$.
  As $\Jxx$ decreases 
  from $\Jdr$, the bridge between the two anti-crossings shrinks, as
  shown by comparing the results in 
  (a) and (b). As $\Jsp (\approx 0.725)$ approaches $\Jcr (\approx
  0.724)$, around $s \sim s_c$ (the splitting point), $|L_0(s)|$ decreases
  while $|R_0(s)|$ increases, as shown in (c) an (d).  Eventually $|L_0(s_c)|\approx 0$
  and $|L_1(s_c)|\approx 1$ as in (e) and (f), there is no longer an
AC (as one can see from the signed overlaps in (e) and (f)); the
min-gap is a non-AC local minimum in the gap spectrum.
Here we see that when two ACs are merged, the ground state actually
changes continuously  and rapidly by ``exchanging'' with
the first excited state not by an anti-crossing mechanism but by
merging of two bridged anti-crossings.
}
  \label{fig:Jmerge}
\end{figure}

\begin{figure}[h]
  \centering
 (a)  \includegraphics[width=0.85\textwidth]{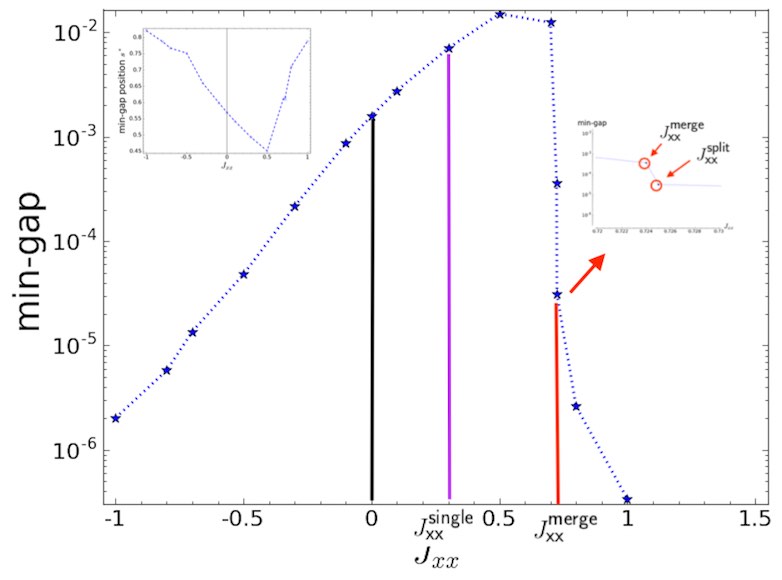}\\
  (b)  \includegraphics[width=0.6\textwidth]{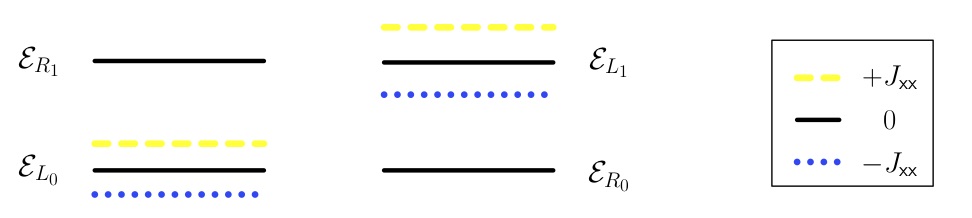}
  \caption{(a) Min-gap $\Delta$  vs \XX-coupler strength $\Jxx$ of $\HD(\Jxx,
  G_{\ms{driver}},G)$ for 
the weighted graph $G$ in Figure~\ref{fig:G1}(a).
    $\Jcr \approx 0.724, \Jsp \approx 0.725$.
    $\Jsr \approx 0.3$. 
    For $\Jxx \in (0,\Jsr]$, $\Delta(\Jxx) >\Delta(0)>\Delta(-\Jxx)$ (``de-signed'' is
    smaller);
    for $\Jxx > \Jcr$, $\Delta(\Jxx) <\Delta(-\Jxx)$ (``de-signed'' is
    greater), but they are actually asymptotically the same. In this example, for $\Jxx \in (\Jsr, \Jcr]$,
    $\Delta(\Jxx) > \Delta(-\Jxx)$ (``de-signed'' is
    smaller) but this is in general unknown as we do not know how
    small is the non-AC gap near $\Jcr$.  For $\Jxx \in (0.3,
    0.5]$, the min-gap is a weak-AC gap which is not necessarily large
    (in the problem size).
    (b) For a small $\Jxx$, $\EL(-\Jxx) < \EL(0) < \EL(+\Jxx)$ and 
    $\FL(-\Jxx) < \FL(0) < \FL(+\Jxx)$; while $\ER$ and $\FR$ stay the
    same because the \XX-couplers in the driver graph affect $L$
    only. Also, $\FL(-\Jxx) - \EL(-\Jxx) = \FL(+\Jxx) -
    \EL(+\Jxx)$ (c.f. Figure 1 in \cite{de-sign}).
    Both the anti-crossing point $\ap$ and the width $\delta$
    get shifted, result in either strengthened-AC ($-\Jxx$) or
    weakened-AC ($+\Jxx$) and
    $\ACgap(-\Jxx)<\ACgap(0) <\ACgap(+\Jxx)$. 
  }
  \label{fig:mg-Jxx}
\end{figure}

We justify our Claim by three main observations below, supported by the
numerical results of a scalable example. We leave a
rigorous proof as future work.
Nevertheless, our results already reveal the significance of the
driver graph played in the non-stoquastic QA, and the
essential difference of the $\ES{}$ and $\PNS{}$ due to different
strength of $\Jxx$. The evolutions of $\ES{}$ QA and $\PNS{}$ QA can be very
different. It is essential to distinguish them when considering
non-stoquastic Hamiltonians. 
For example, our result in (I), where the non-stoquastic is
$\ES{}$, 
already 
provides a counter-example\footnote{The driver graph actually 
  only needs to be a subgraph of $G|_L$, which can be efficiently
  identified.} to \cite{de-sign}, as shown in
Figure~\ref{fig:mg-Jxx}(a) (and in 
Figure~\ref{fig:mg-Jxx-Q12} for the graph $G'$ in the Appendix).
In fact, it is more of a counter-argument, illustrated in Figure~\ref{fig:mg-Jxx}(b),
to the intuition of Figure 1 in \cite{de-sign}.

For simplicity, since $G_{\ms{driver}},G$ are fixed,  we shall refer the system Hamiltonian
by $\HD(\Jxx)$ for each $\Jxx$.
Let ${\mathcal{E}}_{A_i}(\Jxx,s^*) \mdef \OL{\HD(\Jxx)(s^*)}{A_i,s^*} =
\bra{A_i(s^*)}\HD(\Jxx)(s^*)\ket{A_i(s^*)}$
  when (the eigenstate of  $\HD(\Jxx)(s^*)$) $\ket{E_i(s^*)} \simeq \ket{A_i(s^*)}$ for $i\in \{0,1\}$, $A \in \{L,R\},
  s^* \in \{\apm, \app\}$.
  We will consider the changes of the four energy values $\EL(\Jxx,\apm)$,
  $\FL(\Jxx,\app)$, $\FR(\Jxx,\apm)$, $\ER(\Jxx,\app)$ around the
  anti-crossing as $\Jxx$ changes.    
Since $\HD(0) = \Htr$ is stoquastic,  by continuity, there exists a $\Jsr >0$ such that
  $\HD(\Jxx)$ is eventually stoquastic, for $\Jxx \in (0,
  \Jsr]$. Thus for, $\Jxx \in [-\Jsr,
  \Jsr]$, the ground state coefficients $c_i(s) \ge 0$ for all $i$ and
  all $s \in [0,1]$.

  \paragraph{Observation 1.}
  For sufficiently small $\Jxx>0$, $\HD(\Jxx)$ can be seen as a
  perturbed Hamiltonian from $\HD(0)$,
where the perturbation $V =  (1-s)s \sum_{ij \in
  G_{\ms{driver}}}\sigma_i^x \sigma_j^x$. That is, $\HD(\Jxx) = \HD(0)
+ \Jxx V$. 
  The entire evolution of  $\HD(\Jxx)$ is a perturbed evolution of
  $\HD(0)$. In particular, the $\AC{L,R}$ of $\HD(0)$ at $\ap$
  also evolves.

  Observe that
  $\OL{V}{L_0,s} = \sum_{i \oplus j
    \in \EG} c_i(s)c_j(s) >0$ because $c_i(s)>0, c_j(s)>0$ and  $\OL{V}{L_1,s} = \sum_{i \oplus j
    \in \EG} d_i(s)d_j(s) >0$ because
  $\sign(d_i(s)d_j(s))=\sign(c_i(s)c_j(s))>0$; while $\OL{V}{R_i,s}=0$
  because there are
 no XX-couplers
 within $R$.
For $\Jxx>0$, $\EL(\Jxx,\apm) > \EL(0,\apm),
 \FL(\Jxx,\app)>\FL(0,\app)$ while $\ER(\Jxx,\app) \approx
 \ER(0,\app), \FR(\Jxx,\apm) \approx \FR(0,\apm)$.
 This results in a weaker $\AC{L,R}$ at $\ap' <
 \ap$ with a larger $\delta'>\delta$.

Since
  \begin{align*}
    \begin{cases}
      \OL{H_X}{\ket{E_0(\apm)}}\approx \OL{H_X}{L_0,\apm}
      \mdef -2\sum_{l \in L}\sum_{k \in \nbr_{H_X}(l)} c_l(\apm)c_k(\apm)\\
      \OL{H_{\XX}}{\ket{E_0(\apm)}} \approx \OL{H_{\XX}}{L_0,\apm}
      \mdef 2 \Jxx \sum_{l \in L}\sum_{k \in \nbr_{H_\XX}(l)} c_l(\apm)c_k(\apm),\\
        \end{cases}
  \end{align*}
 we have  $\OL{H_X}{\ket{E_0(\apm)}}<0$ and $\OL{H_{\XX}}{\ket{E_0(\apm)}}$
  is of the same sign of $\Jxx$. $\OL{H_{D}}{\ket{E_1(\apm)}} =
  \OL{H_{X}}{\ket{E_1(\apm)}}$ as they are no \XX-couplers within R. 
  Thus, as $\Jxx>0$ increases, the denominators in Eq.~(\ref{eq:coro}), $\OL{\delta H}{\ket{E_0(\apm)}} -\OL{\delta H}
      {\ket{E_1(\apm)}}$ and $\OL{\delta H}{\ket{E_1(\app)}} -\OL{\delta H}
      {\ket{E_0(\app)}}$, decrease, 
 while the numerator in Eq.~(\ref{eq:coro}), $\delta$,
      increases, we have  $\ACgap(\Jxx) >  \ACgap(0)$ by Corollary
      \ref{coro}.
      
Conversely, $-\Jxx<0$, $\EL(-\Jxx,\apm) < \EL(0,\apm),
 \FL(-\Jxx,\app)<\FL(0,\app)$, and it results in a stronger
 $\AC{L,R}$ at $\ap'' >
 \ap$ with a smaller $\delta'' < \delta$ such that $\ACgap(0)
  >\ACgap(-\Jxx)$.
That is, we have $\ACgap(-\Jxx)<\ACgap(0) <\ACgap(+\Jxx)$.
See Figure~\ref{fig:AC-strength} for an example.
By continuity, for $\Jxx \in [-\Jsr,
\Jsr]$, $\ACgap(\Jxx)$ increases as $\Jxx$ increases.
See Figure~\ref{fig:mg-Jxx} of a plot of $\ACgap{}$ vs $\Jxx$ for
$\HD(\Jxx, G_{\ms{driver}},G)$.
 
\paragraph{Observation 2.}
If there is a $\Jxx$ such that $L$ is an arm of an
anti-crossing in \PNS{} $\HD(G, \Jxx, G_{\ms{driver}})$, then $L$ will be
split into $L^+$ and $L^-$, with the coefficients for states in $L^+$
($L^-$ resp.)
being positive (negative resp.).
When $\ket{E_0(s)} \simeq \ket{L_0(s)}$, 
the ground state energy 
\begin{align}
E_0(s)
\approx (1-s)s \Jxx \sum_ { k\oplus k'\in \EG} c_k(s)
c_{k'}(s)
    + \sum_{k
  \in L^{+} \union L^{-}} |c_k(s)|^2 E_k.
\end{align}
Since $\Jxx>0$, the minimum of the (ground state) energy  is attained when $L$ is split into $L^+$
and $L^-$ such that $c_k(s) c_{k'}(s) <0$ for ${k
  \in L^{+}, k' \in L^{-}}$, and $\sum_ { k\oplus k'\in \EG} c_k(s)
c_{k'}(s)$ is minimized.
See Figure~\ref{fig:Q9-SignedOverlap} for an illustration.

Furthermore, if we assume that $t_i\le 2$, $L$ can be split into $L^+$
and $L^-$ such that the only \XX-neighboring is between  one state in $L^+$ and
another state in $L^-$. That is, there is no \XX-neighboring between states
within $L^+$ or states within
$L^-$.
By {\em \XX-neighboring} we mean that  there is exactly one XX-coupler between the two
states. 
Since $\ket{E_1(s)} \approx \ket{R_1(s)}$ for $s \in [s_1(\Jxx), s_2(\Jxx)]$, there
can not have an anti-crossing between the first level and the second
level during this interval.
However,
  this does not immediately imply that the second-level gap
  is large,
  which is required in order to apply DQA-GS successfully.
  The second-level gap can still be small if the second-level energy $E_2(s)$ is close to
  $E_1(s)$  for $s \in [s_1, s_2]$. However, it is
  likely that the second level energy $E_2(s)$ will be different for different
  $\Jxx$ while the first level energy $E_1(s)$ is independent with
  $\Jxx$. Therefore,  there is a $\Jxx
  \in (\Jcr, \Jdr]$ such that the second-level gap is large.

\paragraph{Observation 3.}
There
exist $\Jsp>0$ and $s_c$ such  that $\HD(\Jsp)$ has two anti-crossings bridged by
$(L^+, L^-)$, namely,  a $\AC{R,L}$ at $s_1= s_c-\varepsilon$
and an $\AC{L,R}$ at $s_2=s_c+\varepsilon$ for some
$\varepsilon>0$. The necessary conditions (in Eq.~(\ref{eq:17})) for
both ACs
are satisfied. 
By the condition that 
$\EL(\Jsp, s_c) \approx \ER(\Jsp,s_c) \approx \ER(0,\ap)$ and $s_c \approx \ap$, we obtain 
$\Jsp \approx \mu(\bar{E_L} - \bar{E_R})$
for some $\mu>0$, where $\bar{E_L}$ and $\bar{E_R}$ are the average energy (w.r.t problem
Hamiltonian) in $L$ and $R$.
  As $\Jsp$ decreases, there is a sharp exchange between the ground
  and first excited
  states and the two
anti-crossings are merged, resulting in no AC at $\Jcr=\Jsp-\epsilon$.
  See Figure~\ref{fig:Jmerge} for an illustration.
As $\Jxx$ increases from $\Jsp$, the bridge length $2\varepsilon$ increases,
the first AC weakens slightly, while the second AC strengthens
slightly. There exists a $\Jdr>\Jsp$, such that when $\Jxx>\Jdr$, the
fist AC is too weak to be qualified as an anti-crossing.

\paragraph{$L$ and $R$ share some common vertices.} In our example, for illustrative purpose, $L$ and $R$ are not only
disjoint subsets of $2^{[N]}$, but they are also disjoint in the
ground set ($[N]$). However, this is not a required condition, as shown in
$G'$ in Figure~\ref{fig:G1}(b) where the vertex $9$ appears both in
some local minima and the
global minimum. This is an important feature for otherwise the problem
would be solved efficiently classically by solving $G\setminus L$.

\paragraph{Nested double-ACs or a double multi-level AC.} 
When there is a $t_i>2$, we take {\em all} the edges in the independent-cliques
as the driver graph. However, in this case, there may be an AC between
the first excited state and higher and we no longer can guarantee the large
second-level gap. Instead,  it may have a sequence of nested double-ACs
(as shown in Figure~\ref{fig:mgs}(f)) or a double multi-level anti-crossing
where the lowest excited states 
 form a narrow band, as shown in Figure~\ref{fig:Q15AB}(a) in the Appendix for the 15-qubit instance in
  \cite{Amin-Choi}.
In this case, the system would undergo diabatic transitions
to the higher excited states through a sequence of the first ACs of
the nested double-ACs (or through the first AC of a double multi-level
AC),
and then returns to the
ground state through the sequence of second ACs of the nested
double-ACs (or through the second AC of a double multi-level AC).
  More details of the relationship between the driver graph structure and 
 the nested double-ACs (or double multi-level ACs ) will be reported in our subsequent work.

 \paragraph{Rules for constructing the driver graph.}
  In our above Claim, the driver graph $G_{\ms{driver}}$ is taken to
  be $G|_{L}$. Our arguments for the Observations actually reveal
  the intuition for constructing the driver graph. Namely, 
  the idea is to include the \XX-couplers between local minima in $L$
  (so to cause the split) 
  and avoid including \XX-couplers that are coupling states in $R$ and its
  neighbors (so as not to weaken $\ER$).
We impose the independent-cliques structure in the local minima for
the reason of justification, and also for the efficacy of identifying $G|_{L}$ with partial information of $L$, to be
elaborated in the next section. 
We note that  it is possible to
  relax the independent-cliques condition
to {\em almost-independent-cliques} by allowing some edges between these
cliques.

\subsubsection{A General Procedure without Prior Knowledge of the
  Problem Structure}
A natural question is: without knowing the problem structure,
can one design an appropriate driver graph efficiently?
As we note above, for the special \gic{} instances that we are considering,
knowing the set $L$ (without any knowledge of $R$) will be sufficient to construct the
appropriate driver graph.
The idea follows that if we can identify some elements in $L$ (which
is possible because these are the ``wrong answers'' from TFQA), it is possible to recover $G|_L$
explicitly and $L$ implicitly (as $L$ may consist of exponentially many maximal independent sets) by exploring
the graph locally with a classical procedure, because of the special independent-cliques structure.
More specifically, since we assume that the instance has an $\AC{L,R}$ when running with TFQA
algorithm, and if we further assume that the $\AC{L,R}$ is the only
super-polynomially small-gap in the system Hamiltonian,  then with
polynomial time, either by a stoquastic quantum annealer,  or the simulated quantum
annealing (SQA) \cite{SQA2,Hen-2015} ,
one can identify at least one local minimum $l_0$ in $L$, or  a subset of local minima $P
\subseteq L$ by repeating the procedure a polynomial number of
times. This is (Step 1.1) of the \DDD{} algorithm as described in Table \ref{tab:DDDalg}.


Next, we  
identify an IC  through $P$ (Step 1.2) as follows.
First, based on the induced subgraph from the vertices in $P$, we
identify a set of partial cliques of partites. (With some extra work, it is
possible to screen out noisy local minima that are in $P$ but do not
belong to the IC.)
Then we greedily extend the partial cliques to include more vertices
whose weights are similar to the weight of the partites. That is, we
include the similar-weight vertices which are adjacent to the existing
partites in the partial clique, and they are independent from other
partites.
Then we eliminate the vertices in the cliques found that
are adjacent to some vertices in another clique to make sure the
cliques are independent. 
This way we obtain a set of independent cliques of partites, IC, that
generates a set of local minima $L$ containing $P$.
The above procedure can be implemented in polynomial time. 
If  $G|_{L}$ has a ``clear cut boundary'' (i.e. no
ambiguous vertex which connects two cliques), we have
$G[\mbox{IC}] = G|_{L}$.
For example, the IC so discovered through a seed
$l_0=\{1,3,6\}$ in Figure~\ref{fig:G1}(b) will be the same as
$G'|_{L}$.
However, the IC may not be unambiguous.
Many different ICs are possible.
Different procedures (e.g. with
different vertex-weight allowances; or different vertex selection
criteria)  may be used to obtain different potential ICs.

To estimate a range of $\Jxx$ as in (Step 2.1),  we first obtain an upper
bound for the MWIS
(corresponds to $-\bar{E^R}$).  Such a good upper bound can be
obtained through a weighted clique cover as used in the
branch-and-bound algorithm for MWIS \cite{HCS2012,clique-cover,Balas-Xue}.
Together with the estimate bound for $\bar{E_L}$, we obtain a possible
range $(0,U]$ for $\Jxx$. We then in (Step 2.2) run the QA a polynomial number of
different $\Jxx$ evenly distributed in the range$(0,U]$, by annealing $\HD(\Jxx,
    G_{\ms{driver}},G)$  in  
    polynomial time and record the best answer.

\paragraph{Quantum Speedup.}
If the instance has the independent-cliques
structure such that $G[\mbox{IC}] = G|_{L}$, and there is a $\Jxx \in (0,U]$
 such that there is a double-AC in $\HD(\Jxx,
  G_{\ms{driver}},G)$ as described in Claim
  (II),
 the above algorithm \DDD{} will successfully find the ground state
through DQA-GS in polynomial time,
with an $O(c^\alpha)$ speedup over
the $\AC{L,R}$ plagued
stoquastic algorithm where $\alpha=\dist_{H_D}(L,R)$.
It may
have similar speedup for some other classical heuristics that
developed in comparison with QA algorithms (see
e.g. \cite{Hen-2015}).
While Claim (II) remains to be more rigorously proved, 
it is worthwhile to point out that the mechanism of  achieving speedup
here is
through a bridged double-AC which requires the necessity of +$\XX$-interactions,
so as to overcome the single-AC in the stoquastic case.
At the same time, this is achieved through a ``cancellation'' effect, and
not by removing the local minima from
the graph as in the classical algorithm.
This point perhaps would gain much appreciation
when comparing with the classical branch-and-bound solver for the MWIS where a clique
cover is pruned. From our perspective, this mechanism of achieving the
speedup  (by overcoming the local minima) does not seem to
have a similar classical counterpart, c.f. \cite{tunneling-MAL}.

\paragraph{Remark.} For $\Jxx \in (\Jsr, \Jcr)$,  there is no AC in $\HD(\Jxx,
  G_{\ms{driver}},G)$. However, the min-gap may still
be exponentially small. If the min-gap is polynomially large,
\DDD{} would also successfully solve the problem adiabatically in
polynomial time. It will be important to investigate the question
whether for this range of
$\Jxx$, the non-stoquastic but eventually stoquastic $\HD(\Jxx,
G_{\ms{driver}},G)$ is VGP or if it can be efficiently simulated by QMC algorithms.


\section{Discussion}
\label{sec:discussion}
In this paper, we point out the essential distinction of eventually
stoquastic (\ES{}) and proper-non-stoquastic 
(\PNS{}) for the non-stoquastic Hamiltonians from the algorithmic
perspective. The quantum evolutions of the \ES{} QA and \PNS{} QA
can be very different, especially when comparing their spectral gaps 
with their de-signed counterparts. Furthermore, we demonstrate the
essentiality of the design of the \XX-driver graph (which
specifies which \XX-couplers to be included) that has not been
considered by other previous work. 

In particular, we describe a proper-non-stoquastic QA
algorithm that can overcome the anti-crossing presented in the
transeverse-field QA algorithm
 and achieve a quantum
speedup for quantum optimization through DQA-GS, with the two essential ingredients:
non-stoquastic +\XX-couplers and the structure of the \XX-driver graph.
This is in contrast to the recent comments made in  \cite{CL2020} that the
``non-stoquasticity is desirable but not essential for quantum enhancement'',
where neither
the distinction of  non-stoquasticity  nor the structure of driver graph is taken
into consideration.
The essentiality of the \PNS{} +\XX-couplers   comes from the ability to
overcome the inevitable single-AC small gap in the TFQA by
``splitting'' the AC into two 
bridged ACs (a double-AC) that would then enable DQA-GS to solve the problem
efficiently.
The possibility of constructing the desired driver graph efficiently
is because one can make use of the ``wrong answers'' from the
anti-crossing plagued TFQA, together with the imposed special independent-cliques
condition.
The quantum speedup through the idea of double-AC-enabled-DQA was
first proposed in \cite{diabatic1}, where an {\em oracular} stoquastic QA algorithm for solving the
glued-tree problem in polynomial time is presented.
There the double-AC is formed by two identical anti-crossings (due to  the symmetric evolution at the
middle). 
In contrast, our double-AC is formed by two anti-symmetric
anti-crossings that share a common arm of the opposite sign as the bridge. As we argue
above, the proper-non-stoquastic interactions (which cause the negative
amplitudes in the ground state) are essential in our argument
to form the bridge.
Admittedly, our algorithmic procedure still requires a more rigorous
proof, nevertheless, we believe that we have provided enough arguments
and evidence
that are confirmed by (scalable) numerical examples.
In particular,
when the hard instances whose local minima subgraph is indeed the
same as the 
graph $G[\mbox{IC}]$ discovered by our algorithm \DDD{},
then there would be a double-AC or a sequenced of nested double-ACs or
a narrow band
that enable DQA to solve the problem successfully in polynomial
time.
This would
achieve  
an exponential speedup over the stoquastic algorithm, and possible
exponential speedup over the state-of-the-art
classical algorithms for such instances.

Furthermore, there are some reasons to further support our arguments:
(1) We make use of the exclusive quantum feature (negative amplitudes)
in our algorithm. 
(2) The special structure of the \gic{} instances are believed to one of the
obstacles for the efficient classical algorithms (heuristics or exact
solvers).
(3) Proper-non-stoquastic (with \XX-driver) may not be VGP\cite{vgp}, and thus
not QMC-simulable.

The insights of this work are obtained based on the novel characterizations of a
  modified and generalized parametrization definition of an
  anti-crossing in
  the context of QOA. We demonstrate how the parametrized AC can be
  used as a
 tool to open up the `black-box' of the QOA algorithm, and thus
 facilitate the analysis and design of adiabatic or diabatic QA
 algorithms.
In particular, we derive the necessary conditions for the formation of
an anti-crossing.
This provides us algorithmic insight into the relationship between an
anti-crossing
and the structure of the local and global minima of the problem.
In our subsequent work, we study with what extra conditions, the
necessary conditions are also sufficient. 


There are several future works to be followed:
(1) We are developing a framework to better understand
the role of the negative amplitudes play in the non-stoquastic
Hamiltonians.
(2) We construct  a scalable
family of \gic{} instances that are believed to be hard for the
classical algorithms and the TFQA algorithms.
We will also address the question of practical applications of \gic{}
instances, and the relaxation of the independent-cliques condition to the
almost-independent-cliques such that some edges between cliques are
allowed.
(3) We will investigate how to best obtain a subset of local minima to
build the \XX-driver graph. There are also several different ways to improve the design of
the \XX-driver graph. 
(4) We can consider the possible generalization of the AC
definition to include a multi-level anti-crossing such
that the anti-crossing is between one energy level and a band of 
of closely tie together energy
levels (a pseudo degenerate state) which can make the analysis of the
algorithm more robust.
(5) The problem of the {\em intertwined} {\bf sparse} universal hardware graph
and the minor-embedding problem \cite{minor1,minor2}, including both
problem graph
and driver graph, will also be addressed.

Finally, since proper-non-stoquastic Hamiltonians can be
simulated by other quantum models with polynomial overhead, it is our
hope that our method
provides a quantum algorithm design paradigm for solving the
optimization problems.

\section{Methods}
\subsection{Proof of Proposition \ref{prop1}}
\begin{proof}
  The derivation is based on the assumption of the
linear interpolation between $H_D$ (driver Hamiltonian)  and $H_P$
(problem Hamiltonian). That is,
$H(s) = (1-s)H_D + s H_P$.  
(The catalyst version will be approximated.) 
Note that we can write $H(s) = H(\ap) + (s-\ap) (H_P -H_D)$. 
The idea is based on the non-degenerate perturbation theory where the unperturbed Hamiltonian $H^{(0)}= H(\ap)$, and the perturbation
$\delta H (=\frac{\delta H}{\delta s}) =H_P-H_D$. We apply the
perturbation to 
$H(\ap+\lambda) = H(\ap) + \lambda \delta H$, 
where $|\lambda|$ is sufficiently small. (The standard non-degenerate perturbation
theory usually apply to the positive small $\lambda$. See
e.g. \cite{NotesbyZwiebach}.
Here we apply to
both positive and negative $\lambda$. For the small negative $\lambda$, one can
think of it as it equivalently applies the negative sign to
$\delta H$.)
From the perturbation theory (see, e.g. \cite{NotesbyZwiebach} Chapter
1 page 8), we have the states and energies for $H(\ap+\lambda) = H(\ap) +
\lambda \delta H$:

\begin{align}
  \label{eq:4}
\begin{cases}
  \ket{E_n(\ap+\lambda)}=\ket{E_n(\ap)} - \lambda \sum_{k \neq n}
  \frac{\delta H_{kn}(\ap)}{E_k(\ap) - E_n(\ap)} \ket{E_k(\ap)} +
  O(\lambda^2)\\
E_n(\ap + \lambda) = E_n(\ap) + \lambda \delta H_{nn}(\ap) - \lambda^2 \sum_{k \neq n}
  \frac{|\delta H_{kn}(\ap)|^2}{E_k(\ap) - E_n(\ap)} + O(\lambda^3) 
\end{cases}
\end{align}
where $\delta H_{mn}(\ap) \equiv \bra{E_m(\ap)} \delta H | E_n(\ap)
\rangle$.
By definition (condition (a)), $\Delta_{10}(\ap)=E_1(\ap) - E_0(\ap) << E_2(\ap) - E_0(\ap)=\Delta_{20}(\ap)$, and
$|\lambda|$ is sufficiently small, we apply the above formulae to
$n=0,1$ and obtain:
\begin{align}
  \label{eq:5}
  \begin{cases}
    \ket{E_0(\ap+\lambda)} \simeq \ket{E_0(\ap)} - \lambda 
  \frac{\delta H_{10}(\ap)}{\Delta_{10}(\ap)} \ket{E_1(\ap)} \\
\ket{E_1(\ap+\lambda)} \simeq \ket{E_1(\ap)}  + \lambda 
  \frac{\delta H_{10}(\ap)}{\Delta_{10}(\ap)} \ket{E_0(\ap)} \\
  \end{cases}
\end{align}
and
\begin{align}
  \label{eq:6}
  \begin{cases}
E_0(\ap + \lambda) \doteq E_0(\ap) + \lambda \delta H_{00}(\ap) - \lambda^2 
  \frac{|\delta H_{10}(\ap)|^2}{\Delta_{10}(\ap)} \\
  E_1(\ap + \lambda) \doteq E_1(\ap) + \lambda \delta H_{11}(\ap) + \lambda^2 
  \frac{|\delta H_{10}(\ap)|^2}{\Delta_{10}(\ap)} 
   \end{cases}
\end{align}
[where the error tolerance in $\doteq$ due to perturbation is $\epsilon_p << \Delta_{10}(\ap)/\Delta_{20}(\ap))$.]
Eq.\eqref{eq:6} thus gives rise to the two parabolas with $\beta_0=\delta
H_{00}(\ap)$ and $\beta_1  = \delta H_{11}(\ap)$ 
and $ \alpha =  \frac{|\delta H_{10}(\ap)|^2}{\Delta_{10}(\ap)}>0$.
Furthermore, $ \beta_0 \simeq \beta_1$ (by Property (C4)).
This proves Property (C1). 

From Eq.~\eqref{eq:5},  for all $k \in \widetilde{L} \union
\widetilde{R}$, we have
\begin{align}
  \label{eq:7}
  \begin{cases}
    c_k (\ap+\lambda) \doteq c_k(\ap)  -\lambda 
  \frac{\delta H_{10}(\ap)}{\Delta_{10}(\ap)}  d_k(\ap)\\
 d_k (\ap+\lambda) \doteq  d_k(\ap)  +\lambda 
  \frac{\delta H_{10}(\ap)}{\Delta_{10}(\ap)}  c_k(\ap)
  \end{cases}
\end{align}
for $\lambda \in [-\delta,\delta]$ (the error tolerance in $\doteq$ is $\epsilon_p$).

From condition (iv),  we have  $|L_0(s-\delta)| \approx 1-\gamma$
while $|L_0(s+\delta)| \approx 0$, we have  $|c_l(s)|
\downarrow$; and $|L_1(s+\delta)| \approx 1-\gamma$
while $|L_1(s-\delta)| \approx 0$, we have $|d_l(s)| \uparrow$, for
all $l \in L$.
Similarly, we have $|c_r(s)| \uparrow, |d_r(s)| \downarrow$, for all
$r \in R$.
This proves Property (C2). 

To show Property (C3), we distinct two possible cases: (1)
$\delta H_{10}(\ap)>0$ and (2) $\delta H_{10}(\ap)<0$.
Assume case (1), i.e. $\delta H_{10}(\ap)>0$.
For $l \in L$, from the above arguments for Property (C2), we $|c_l(s)| \downarrow, |d_r(s)|
\uparrow$, i.e. $|c_l(s)|$ is decreasing while $|d_r(s)|$ is increasing.
By Eq.~\eqref{eq:7} , $c_l(s)$ must be in the same
sign of $d_l(s)$ (i.e. $c_l(s) d_l(s) >0$) (otherwise both would be
increasing, a contradiction.)
Similarly, we can deduce that $c_r(s) d_r(s) <0 \mbox{ for } r \in
R$. 
That is, we have 
$$
\left\{
  \begin{array}{l}
    L_0 \mbox{ and } L_1 \mbox{ in the same sign: }
    \sign(c_l(s))\sign(d_l(s)) =+1 \mbox{ for all } l \in L\\
  R_0\mbox{ and } R_1 \mbox{ in the opposite sign: } \sign(c_r(s)) \sign(d_r(s)) =-1 \mbox{ for all } r \in R\\
  \end{array}
\right.
$$
where $s :\ap -\delta \leadsto \ap+\delta$.

Similarly, we can show that for case (2):
$$
\left\{
  \begin{array}{l}
     L_0 \mbox{ and } L_1 \mbox{ in the opposite sign: }
    \sign(c_l(s))\sign(d_l(s)) = -1 \mbox{ for all } l \in L\\
   R_0\mbox{ and } R_1 \mbox{ in the same sign: } \sign(c_r(s)) \sign(d_r(s)) =+1 \mbox{ for all } r \in R\\
  \end{array}
\right.
$$
where $s :\ap -\delta \leadsto \ap+\delta$.
This proves Property (C3).

To prove Property (C4), substitute $\lambda=-\delta, +\delta$ to
Eq.~\eqref{eq:7},
we get 
\begin{align*}
  \begin{cases}
    c_k (\ap-\delta) \doteq c_k(\ap)  +\delta 
  \frac{\delta H_{10}(\ap)}{\Delta_{10}(\ap)}  d_k(\ap)\\
 c_k (\ap+\delta) \doteq c_k(\ap)  -\delta 
  \frac{\delta H_{10}(\ap)}{\Delta_{10}(\ap)}  d_k(\ap)
  \end{cases}
\end{align*}
Summing up the above  two equations, we get $2 c_k(\ap) \doteq c_k (\ap-\delta) + c_k (\ap
+\delta)$. Similarly, $2 d_k(\ap) \doteq d_k (\ap-\delta) + d_k (\ap
+\delta)$.
By the full-exchange condition in Eq.(\ref{eq:13}), we have  $|c_k
(\ap-\delta)| \doteq  |d_k (\ap
+\delta)|$ and $|c_k (\ap+\delta)| \doteq |d_k (\ap -\delta)|$, consequently we have
$|c_k(\ap)| \doteq |d_k(\ap)|$.
This proves Property (C4).
\end{proof}

\paragraph{Remark on the proof.} This derivation is based on the linear
interpolation path. However, for the catalyst, we can similarly apply
the arguments to the
linear approximation near the anti-crossing point, with $\delta H =
H_P - \ap H_D$.

\subsection{Proof of the Necessary Conditions}
As in the proof of the properties of the anti-crossing,
we express the system Hamiltonian in a perturbed form (but
we are not applying the perturbation theory here; our arguments are
exact).
For any $s_0 \in [0,1]$, $0<\lambda (<1-s_0)$,
\begin{align}
  \label{eq:1}
H(s_0+\lambda) = H(s_0) + \lambda \delta H  
\end{align}
where  $\delta H= H_P - H_D$.
 Let $\OL{\delta\ham}{\ket{\Psi}} \mdef
\bra{\Psi}\delta\ham\ket{\Psi}$.

\begin{lemma}
  For any two levels $i,j \ge 0$,  $s_0 \in [0,1]$,  and $0<\lambda (<1-s_0)$,
  we have 
  \begin{align}
  \label{eq:3}
  E_i(s_0 +\lambda) - E_j(s_0) = \lambda \frac{\bra{E_j(s_0)} \delta
  H \ket{E_i(s_0 + \lambda)}}{
    \langle{E_j(s_0)}\ket{E_i(s_0+\lambda)}}
\end{align}
\label{lemma41}
\end{lemma}

\begin{proof}
  Since $H(s_0+\lambda)\ket{E_i(s_0+\lambda)} = E_i(s_0+\lambda)
\ket{E_i(s_0+\lambda)}$,
we have
\begin{align}
  \label{eq:5}
  (E_i(s_0+\lambda) - H(s_0)) \ket{E_i(s_0 +\lambda)} = \lambda \delta
  H \ket{E_i(s_0 +\lambda)}
\end{align}

Apply $\bra{E_j(s_0)}$ to the left of the above equation, we get Eq. (\ref{eq:3}).
\end{proof}

\begin{proof} (of Theorem \ref{thm:nece})
We apply Lemma \ref{lemma41} to an $\AC{L,R}$ at $\ap$ with width
$\delta$, by setting
$\{i,j\}=\{0,1\}$, and $s_0=\apm, s_0+\lambda =\app$. That is,
$\lambda=2\delta$. 
We get
\begin{align*}
  \begin{cases}
  E_1(\app) - E_0(\apm) = 2\delta \frac{\bra{E_0(\apm)} \delta H
  \ket{E_1(\app)}}{\langle{E_0(\apm)}\ket{E_1(\app)}}\\
E_0(\app) - E_1(\apm) = 2\delta \frac{\bra{E_1(\apm)} \delta H
  \ket{E_0(\app)}}{\langle{E_1(\apm)}\ket{E_0(\app)}}
  \end{cases}
\end{align*}

By the full-exchange condition in Eq. (\ref{eq:full-exchange}), we have 
\begin{align}
  \label{eq:2}
  \begin{cases}
  |{\langle{E_0(\apm)}\ket{E_1(\app)}} | \approx
  |{\langle{L_0(\apm)}\ket{L_1(\app)}} |  \approx 1\\
  |{\langle{E_1(\apm)}\ket{E_0(\app)}}| \approx 
 |{\langle{R_1(\apm)}\ket{R_0(\app)}}| 
\approx 1.
  \end{cases}
\end{align}

We thus have 
\begin{align}
   \label{eq:33}
  \begin{cases}
    E_1(\app) - E_0(\apm) \approx 2\delta \OL{\delta H}{\ket{E_0(\apm)}} \\
E_0(\app) - E_1(\apm) \approx 2\delta \OL{\delta H}{\ket{E_1(\apm)}}
  \end{cases}
\end{align}

By Property (C1), 
\begin{align}
  \label{eq:4}
  E_1(\app) - E_0(\app) =  E_1(\apm) -  E_0(\apm)
  \doteq 2 \alpha \delta^2
\end{align}
where 
  $\alpha =  \frac{|\bra{E_1(\ap)}\delta H\ket{E_0(\ap)}|^2}{\Delta_{10}(\ap)}>0$.
  Subtract the two equations in Eq.(\ref{eq:33}),
  we thus get
\begin{align}
   \label{eq:171}
   \begin{cases}
  \OL{\delta H}{\ket{E_0(\apm)}} -\OL{\delta H}
  {\ket{E_1(\apm)}}  \approx  2 \eta^2\frac{\delta }{\Delta_{01}(\ap)} \\
  \OL{\delta H}{\ket{E_1(\app)}} -\OL{\delta H}
  {\ket{E_0(\app)}} \approx 2 \eta^2\frac{\delta }{\Delta_{01}(\ap)} 
   \end{cases}
 \end{align}
 where $\eta =\bra{E_1(\ap)}\delta H\ket{E_0(\ap)}$.
 
    In particular, we have
 \begin{align}
   \label{eq:171}
   \begin{cases}
  \OL{\delta H}{\ket{E_0(\apm)}}  > \OL{\delta H}
  {\ket{E_1(\apm)}}  \\
  \OL{\delta H}{\ket{E_1(\app)}} > \OL{\delta H}
  {\ket{E_0(\app)}} 
   \end{cases}
 \end{align}
 \end{proof}

 \subsection{Proofs of AC-gap Bound}

 \begin{lemma}
  \label{lemma-gap}
      
 Assume $\alpha=\dist_{H_D}(L,R) >3.$
          \begin{align}
            \label{eq:1}
            \Delta_{10}(\ap)
                      =  \Theta\left(\left\vert\sum_{i \in g(L), j \in g(R), \bra{i}H_D\ket{j}
            \neq 0}
            c_i(\ap) c_j(\ap) \bra{i}H_D\ket{j}\right\vert\right)
\end{align}          
where  
$g(L) = \{ i \in n(L): \exists j \in n(R) \mbox{ s.t.  } \bra{i}H_D\ket{j}
            \neq 0 \}$,
$g(R) = \{j \in n(R): \exists i \in n(L) \mbox{ s.t. } \bra{i}H_D\ket{j}
            \neq 0 \}$.
\end{lemma}

\begin{proof}
By definition, $E_i(\ap)=\bra{E_i(\ap)} H(\ap)
\ket{E_i(\ap)}$ for $i=0,1$.
Thus,  we have
\begin{align}
  \label{eq:14}
  \Delta_{10}(\ap) =\bra{E_1(\ap)} H(\ap)
  \ket{E_1(\ap)} - \bra{E_0(\ap)} H(\ap) \ket{E_0(\ap)}
\end{align}

 We assume that the anti-crossing is $(1-\epsilon_v)$-full-exchange
 and it satisfies the approximate SAS property in (C4).

First, we show that if the SAS is exact, i.e.,
\begin{align*}
  (*)
  \begin{cases}
    \ket{E_0(\ap)} =\ket{\widetilde{L}(\ap)} \textcolor{red}{+} \ket{\widetilde{R}(\ap)}\\
           \ket{E_1(\ap)} = \ket{\widetilde{L}(\ap)} \textcolor{red}{-} \ket{\widetilde{R}(\ap)}
         \end{cases}
\end{align*}
the \ACgap{} is contributed from the
coefficients of negligible states. 
Substitute  (*) into Eq. (\ref{eq:14}), because of the cancellation of
the opposite terms, we have $ \Delta_{10}(\ap) = -4 \bra{\widetilde{L}(\ap)} H(\ap)
\ket{\widetilde{R}(\ap)}$.
By assumption that $\widetilde{L} \cap \widetilde{R} = \emptyset$, we have $\bra{\widetilde{L}(\ap)} H_P
\ket{\widetilde{R}(\ap)}=0$.
Thus, 
$\Delta_{10}(\ap) \doteq -4 (1-s)\bra{\widetilde{L}(\ap)} H_D
\ket{\widetilde{R}(\ap)}.$
If $\dist_{H_D}(L,R) > 3$, it implies $\bra{L} H_D
\ket{R} =0$, and thus we have $\Delta_{10}(\ap) \doteq -4 (1-s)\bra{n(L)(\ap)} H_D
\ket{n(R)(\ap)}$.

Let $g(L) = \{ i \in n(L): \exists j \in n(R) \mbox{ s.t.  } \bra{i}H_D\ket{j}
            \neq 0 \}$,
$g(R) = \{j \in n(R): \exists i \in n(L) \mbox{ s.t. } \bra{i}H_D\ket{j}
            \neq 0 \}$ be the set of states that contribute non-zero
            values to the gap.

That is, 
\begin{align*}
            \Delta_{10}(\ap)) & = \kappa \bra{g(L)(\ap)} H_D
          \ket{g(R)(\ap)}\\
          &= \kappa  \sum_{i \in g(L), j \in g(R), \bra{i}H_D\ket{j}
            \neq 0}
            c_i(\ap) c_j(\ap) \bra{i}H_D\ket{j}\\
            & = \kappa  \sum_{i \in g(L), j \in g(R), \bra{i}H_D\ket{j}
            \neq 0}
            d_i(\ap) d_j(\ap) \bra{i}H_D\ket{j}
\end{align*}    
where $\kappa =-4(1-s)$.

When the AC is $(1-\epsilon_v)$-full-exchange, the small error due to
the differences in $c,
d$ of the non-negligible states  is
absorbed by the above term. 
\end{proof}

\paragraph{Remark.} 
As the AC becomes weaker, the error term due to the difference in $c,
d$ of the non-negligible states can be large enough to dominate the
gap size, in this case
the min-gap will be 
$\omega(\zeta^\alpha)$ (strictly greater than the order of $\zeta^\alpha$).

The above lemma implies that the non-zero contributions to the gap come from the
(at least $(\alpha/2)!$ ) states in
$g(L)$ and $g(R)$ from the ground state . Notice that we express the  \ACgap{} 
only on the ground state (it can be also on the first excited state as
$c_k(\ap) \doteq d_k(\ap)$ )
instead of the difference between the two levels. 

Next we show that the coefficients $c_k$ decrease in geometric order
of its distance from $L,R$, and thus the dominated terms are those in
the path of the 
shortest distance between $L,R$.

\begin{proof} (of Theorem \ref{thm-gap})
First, assume the anti-crossing is near the end of
annealing (i.e. $\ap \approx 1$) (with $R$ as the ground state, and
$L$ as the first excited state) as in the perturbative crossing case.

Observe that:
For $k \in n(R)$: $c_k(\ap - \delta) \neq 0 \longrightarrow c_k(\ap +
\delta) =0$, we have $c_k(\ap) \doteq 1/2 c_k(\ap - \delta)$.

For $k \in n(L)$: we have $d_k(\ap - \delta) \neq 0 \longrightarrow
d_k(\ap + \delta)  =0$,
thus $c_k(\ap) \doteq d_k(\ap)  \doteq 1/2 d_k(\ap - \delta)$.

Both $c_k(\ap -\delta)$ and $d_k(\ap -\delta)$ can be computed from
the high-order corrections using
Brillouin-Wigner perturbation theory \cite{BWnotes}. 
Let $H(\lambda) = H_P + \lambda H_D$, where the
unperturbed problem Hamiltonian $H_P$, while the perturbation
Hamiltonian is $H_D$.
The high order corrections of the perturbed state is given by 
\begin{align*}
  \label{eq:3}
  \ket{E_n(\lambda)} = \ket{n} + \lambda \sum_{m_1 \neq n} \ket{m_1}
\frac{\bra{m_1}H_D\ket{n}}{E_n(\lambda) - E_{m_1}}
+ \lambda^2 \sum_{m_1 \neq n} \sum_{m_2 \neq n}\ket{m_1}
\frac{\bra{m_1}H_D\ket{m_2}\bra{m_2}H_D\ket{n} }  {(E_n(\lambda) -
  E_{m_1})(E_n(\lambda) - E_{m_2})}
+ \ldots\\
+ \lambda^k \sum_{m_1 \neq n} \sum_{m_2 \neq n}  \ldots \sum_{m_k \neq n} 
\ket{m_1}
\frac{\bra{m_1}H_D\ket{m_2} \bra{m_2}H_D\ket{m_3}  \ldots \bra{m_k}H_D\ket{n} }  {(E_n(\lambda) -
  E_{m_1})(E_n(\lambda) - E_{m_2})\ldots (E_n(\lambda) - E_{m_k})}
+ \ldots
\end{align*}
where $E_n(\lambda)$ in the denominator is the (unknown) energy of the $\ket{E_n(\lambda)}$.

We apply the above perturbation formula to the ground state ($R$) to compute
$c_k(\ap)$ for $k \in n(R)$, and to the first exited state ($L$) to
compute $c_k(\ap)$ for $k \in n(L)$.  [Using $\ER(\app)$ or
$\EL(\apm)$ for $E_n(\lambda)$.]

The coefficients decrease in an almost geometric order, that is, $c_k
\approx a^k$ with $0<a<1$, $d_k \approx b^k$ with $0 <b <1$. 
The larger $k$ ($H_D$-driver distance) the
smaller the coefficients.
Therefore it is the pairs with smallest distance $t$ (corresponding to
the $t!$ shortest paths)  between $L,R$ that dominate the
gap size in Eq. (\ref{eq:1}) of Lemma, where $t=\dist_{H_D}(L,R)$.

For the case that $\ap << 1$, we can however shift the
anti-crossing to near the end, using the scaling theorem in
\cite{Choi2020}.
\end{proof}
The above argument can be generalized (by re-deriving Brillouin-Wigner
perturbation theory) to include when $L$ consists of
  the almost degenerate first excited states, and/or $R$ consists of
  $\GS$ and its LENS, when the property that coefficients decrease in
geometric order of its distance still holds.

\appendix
\centerline{Appendix}
\section{Maximum-Weight Independent Set (MWIS) Problem}
The Maximum-Weight Independent Set (MWIS)
problem (optimization version) is defined as:

\smallskip
\hspace*{0cm}{\bf Input:} An undirected graph $G (=(\ver(G),\edge(G)))$, where each vertex $i \in \ver(G) = \{1, \ldots, n \}$ is weighted by a
positive rational number $w_i$

\hspace*{0cm}{\bf Output:} A subset $S \subseteq \ver(G)$ such that
$S$ is independent (i.e., for each $i,j \in S$, $i\neq j$, $ij
\not \in \edge(G)$) and the total
{\em weight} of $S$ ($=\sum_{i \in S}
w_i$) is maximized. 
Denote the optimal set by $\wmis(G)$.
\smallskip

We recall a 
quadratic binary optimization formulation (QUBO) of the problem.
More details can be found in \cite{minor1}.
\begin{theorem}[Theorem 5.1 in \cite{minor1}]
If $\lambda_{ij} \ge \min\{w_i,w_j\}$ for all $ij \in \edge(G)$, then the maximum
  value of
  \begin{equation}
\oy(x_1,\ldots, x_n) = \sum_{i \in \ver(G)}w_i x_i - \sum_{ij \in \edge(G)}
  \lambda_{ij}x_ix_j
\label{eq:Y}
  \end{equation}
is the total weight of the MIS. 
In particular if $\lambda_{ij} > \min\{w_i,w_j\}$ for all
      $ij \in \edge(G)$, then $\wmis(G) = \{i \in \ver(G) : x^*_i = 1\}$,
where $(x^*_1, \ldots, x^*_n) = \argmax_{(x_1, \ldots, x_n) \in \{0,1\}^n}
\oy(x_1, \ldots, x_n)$.
\label{thm:mis}
\end{theorem}

Here the function $\oy$ is called the pseudo-boolean function for MIS,
where the boolean variable $x_i \in \{0,1\}$, for $i=1,\ldots, n$.  The proof is quite
intuitive in the way that one can think of $\lambda_{ij}$ as the {\em energy
  penalty} when there is an edge $ij \in \edge(G)$.
In this formulation, we only require $\lambda_{ij} > \min\{w_i,w_j\}$, and thus there is freedom in
choosing this parameter.

\subsection*{MIS-Ising Hamiltonian}
By changing the variables ($x_i=\frac{1+s_i}{2}$ where $x_i \in
\{0,1\}, s_i \in \{-1,1\})$, it is easy to show that MIS is equivalent
to minimizing the following function, known as the {\em Ising energy function}:
\begin{eqnarray}
  \energy(s_1, \ldots, s_n) &=& \sum_{i \in \ver(G)} h_i s_i + \sum_{ij \in \edge(G)} J_{ij}s_is_j,
\end{eqnarray}
which is the 
eigenfunction of the following
 {\em Ising Hamiltonian}:
\begin{equation}
\ham_{\ms{Ising}} = \sum_{i \in \ver(G)} h_i \sigma^z_i + \sum_{ij \in \edge(G)} J_{ij}
\sigma^z_i \sigma^z_j
\label{eq:Ising}
\end{equation}
where $h_i = \sum_{j \in \nbr(i)}
  \lambda_{ij} - 2w_i$, (conversely $w_i = 1/2(\sum_{j \in \nbr(i)}
  J_{ij} - h_i)$), $J_{ij}=\lambda_{ij}$, $\nbr(i) =\{j: ij \in \edge(G)\}$,
for $i \in \ver(G)$.

For convenience, we will refer to a Hamiltonian
in such a form as an {\em MIS-Ising} Hamiltonian.

\section{Resolving the 15-qubit graph $G_{rm}$ in \cite{Amin-Choi}}

\begin{figure}[h]
  \centering
  $$
  \begin{array}[h]{cc}
  \includegraphics[width=0.4\textwidth]{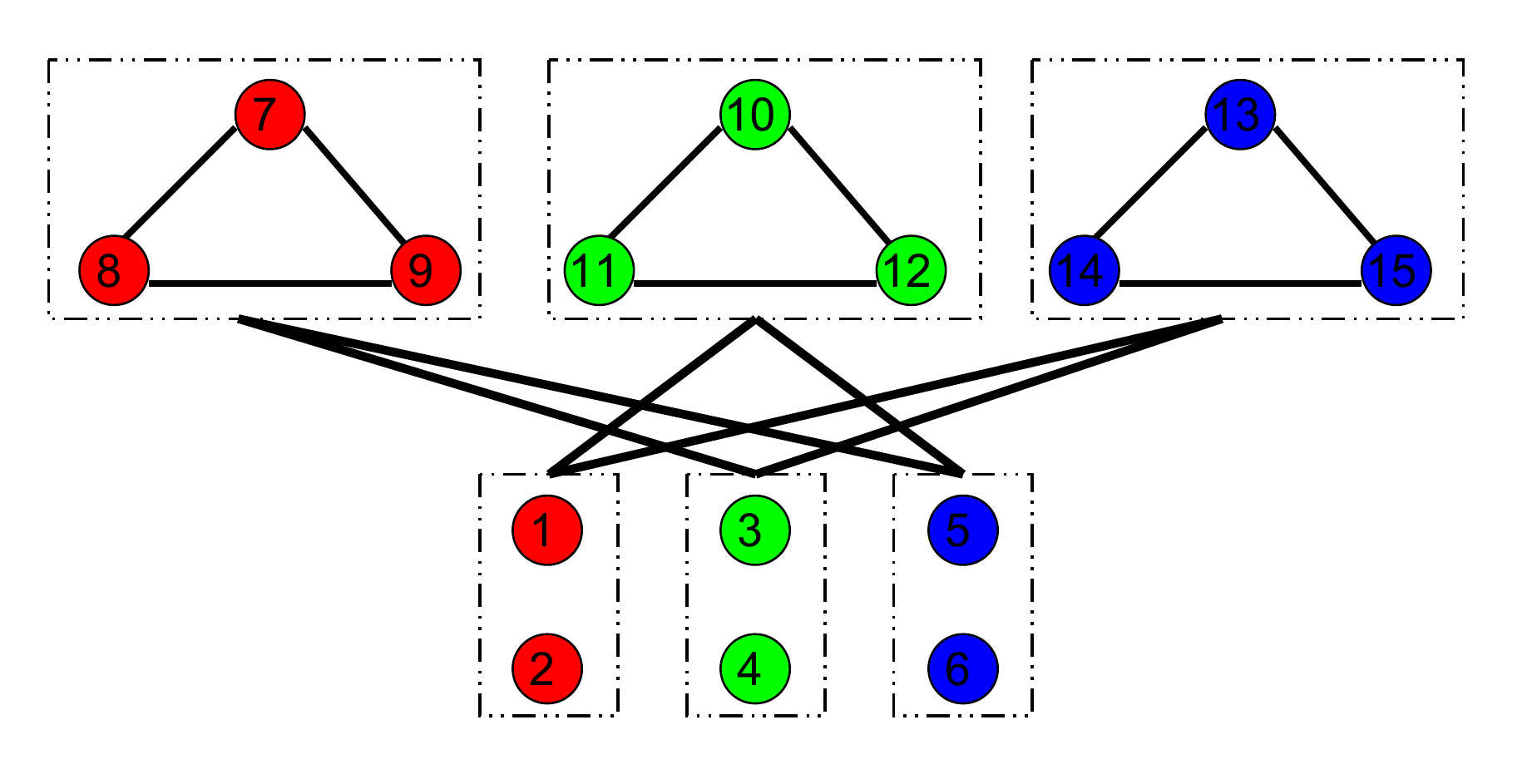}
&  \includegraphics[width=0.4\textwidth]{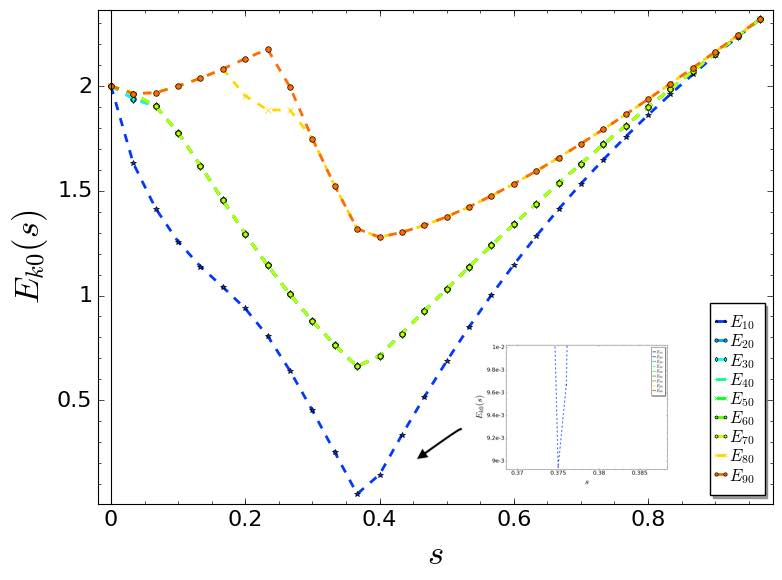}
    \\
    (a) & (b)
  \end{array}
  $$
    \caption{(a) The 15-qubit graph $G_{rm}$ from
      \cite{Amin-Choi}. Each vertex of 
    $1..6$ has a weight $W_G$, while each vertex of $7..15$ (in
    three triangles)  has a 
    weight $W_L$. For $W_L<2W_G$, the first six vertices make the
    global minimum (MWIS), while every combination of 3 vertices each
    from one triangle is a maximal independent set, altogether making
    27 degenerate local minima.
  (b) $\HD(\ham_X,G_{rm})$ has
an anti-crossing with a small gap ($\lesssim 9e-3$) at $\ap \approx 0.375$.}
\label{fig:Q15G}
\end{figure}

\begin{figure}[h]
  \centering
$$
  \begin{array}[h]{cc}
   \includegraphics[width=0.45\textwidth]{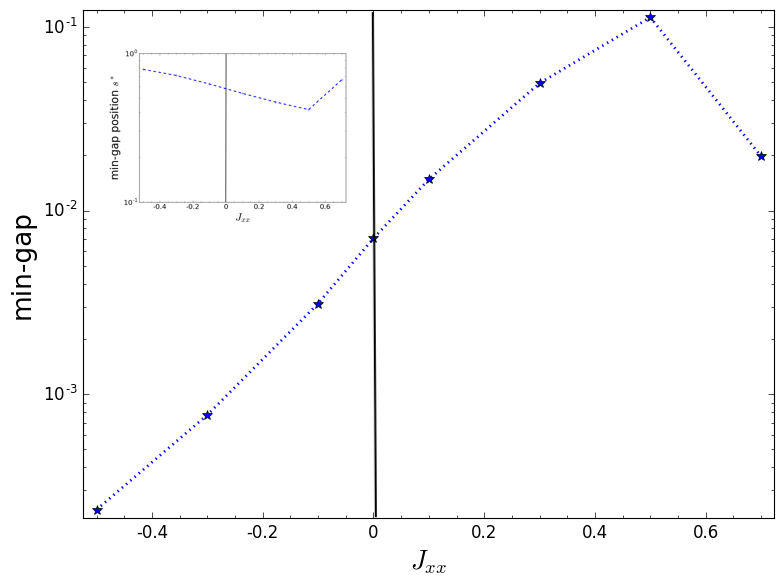} &
                                                                \includegraphics[width=0.45\textwidth]{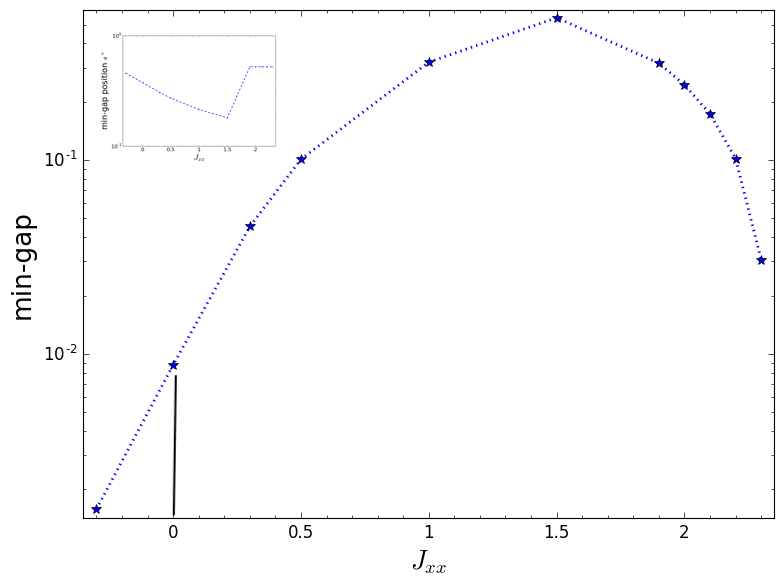}\\
    (a) G' \mbox{ in Figure~\ref{fig:G1}(b)} & (b) G_{rm} \mbox{ in Figure~\ref{fig:Q15G}(a)}
  \end{array}
  $$
    \caption{Min-gap $\Delta$  vs \XX-coupler strength $\Jxx$ of $\HD(\Jxx,
  G_{\ms{driver}},G)$ where (a) $G=G'$  in Figure~\ref{fig:G1}(b); (b)
  $G=G_{rm}$ in  Figure~\ref{fig:Q15G}(a).
(a)    For $\Jxx \in (0,0.5]$, $\Delta(\Jxx) >\Delta(0)>\Delta(-\Jxx)$ (``de-signed'' is
    smaller).
(b) $\Delta(+0.3) >\Delta(0)>\Delta(-0.3)$. For $\Jxx \in [0.5,2.2]$, there is no AC
and the min-gap is at least $0.1$.
  }
  \label{fig:mg-Jxx-Q12}
\end{figure}

\begin{figure}[h]
  \centering
  $$
  \begin{array}[h]{cc}
    \includegraphics[width=0.6\textwidth]{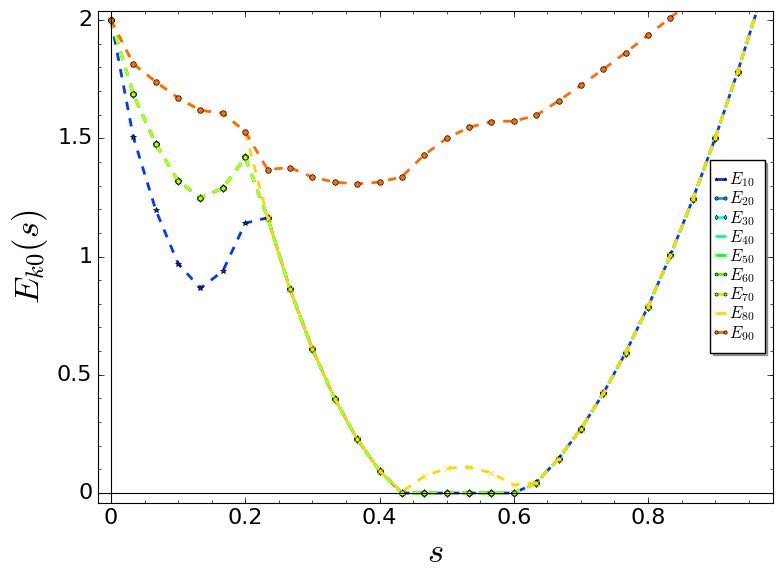} & \includegraphics[width=0.3\textwidth]{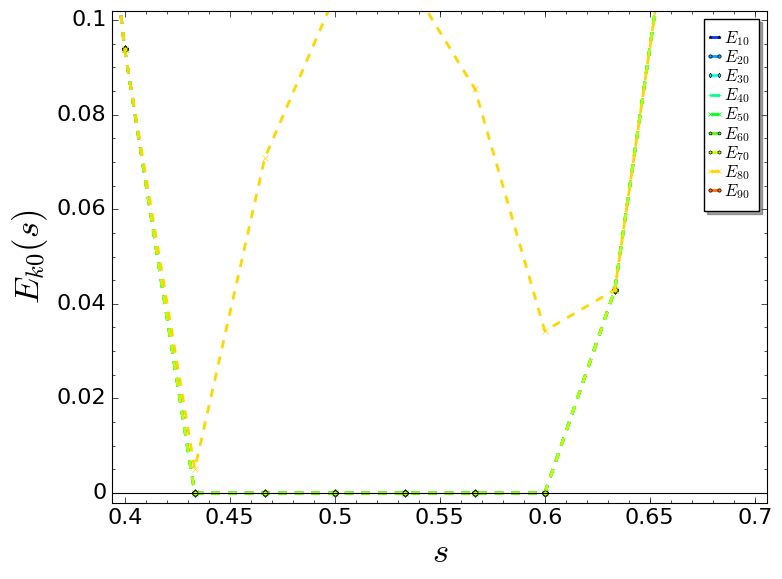}\\
                                                              
    (a) \HD(\Jxx, G_{\ms{driver}}, G_{rm}), \Jxx=2.5& (a2)
                                                                    s
                                                                    \in
                                                                    [s_1,s_2]:
                                                        E_{80}(s)>0.01\\
    
     \includegraphics[width=0.4\textwidth]{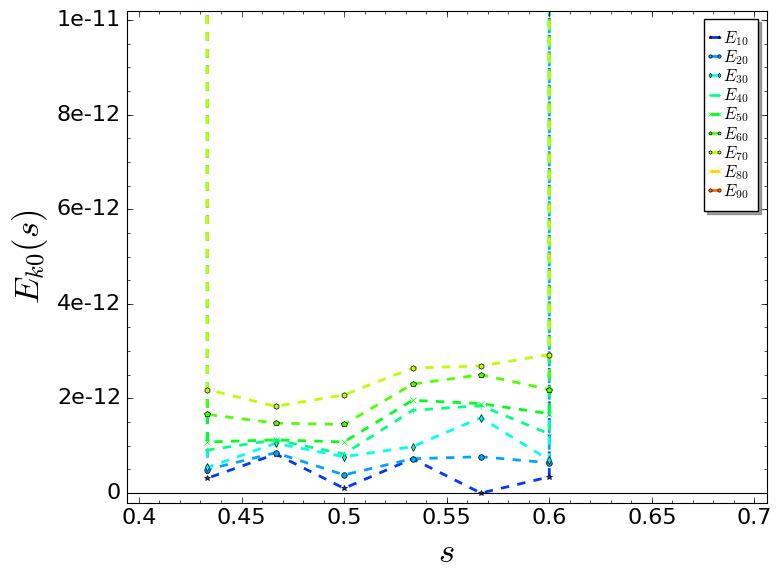} &\includegraphics[width=0.4\textwidth]{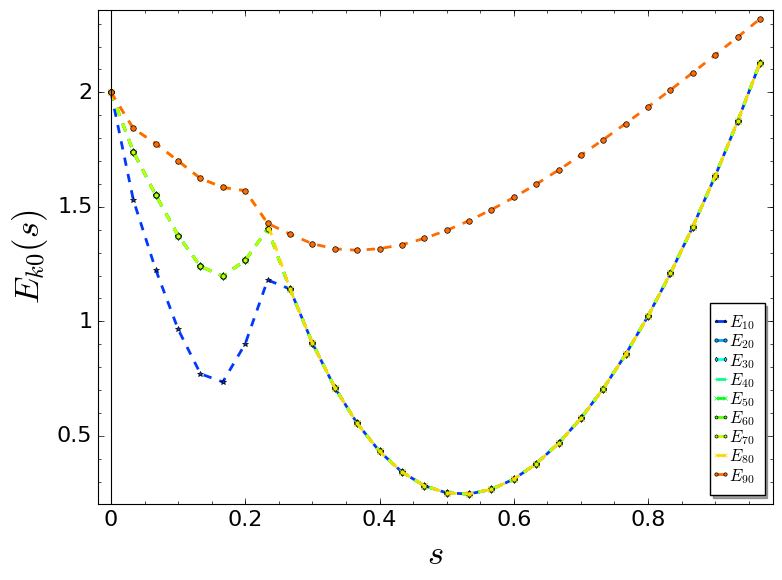} \\
                                                                     
                                                                      (a1)
    s \in [s_1,s_2]: E_{k0}(s) < 3e-12
                                                                  &
    (b) \HD(\Jxx, G_{\ms{driver}}, G_{rm}), \Jxx=2\\

  \end{array}
   $$
   \caption{The gap-spectrum for  $\HD(\Jxx, G_{\ms{driver}}, G_{rm})$ 
where
the weighted graph $G_{rm}$ shown in Figure~\ref{fig:Q15G}(a),
with $W_L=1.8$, $W_G=1$, $G_{\ms{driver}}$ consists of the three
triangles, (a) $\Jxx=2.5$; (b) $\Jxx=2.0$.
(a)
There is a double {\em multi-level} anti-crossing at $s_1=0.43$ and
$s_2=0.6$. For $s \in [s_1,s_2]$,
$E_{k0}(s) < 3e-12$ (shown in (a1)) for $k=1..7$, and
$E_{80}(s)>0.01$ (shown in (a2)).
The lowest 7 excited states ($\ket{E_1}\ldots \ket{E_7}$) form a
narrow band (as if it is a pseudo-degenerate state).
Similar to the diabatic cascade in \cite{tunneling-MAL},
the system can diabatically transition to $\ket{E_7(s)}$
     at $s_1$, and then through another diabatic transition
     back to $\ket{E_0(s)}$
     at $s_2$, when annealing time is short.
     One can further verify if indeed DQA-GS  can be successfully applied
     to this example through HOQST \cite{HOQST}.
     Remark: if one perturbs the vertex weights so as to break the
     $3^3$-fold degeneracy, one would obtain  a sequence of nested
     double-ACs for some $\Jxx$.
(b) For $\Jxx=2.0$, there is no AC  with
min-gap $=0.244$ at $0.523$. For $\Jxx \in [0.5,2.2]$, there is no AC
and the min-gap is at least $0.1$.}
\label{fig:Q15AB}
\end{figure}



The example graph  $G_{rm}$ from \cite{Amin-Choi} is shown in
Figure~\ref{fig:Q15G} (a). The stoquastic QA 
$\HD(\ham_X,G_{rm})$
has
an anti-crossing with a small gap ($\lesssim 9e-3$), as shown in
Figure~\ref{fig:Q15G} (b).
Taking the three triangles (three
independent cliques) as the driver graph $G_{\ms{driver}}$
$\HD(\Jxx, G_{\ms{driver}}, G_{rm})$ has no anti-crossing for $\Jxx
\in [0.5,2.2]$. In particular, for $\Jxx=2.0$, min-gap is greater than
$0.24$ at $0.523$, as shown in Figure~\ref{fig:Q15AB}(b).
For $\Jxx=2.5$, $\HD(\Jxx, G_{\ms{driver}}, G_{rm})$ has a double
multi-level anti-crossing, as shown in Figure~\ref{fig:Q15AB}(a).
DQA-GS can be applied to this example by a diabatic cascade at the
first AC and then return to ground state through another diabatic
cascade at the second AC. This
example also illustrates the difference from DQA in \cite{CL2020}
where the system remains in the subspace and does not necessarily
return to the ground state.

\section*{Acknowledgments}
I would like to thank Jamie Kerman for introducing to me the XX-driver graph
problem which directly rekindle this research, and his continuing
support and collaboration of this project. 
Special thanks to Itay
Hen for the very helpful discussion and comments and his help.
I would like to thank Daniel Lidar for his comments, especially
about diabatic cascades, and for the opportunity to
participate in DARPA-QEO and DARPA-QAFS programs.
I would also like to thank Tameem Albash and Elizabeth Crosson for the immediate responses
to the first draft of this paper.
Thanks also go to Siyuan Han and Federico Spedalieri for many comments
and discussions.
I gratefully acknowledge the help and comments from Sergio Boixo, Vadim Smelyansky,
and especially the advice and discussions with Eddie Farhi. 
 \end{document}